\newif\ifllncs  
\newif\iffull   

\fulltrue

\def\shownotes{1}

\ifnum\shownotes=1
    \newcommand{\authnote}[2]{\textcolor{red}{\textsf{#1 }\textcolor{blue}{ #2}}}
\else
    \newcommand{\authnote}[2]{}
\fi
\newcommand{\yilei}[1]{{\authnote{Yilei}{#1}}}

\ifllncs
  \documentclass[runningheads,a4paper]{llncs}

\else
  \documentclass[11pt]{article}
\fi

\usepackage[pdftex,pagebackref,colorlinks=true,pdfpagemode=UseNone,urlcolor=blue,linkcolor=blue,citecolor=blue,pdfstartview=FitH]{hyperref}

\usepackage{amsmath,amsfonts, amssymb}
\usepackage{graphicx}
\usepackage{amsthm}
\usepackage{algorithm}
\usepackage{algpseudocode}
\usepackage{tikz}
    \usetikzlibrary{arrows.meta}
    \usetikzlibrary{decorations.text}
\usepackage{url}
\usepackage[nameinlink]{cleveref}
\usepackage{caption}
\usepackage{subcaption}
\usepackage{cite}

\iffull
\fi

\ifllncs\else
\newtheorem{theorem}{Theorem}
\newtheorem{lemma}[theorem]{Lemma}
\newtheorem{definition}[theorem]{Definition}

\newtheorem{remark}[theorem]{Remark}

\newtheorem{corollary}[theorem]{Corollary}

\fi

\newcommand{\bra}[1]{\langle#1|}
\newcommand{\ket}[1]{|#1\rangle}
\newcommand{\bk}[2]{\langle#1|#2\rangle}
\newcommand{\kb}[2]{ \ket{#1}\bra{#2} }

\newcommand{\td}{\delta}
\newcommand{\reg}[1]{\mathsf{#1}  }

\newcommand{\poly}{{\sf poly}}
\newcommand{\negl}{{\sf negl}}

\newcommand{\GAP}{\mathsf{Gap}}
\newcommand{\DGS}{\mathsf{DGS}}
\newcommand{\QDGS}{\mathsf{\ket{DGS}}}

\newcommand{\SVP}{\mathsf{SVP}}

\newcommand{\SIVP}{\mathsf{SIVP}}

\newcommand{\CVP}{\mathsf{CVP}}
\newcommand{\EDCP}{\mathsf{EDCP}}

\newcommand{\LWE}{\mathsf{LWE}}

\newcommand{\QLWE}{\mathsf{S\ket{LWE}}}
\newcommand{\LWEstate}{\mathsf{C\ket{LWE}}}

\newcommand{\la}{\leftarrow}
\newcommand{\dist}{ \mathrm{dist} }
\newcommand{\sgn}{ \mathrm{sgn} }

\newcommand{\Lattice}{\mathcal{L}}

\newcommand{\rd}[1]{\left \lfloor #1 \right \rceil}
\newcommand{\upperrounding}[1]{\left \lceil #1 \right \rceil}
\newcommand{\lowerrounding}[1]{\left \lfloor #1 \right \rfloor}
\newcommand{\ipd}[2]{\left \langle {#1}, {#2} \right \rangle}

\newcommand{\mat}[1] { \mathbf{#1} }		
\newcommand{\ary}[1] { \mathbf{#1} }		

\newcommand{\C}{\mathbb{C}}
\newcommand{\N}{\mathbb{N}}

\newcommand{\R}{\mathbb{R}}
\newcommand{\set}[1]{ \left\{ #1 \right\}  }   
\newcommand{\abs}[1]{ \left| #1 \right|  }   
\newcommand{\tr}{\operatorname{tr}}

\newcommand{\Z}{\mathbb{Z}}

\newcommand{\QFT}{\mathsf{QFT}}
\newcommand{\DFT}{\mathsf{DFT}}

\newcommand{\cB}{\mathcal{B}}

\begin{document}

\title{LWE with Quantum Amplitudes: Algorithm, Hardness, and Oblivious Sampling}

\iffull

\author{
Yilei Chen\thanks{IIIS, Tsinghua University; Shanghai Artificial Intelligence Laboratory; and Shanghai Qi Zhi Institute. \texttt{chenyilei@mail.tsinghua.edu.cn}. \texttt{chenyilei.ra@gmail.com}. Supported by Tsinghua University startup funding. }
\and Zihan Hu\thanks{EPFL \texttt{zihan.hu@epfl.ch}. The work was initiated while Zihan Hu was in IIIS, Tsinghua University.}
\and Qipeng Liu\thanks{UC San Diego \texttt{qipengliu0@gmail.com}. The work was initiated while Qipeng Liu was a postdoc researcher in Simons Insititute, supported in part by the Simons Institute for the Theory of Computing, through a Quantum Postdoctoral Fellowship and by the DARPA SIEVE-VESPA grant No.HR00112020023. }
\and Han Luo\thanks{IIIS, Tsinghua University \texttt{luohan23@mails.tsinghua.edu.cn}.}
\and Yaxin Tu\thanks{Department of Computer Science, Princeton University \texttt{yaxin.tu@princeton.edu}. The work was initiated while Yaxin Tu was in IIIS, Tsinghua University. }
}

\date{\today}

\else

\author{Anonymous submission to Eurocrypt 2025
}
\fi

\maketitle

\begin{abstract}
The learning with errors problem (LWE) is one of the most important building blocks for post-quantum cryptography. To better understand the quantum hardness of LWE, it is crucial to explore quantum variants of LWE. To this end, Chen, Liu, and Zhandry [Eurocrypt 2022] defined $\QLWE$ and $\LWEstate$ problems by encoding the error of LWE samples into quantum amplitudes, and showed efficient quantum algorithms for a few interesting amplitudes. However, algorithms or hardness results of the most interesting amplitude, Gaussian, were not addressed before.

In this paper, we show new algorithms, hardness results and applications for $\QLWE$ and $\LWEstate$ with real Gaussian, Gaussian with linear or quadratic phase terms, and other related amplitudes. 
Let $n$ be the dimension of LWE samples. Our main results are
\begin{enumerate}
\item There is a $2^{\widetilde{O}(\sqrt{n})}$-time algorithm for $\QLWE$ with Gaussian amplitude with \emph{known} phase, given $2^{\widetilde{O}(\sqrt{n})}$ many quantum samples. The algorithm is modified from Kuperberg's sieve, and in fact works for more general amplitudes as long as the amplitudes and phases are completely \emph{known}.
\item There is a polynomial time quantum algorithm for solving $\QLWE$ and $\LWEstate$ for Gaussian with quadratic phase amplitudes, where the sample complexity is as small as $\tilde{O}(n)$. As an application, we give a quantum oblivious LWE sampler where the core quantum sampler requires only quasi-linear sample complexity. This improves upon the previous oblivious LWE sampler given by  Debris-Alazard, Fallahpour, Stehl\'{e} [STOC 2024], whose core quantum sampler requires $\tilde{O}(nr)$ sample complexity, where $r$ is the standard deviation of the error.
\item There exist polynomial time quantum reductions from standard LWE or worst-case GapSVP to $\QLWE$ with Gaussian amplitude with small \emph{unknown} phase, and arbitrarily many samples. Compared to the first two items, the appearance of the unknown phase term places a barrier in designing efficient quantum algorithm for solving standard LWE via $\QLWE$. 
\end{enumerate} 
\end{abstract}

\iffull
\newpage
\tableofcontents
\newpage
\fi

\iffull
\setcounter{page}{1}
\else
\fi

\section{Introduction}\label{sec:intro}

The learning with errors problem asks to learn a secret vector given many noisy linear samples.
\begin{definition}[Learning with errors (LWE)~\cite{DBLP:journals/jacm/Regev09}]\label{def:LWE}
Let $n$, $m$, $q$ be positive integers. 
Let $\ary{s} \in \Z_q^n$ be a secret vector. 
The search LWE problem $\LWE_{n,m,q,\alpha}$ asks to find the secret $\ary{s}$ given access to an oracle that outputs $\ary{a}_i$, $\ipd{\ary{s}}{\ary{a}_i} +e_i \pmod q$ on its $i^{th}$ query, for $i = 1, \cdots, m$. Here each $\ary{a}_i$ is a uniformly random vector in $\mathbb{Z}_q^n$, and each error term $e_i$ is sampled from the Gaussian distribution over $\Z$ with standard deviation $\alpha q/\sqrt{2\pi}$. 
\end{definition}

The LWE problem is extremely versatile, leading to advanced encryption schemes such as fully homomorphic encryptions~\cite{DBLP:conf/stoc/Gentry09,DBLP:conf/focs/BrakerskiV11,DBLP:conf/focs/Mahadev18a}. 
It is also shown by Regev to be quantumly as hard as the approximate short vector problems for all lattices~\cite{DBLP:journals/jacm/Regev09}. 
LWE and lattice problems in general (e.g.~\cite{DBLP:conf/ants/HoffsteinPS98,DBLP:journals/jacm/Regev09}) are also popular candidates for the NIST post-quantum cryptography standardization, due to their conjectured hardness against quantum computers. 
In fact, the fastest quantum and classical algorithms for LWE all run in $2^{\Omega(n)}$ time. 
However, the conjectured quantum hardness of lattice problem is still lacking solid evidences. Finding quantum algorithms for lattice problems has therefore been a major open problem in the area of quantum computation and cryptography in the past decade.

One way of exploring the quantum power for solving LWE is to consider the quantum variants of LWE, by encoding quantum states into the LWE problem (henceforth, we refer to the original LWE problem as ``classical LWE'' or ``standard LWE'' to distinct them from the quantum variant mentioned below). 
To this end, Chen, Liu, and Zhandry~\cite{DBLP:conf/eurocrypt/ChenLZ22} define the following quantum variants of LWE:

\begin{definition}[Solve $\ket{\mathsf{LWE}}$, $\QLWE$]\label{def:QLWEproblem}
    Let $n$, $m$, $q$ be positive integers. Let $f$ be a function from $\Z_q$ to $\C$.
    Let $\ary{s} \in \mathbb{Z}_q^n$ be a secret vector. 
    The problem $\QLWE_{n,m,q,f}$ asks to find $\ary{s}$ given access to an oracle that outputs independent samples $\ary{a}_i$, $\sum_{e_i\in\Z_q} f(e_i) \ket{ \ipd{\ary{s}}{\ary{a}_i} + e_i \bmod q}$ on its $i^{th}$ query, for $i = 1, \cdots, m$. Here each $\ary{a}_i$ is a uniformly random vector in $\mathbb{Z}_q^n$. 
\end{definition}

\begin{definition}[Construct $\ket{\mathsf{LWE}}$ states, $\LWEstate$]\label{def:LWEstateproblem}
Let $n$, $m$, $q$ be positive integers. Let $f$ be a function from $\Z_q$ to $\C$. 
The problem of constructing LWE states $\LWEstate_{n,m,q,f}$ asks to construct a  quantum state of the form 
$$\sum_{\ary{s}\in\Z_q^n} \bigotimes_{i = 1}^m\left( \sum_{e_i\in\Z_q} f(e_i) \ket{ \ipd{\ary{s}}{\ary{a}_i} + e_i \bmod q}\right),$$ given the input $\set{ \ary{a}_i }_{i = 1, ..., m}$ where each $\ary{a}_i$ is a uniformly random vector in $\mathbb{Z}_q^n$.
\end{definition}

Note that if there is a quantum algorithm that solves $\QLWE_{n,m,q,f}$ without collapsing the input state, then there is a quantum algorithm that solves $\LWEstate_{n,m,q,f}$ for the same $n, m, q, f$. So in the introduction we will say more about $\QLWE$ since it is more fundamental than $\LWEstate$. We will mention $\LWEstate$ when $\LWEstate$ is necessary for the application. 

Chen, Liu, and Zhandry~\cite{DBLP:conf/eurocrypt/ChenLZ22} then show a quantum filtering technique to solve $\QLWE$ when the DFT of the error amplitude is non-negligible everywhere over $\Z_q$.
Two interesting error amplitudes covered by their result are the bounded uniform amplitude and Laplacian amplitude. In particular, the classical LWE problem with bounded uniform error distribution is proven to be as hard as worst-case lattice problems~\cite{DBLP:conf/eurocrypt/DottlingM13,DBLP:conf/crypto/MicciancioP13}. Thus, their work gives a strong indication that $\QLWE$ is easier than classical LWE for certain error distributions.

However, the most interesting amplitude, Gaussian, was not addressed in~\cite{DBLP:conf/eurocrypt/ChenLZ22}.
For classical LWE, Gaussian distribution is the default error distribution since Regev~\cite{DBLP:journals/jacm/Regev09} shows a quantum reduction from worst-case lattice problems to LWE with Gaussian error, given \emph{arbitrarily} many LWE samples. 
Interestingly, it was implicitly shown in~\cite{DBLP:conf/asiacrypt/StehleSTX09} and~\cite{DBLP:conf/pkc/BrakerskiKSW18} that $\QLWE$ with Gaussian amplitude is as hard as standard LWE when the number of samples is \emph{very small}. But if we are given \emph{arbitrarily} many samples, is the $\QLWE$ problem with Gaussian amplitude still hard?

Not only did understanding $\QLWE$ with Gaussian amplitude shed light on our knowledge of the standard LWE, but it was also considered a bedrock on which interesting quantum protocols can be based, especially unclonable cryptography. The idea was first shown by Zhandry~\cite{DBLP:conf/eurocrypt/Zhandry19a} for a potential approach to construct a very powerful cryptographic primitive called quantum lightning, while its concrete and secure instantiation still remains unknown. After that, Khesin, Lu, and Shor proposed another lightning construction~\cite{khesin2022publicly} based on $\QLWE$ but later was broken by Liu, Montgomery, and Zhandry~\cite{liu2023another}. Poremba~\cite{DBLP:conf/innovations/Poremba23}, and Ananth, Poremba, and Vaikuntanathan~\cite{ananth2023revocable} build certifiable deletion based on $\QLWE$, basing on certain conjectured security. If $\QLWE$ with Gaussian amplitude was not secure, then all schemes mentioned before may not even have semantic security, let alone its unclonability. 

In addition, recently, Debris{-}Alazard, Fallahpour, and Stehl{\'e}~\cite{debris2024quantum} showed that solving the $\LWEstate$ problem for proper amplitudes implies a quantum \emph{oblivious LWE sampler}. Roughly speaking, an oblivious LWE sampler is able to take as input the public matrix $\mat{A}$, output an LWE sample $\mat A^T\ary{s}+\ary{e} \bmod q$ without knowing the secret $\ary{s}$. If an oblivious LWE sampler exists, it refutes the typical knowledge assumption for LWE used in several lattice-based SNARK constructions (e.g.~\cite{gennaro2018lattice}, see~\cite{debris2024quantum} for more discussions) which assume that anyone who generates an LWE sample must know the LWE secret. While we don't know of any classical oblivious LWE sampler, Debris{-}Alazard et al. showed that solving $\LWEstate$ is perfect for the task of oblivious LWE sampling, since if we are able to generate a state like $\sum_{\ary{s}\in\Z_q^n} \bigotimes_{i = 1}^m\left( \sum_{e_i\in\Z_q} f(e_i) \ket{ \ipd{\ary{s}}{\ary{a}_i} + e_i \bmod q}\right)$, and then measure it, we will get an LWE sample $\mat{A}^T\ary{s}+\ary{e}\bmod q$ with a random and unknown secret $\ary{s}$. Debris{-}Alazard et al. then design an amplitude $f$ whose norm is Gaussian, and whose phase is cleverly chosen, so that the algorithm of quantum unambiguous measurement~\cite{chefles1998optimum} is able to solve $\LWEstate$ for such an amplitude with relatively few samples.   

In sum, the motivations of this work are addressing the following questions:
\begin{center}
{\em  What amplitudes are useful for $\QLWE$ or $\LWEstate$ towards solving standard LWE or improving quantum applications like oblivious LWE sampling? }
\end{center}
\begin{center}
{\em  Can we show algorithms or hardness for solving $\QLWE$ or $\LWEstate$ for Gaussian and other related amplitudes? }
\end{center}

\subsection{Main results}

In this paper, we show new quantum algorithms, hardness results, and applications for $\QLWE$ and $\LWEstate$ with Gaussian and other related amplitudes. A summary of the results is given in \Cref{table:summary}. 
\iffull
\begin{table}
\footnotesize
    \centering
    \begin{tabular}{cccc}
        \hline
        {\small\bf Error Amplitude} &
        {\small\bf \# Samples} &{\small\bf Algorithm or Hardness} & {\small\bf Reference} \\
        \hline
        {Gaussian} &
        { $\tilde{O}(n)$ } &
        {As hard as LWE or approx-GapSVP} & 
        {\cite{DBLP:conf/asiacrypt/StehleSTX09, DBLP:conf/pkc/BrakerskiKSW18}} \\
        {Gaussian} &
        { $2^{\tilde{O}(\sqrt{n})}$ } &
        { $2^{\tilde{O}(\sqrt{n})}$-time quantum algorithm } & 
        { This work \S\ref{sec:kuper} } \\
        {Complex Gaussian} &
        { $\tilde{O}(n)$ } &
        { $\poly(n)$-time quantum algorithm } & 
        { This work \S\ref{sec:osampling} } \\
        {Gaussian with unknown phase} &
        { Arbitrary } &
        { As hard as LWE or approx-GapSVP } & 
        { This work \S\ref{sec:EDCP},\ref{sec:regev} } \\
        {Known with non-negl DFT} &
        { $\poly(n)$ } &
        { $\poly(n)$-time quantum algorithm } & 
        { \cite{DBLP:conf/eurocrypt/ChenLZ22, debris2024quantum} } \\
        \hline
    \end{tabular}
    \caption{Algorithm and hardness of $\QLWE$ for various error amplitudes}\label{table:summary}
\normalsize
\end{table}
\fi

Before introducing our new algorithms for solving $\QLWE$, let us take a step back and review the existing methods of solving $\QLWE$ for a general amplitude $f$. Fix some amplitude $f:\Z_q \to \C$ and suppose for now that $f$ is centered at $0$. Then each sample of $\QLWE_{n, m, q, f}$ is given by 
\begin{equation}\label{eqn:defphifc}
\ary{a}\in\Z_q^n, ~~~ \ket{ \phi_{f, \ipd{\ary{a}}{\ary{s}}} } := \sum_{x\in \Z_q} f(x) \ket{ x+\ipd{\ary{a}}{\ary{s}} \bmod q }. 
\end{equation}

The state $\ket{ \phi_{f, \ipd{\ary{a}}{\ary{s}}} }$ can be seen as a shift of $\ket{ \phi_{f, 0} }$ where the center of the state is shifted from $0$ to $\ipd{\ary{a}}{\ary{s}}$. A general approach of solving $\QLWE$ is then: try to design a quantum algorithm that predicts the center $\ipd{\ary{a}}{\ary{s}}$ of each state $\ket{ \phi_{f, \ipd{\ary{a}}{\ary{s}}} }$ correctly, and if we predict the centers of $n$ states correctly, use Gaussian elimination to find out $\ary{s}$. 

However, a difficulty faced by the general approach is: for a ``typical'' error amplitude, like Gaussian of some width $r\in (\sqrt{n}, q/\sqrt{n})$, i.e., for $\rho_r(x) := \exp\left( -\pi \frac{x^2}{r^2} \right)$, the overlap between $\ket{ \phi_{\rho_r, 0} }$ and $\ket{ \phi_{\rho_r, 1} }$ is quite large, which means it is inherently hard to distinguish states with adjacent shifts with high confidence. The quantum algorithms used in \cite{DBLP:conf/eurocrypt/ChenLZ22, debris2024quantum} have to use sophisticated measurements depending on $f$ to guarantee that once the center $c$ of $\ket{ \phi_{f, c} }$ is predicted correctly for some samples, we know it is correct. But the probability of predicting correctly in the algorithms used in \cite{DBLP:conf/eurocrypt/ChenLZ22, debris2024quantum} is proportional to the square of the minimum of $|\DFT_q(f)|$ and is in general very low (the success probability of \cite{debris2024quantum} is better than \cite{DBLP:conf/eurocrypt/ChenLZ22} in its dependency on $q$). Although the success probability is non-negligible for certain $f$, it is exponentially small for the typical Gaussian amplitude. For example, in \cite{debris2024quantum}, the required sample complexity is $\frac{n\cdot \omega(log n)}{q\cdot \min_{x\in\Z_q}|\DFT_q(f)(x) |^2}$. By taking 
$\min_{x\in\Z_q} |\DFT_q(\rho_{r})(x)|^2 \approx  \frac{1}{\frac{q}{\sqrt{2}r} } \cdot \rho_{\frac{q}{\sqrt{2}r}}\left(\frac{q}{2}\right) \approx \frac{\sqrt{2}r}{q }\cdot e^{-\pi \frac{r^2}{2}}$, 
the sample complexity for solving $\QLWE$ for Gaussian amplitude $\rho_r$ is greater than $O(n/r)\cdot e^{\pi \frac{r^2}{2}}$, which is exponential in $n$ when $r>\sqrt{n}$ (when $r$ is in $O(n^{0.5 - c})$ for any $c>0$, the classical Arora-Ge algorithm~\cite{DBLP:conf/icalp/AroraG11} has solved classical LWE with error distribution $\rho_r$ in subexponential time and sample complexity already).
Therefore \cite{DBLP:conf/eurocrypt/ChenLZ22, debris2024quantum} are not able to provide subexponential time algorithms for solving $\QLWE$ for the Gaussian amplitude $\rho_r$ when $r>\sqrt{n}$.

\paragraph{Subexponential time algorithm for general amplitudes. }

Our first result is a sub-exponential time quantum algorithm for solving $\QLWE$ with Gaussian amplitude, given sub-exponentially many $\QLWE$ samples. In fact, it works for a more general amplitude $f$ as long as two points in the discrete Fourier transform of $f$ are more than subexponentially small (the DFT of Gaussian certainly has two non-negligible points). The algorithm combines Kuperberg's sieve \cite{DBLP:journals/siamcomp/Kuperberg05} and quantum rejection sampling~\cite{ozols2013quantum}.

\begin{theorem}[\Cref{thm:subexp_alg_gen}, informal]\label{thm:subexp_alg_intro}
    Let $f:\Z\to \C$ be a known, normalized error amplitude function for $\QLWE$ such that for the $\DFT_q$ of $f$, denoted by $g$, there exists two distinct values $j_1, j_2\in \Z_q$ such that $\gcd(j_1 - j_2, q) = 1$ and   $|g(j_1)|, |g(j_2)|\geq 2^{-\sqrt{n}\log q}$, and $g(j_1), g(j_2)$ are computable in time $2^{\Theta(\sqrt{n}\log q)}$.
    Then there exists a quantum algorithm that, given $m = 2^{\Theta(\sqrt{n}\log q)}$ samples of $\QLWE$ with amplitude $f$, finds the secret within a time complexity of $2^{\Theta(\sqrt{n}\log q)}$.
\end{theorem}

\paragraph{Polynomial time algorithm for complex Gaussian amplitudes. }
Our second result shows when the amplitude $f$ is Gaussian with quadratic imaginary phase (henceforth complex Gaussian), i.e.,  $f_{r, t}(x) = \exp\left( -\pi \left( \frac{1}{r^2} + \frac{i}{t} \right)x^2\right)$ for some $r\in\R$ and $t\in \N^+$, there are $\poly(n, \log q)$ time quantum algorithms for solving $\QLWE$ and $\LWEstate$. 

Our algorithm makes use of an interesting property of the complex Gaussian function $f_{r, t}(x) = \exp\left( -\pi \left( \frac{1}{r^2} + \frac{i}{t} \right)x^2\right)$: when $r>t\cdot \Omega(\sqrt{n})$, the overlap between $\ket{ \phi_{f_{r, t}, c_1} }$ and $\ket{ \phi_{f_{r, t}, c_2} }$ (following the notation in Eqn.~\eqref{eqn:defphifc}) is negligible whenever $c_1 \not\equiv c_2 \pmod t$. In other words, given a state $\ket{ \phi_{f_{r, t}, c} }$, it should in principle be easy to extract the value of $c \bmod t$. Indeed, we show when $r>\tilde{O}(n)\cdot t$, if we simply measure the least significant bits of $\ket{ \phi_{f_{r, t}, c} }$ by the purely imaginary Gaussian basis $\left\{\ket{\psi_d}\right\}_{d\in \Z_{t}}$, where
\iffull
\[
    \ket{\psi_d} = \sum_{x\in \Z_{t}} e^{-\frac{\pi i(x - d)^2}{t}}\ket{x},
\]
\fi
and output the measurement result $d\in \Z_{t}$, then $d = c\bmod t$ with probability $1 - O(1/n)$, which finds the center $c\bmod t$ of $\ket{\phi_{f_{r,t},c}}$ with high probability. Utilizing this center finding procedure, we can solve a variant of $\QLWE$ 
where $q = q_1q_2$ with coprime factors $q_1,q_2$, such that the first half of samples can be viewed as from $\QLWE_{n, \tilde{O}(n), q_1, f_{r,q_1}}$, while the second half can be viewed as samples from $\QLWE_{n, \tilde{O}(n), q_2, f_{r,q_2}}$ with the same secret $\ary{s}\bmod q$. Then we employ Guess-then-Gaussian-Elimination on the two halves to recover $\ary{s}\bmod q_1$ and $\ary{s}\bmod q_2$, and finally use the Chinese Remainder Theorem to recover the whole secret in $\Z_q$. 

The idea can be generalized to the case where $q$ has more than two factors, and it directly works for solving $\LWEstate$. Our main theorem is
\begin{theorem}\label{thm:CLWEintro}
    Let $n, m, q, \ell$ be positive integers and $r$ be a real number. Suppose that $m = 2\ell n \cdot \omega(\log n)$, $q$ is a composite number satisfying $q = q_1q_2\cdots q_\ell$ where $q_1, q_2, \cdots, q_\ell$ are coprime, $r$ satisfies $\frac{q}{\sqrt{n}} > r > 30n\log n\cdot \max{\set{q_1, q_2, \cdots, q_\ell}}$. 
    
    There exists a quantum algorithm running in time $\poly(n, \ell, \log q)$ that, takes input $\mat{A} \gets \mathcal{U}(\Z_q^{n \times m})$, outputs a state $\rho$ such that the trace distance between $\rho$ and $\phi := \kb{\phi}{\phi}$ is negligible, where 
    \[
        \ket{\phi} = \sum_{\ary{s} \in \Z_q^n}\sum_{\ary{x} \in \Z^m}\rho_r(\ary{x})e^{-\pi i \sum_{j = 1}^{\ell} \|\ary x_j\|^2/q_j} \ket{(\mat A^T\ary{s} + \ary{x}) \bmod q}
    \]
    is a quantum LWE state with Gaussian amplitude and quadratic phase terms, here we write $\ary x$ into $\ell$ blocks $\ary x = (\ary x_1, \ary x_2, \cdots, \ary x_\ell)$ with $\ary x_j\in \Z^{m/\ell}, j = 1, 2, \cdots, \ell$.
\end{theorem}

\paragraph{Application to oblivious LWE sampling. }

As mentioned in \cite{debris2024quantum}, to build a quantum oblivious LWE sampler via solving $\LWEstate$, the key is to choose the amplitude $f$ where the real part is Gaussian, and the imaginary part is suitable so that an efficient quantum algorithm can find the center $c$ of $\ket{ \phi_{f, c} }$ with non-negligible probability. In \cite{debris2024quantum}, $f$ is chosen to be \emph{Gaussian with half phase}: $f(x) = \exp\left( -\pi \frac{x^2}{r^2} \right) \cdot \sgn^+(x)$, where $\sgn^+(x)$ outputs $1$ when $x\geq 0$, $-1$ when $x<0$. They are able to solve $\LWEstate_{n, m, q, f}$ for such $f$ with sample complexity $m\in \tilde{O}(nr)$ and with prime modulus $q$. Once they get an oblivious LWE sample with parameter $(n, m, q, r/q)$ (here $r/\sqrt{2\pi}$ is the standard deviation of the LWE error term), they can throw away additional samples, and apply the classical modulus switching transformation \cite{DBLP:conf/stoc/BrakerskiLPRS13} to get oblivious LWE samples with parameter $(n, m', q', r'/q')$ for certain $m'\leq m$, $q'<q$, $r'>r$ (we omit more detailed restrictions in the introduction).

Using the complex Gaussian amplitude, we can directly use \Cref{thm:CLWEintro} with $\ell\in O(1)$ to reduce the sample complexity of $\LWEstate$ from $\tilde{O}(nr)$ to $\tilde{O}(n)$. Here is the main theorem for the core quantum component of our oblivious LWE sampler:
\begin{theorem}[Informal version of \Cref{coro:OSAMP}]\label{coro:OSAMPintro}
    Let $n, m, q, \ell, r$ satisfy the same conditions as in \Cref{thm:CLWEintro}. Assume the quantum hardness of $\LWE_{n, m, q, r/(\sqrt{2}q)}$, there exists a witness-oblivious quantum sampler for $\LWE_{n, m, q, r/(\sqrt{2}q)}$.
\end{theorem}

Once we get an oblivious LWE sampler for some composite modulus $q$, we can also use modulus switching to get an oblivious LWE sampler for certain prime modulus $q'<q$. The detailed parameters are given in \Cref{coro:modulusswitch} in the main body.

In \Cref{table:summaryOLWE} we summarize the known poly-time algorithms for solving $\LWEstate$ and oblivious LWE sampling. 
\iffull
In \Cref{fig:summary} we provide some examples of interesting amplitudes addressed in our paper or previous papers. 
\fi
\begin{table}[t]
\footnotesize
    \centering
    \begin{tabular}{cccc}
        \hline
        {\small\bf Error Amplitude} &
        {\small\bf \# Samples} &{\small\bf Algorithm} & {\small\bf Reference} \\
        \hline
        {Bounded uniform} &
        { $O(n r q^4 )$ } &
        { Quantum filtering } & 
        { \cite{DBLP:conf/eurocrypt/ChenLZ22} } \\
        {Gaussian with half phase} &
        { $\tilde{O}(nr)$ } &
        { Unambiguous measurement } & 
        { \cite{debris2024quantum} } \\
        {Complex Gaussian} &
        { $\tilde{O}(n)$ } &
        { Imaginary Measurement and CRT } & 
        { This work \S\ref{sec:osampling} } \\
        \hline
    \end{tabular}
    \caption{Algorithm and sample complexity for $\LWEstate$ and Oblivious LWE sampling. Here $n$ is the dimension of the LWE secret, $r$ is the standard deviation of the error distribution, $q$ is the modulus. All quantum algorithms in this table run in $\poly(n)$ time.  }\label{table:summaryOLWE}
\normalsize
\end{table}

\iffull
\begin{figure}[h]
\centering
\includegraphics[scale=0.24]{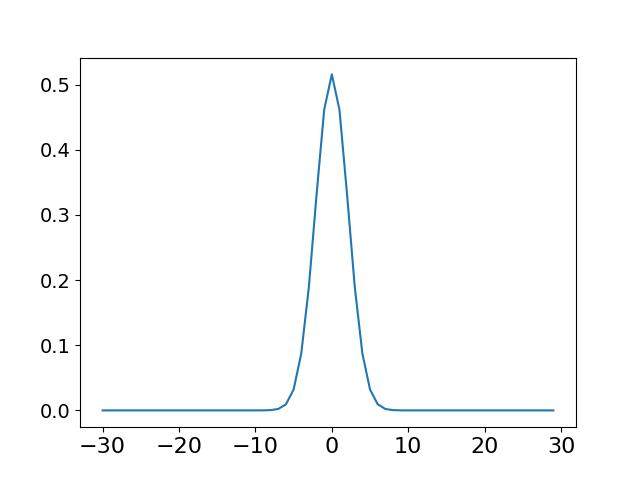}
\includegraphics[scale=0.24]{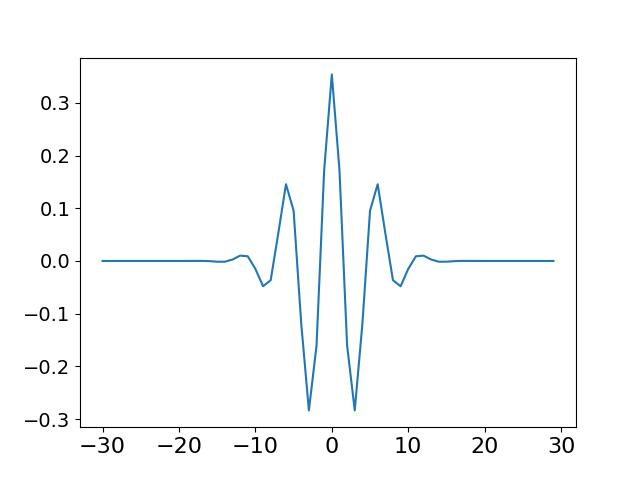}
\includegraphics[scale=0.24]{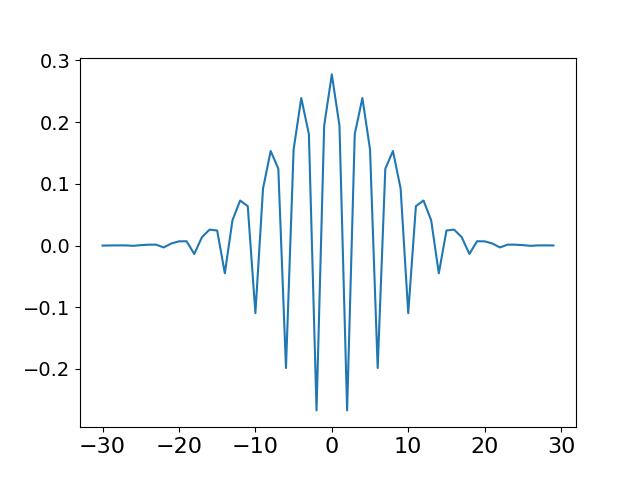}
\includegraphics[scale=0.24]{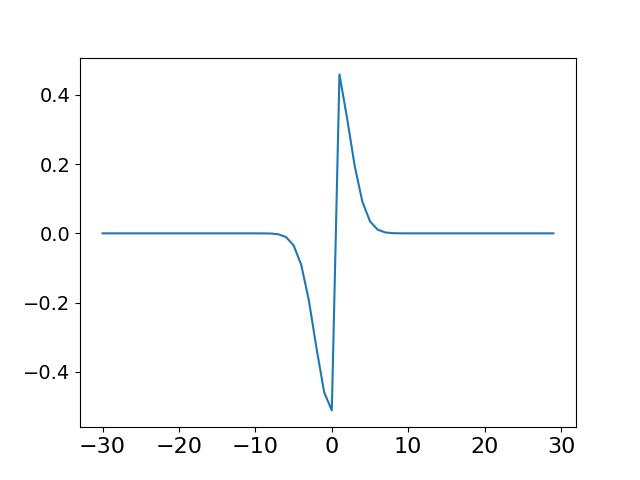}
\includegraphics[scale=0.24]{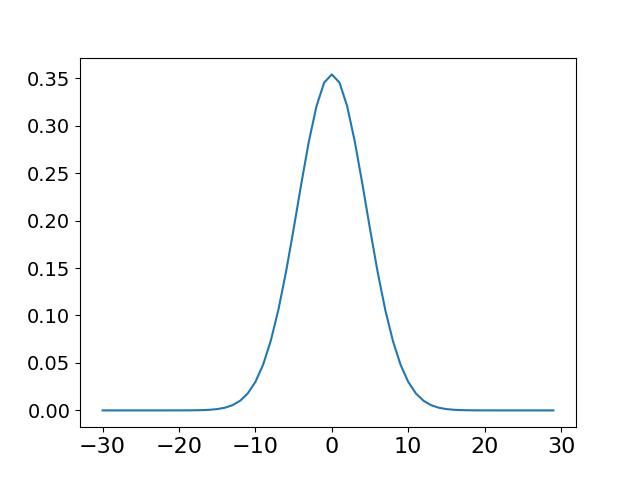}
\includegraphics[scale=0.24]{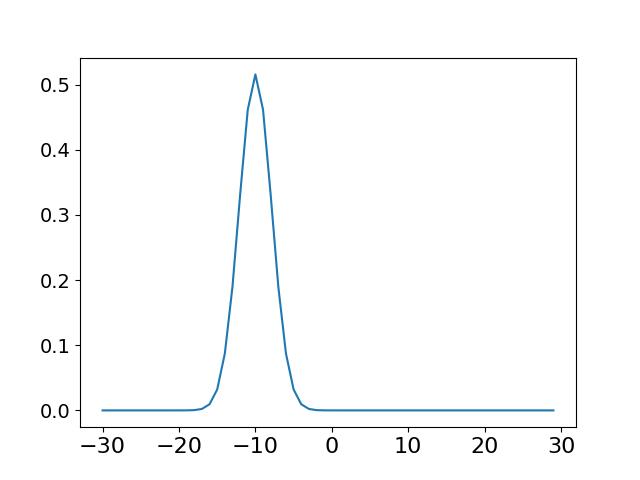}
\includegraphics[scale=0.24]{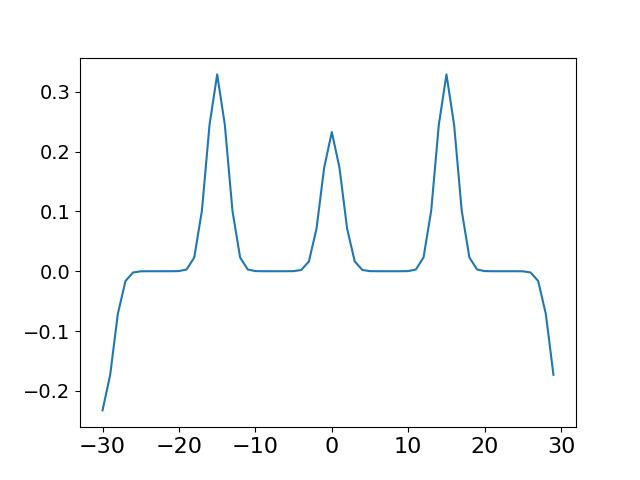}
\includegraphics[scale=0.24]{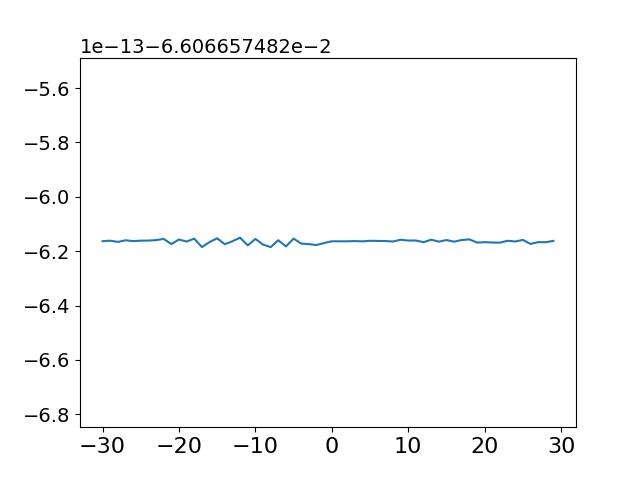}
\caption{ Interesting $\QLWE$ error amplitudes (top) and their DFTs (bottom). 
All pictures are depicting the real parts of the functions. The $x$-axis is the input (from $-30$ to $29$, all examples are given over $\Z_{60}$). The $y$-axis is the amplitude. 
Four pictures on the top from left to right are: (1) Gaussian, where our sub-exponential algorithm applies; (2) Gaussian with imaginary linear phase, where our reductions apply when the phase (or the center of the DFT) is unknown; (3) Gaussian with imaginary quadratic phase, where our oblivious LWE sampler uses; (4) Gaussian where the phase changes in the middle, where the oblivious LWE sampler in \cite{debris2024quantum} uses.  }\label{fig:summary}
\end{figure}
\fi


\paragraph{Hardness results. }
Given our new quantum algorithms for solving $\QLWE$ for Gaussian and complex Gaussian amplitudes, readers may wonder whether they lead to algorithmic advantages for standard LWE. In fact, the complex Gaussian function was introduced by Chen \cite{chen2024quantum} in a failed attempt of solving standard LWE. We haven't been able to fix the problem in \cite{chen2024quantum}. Instead, we extract the intuition from \cite{chen2024quantum} and show that complex Gaussian is at least suitable for solving $\QLWE$, $\LWEstate$, and reduce the quantum sample complexity for oblivious LWE sampling. We are not able to show that solving $\QLWE$, $\LWEstate$ with complex Gaussian amplitude implies solving standard LWE. 

How about the real Gaussian amplitude? Whether \Cref{thm:subexp_alg_intro} leads to sub-exponential time quantum algorithms for standard LWE?
To address this question, let us recall the quantum reduction from GapSVP and SIVP to classical LWE with Gaussian error distribution due to Regev~\cite{DBLP:journals/jacm/Regev09}. This reduction works even given \emph{arbitrarily many} classical LWE samples. At first glance, one would conjecture that the reduction can be modified to a quantum reduction from worst-case lattice problems to $\QLWE$ with Gaussian amplitude. 

However, we are only able to modify Regev's reduction into a quantum reduction from worst-case lattice problem to $\QLWE$ with Gaussian amplitude with some small but \emph{unknown} phase, with arbitrarily many quantum samples. Before stating our main theorem, let us first introduce the definition of $\QLWE$ with phase.

\begin{definition}[$\QLWE^{\sf phase}$]\label{def:SLWEphase}
    Let $n, m, q$ be LWE parameters. Define the following components: (1) an amplitude function $f: \mathsf{supp}(f)\to \R_{\ge 0}$ where $\mathsf{supp}(f) = \{x/Q: -q/2 < x/Q \leq q/2, x \in \Z\}$ for some integer $Q$; (2) a mapping $\theta: \mathsf{supp}(\theta)\to \R$ where $\mathsf{supp}(\theta)$ is a subset of high-dimensional integer vectors; (3) a distribution $D_{\theta}$ over the set $\mathsf{supp}(\theta)$.
    
    We say a quantum algorithm is capable of solving $\QLWE_{n, m, q, f, \theta, D_{\theta}}^{\mathsf{phase}}$, if for any hidden vector $\ary s\in \Z_q^n$, when provided with $m$ samples of 
    \[
        \ary a\gets \mathcal{U}(\Z_q^n), \quad \ary y\gets D_{\theta}, \quad \sum_{e\in \mathsf{supp}(f)} f(e) \exp(2\pi {\rm i}\cdot e\theta(\ary{y})) \ket{(\ipd{\ary{a}}{\ary{s}} + e)\bmod q},
    \]
    the algorithm outputs $\ary s$ with probability at least $1 - 2^{-\Omega(n)}$.
\end{definition}

On the first pass of the definition, readers can think of the integer vector $\ary{y}$ as some auxiliary information. The phase term $\theta(\ary{y})$ is a function of $\ary{y}$. As it is defined, the function $\theta$ may or may not be efficiently computable
(it is not efficiently computable in our result). So we can think of $\QLWE^{\sf phase}$ as a variant of $\QLWE$ with a phase term in the amplitude. 

\begin{theorem}[\Cref{thm:Regev_MainThm}, informal]\label{thm:Regev_MainThm_intro}
    Let $q = q(n) > 10n$ be an integer of at most $\poly(n)$ bits, $\alpha \in (0, \frac{1}{5\sqrt{n}})$ such that $\alpha q > 2\sqrt{n}$. 
    Suppose there exists a quantum algorithm that solves $\QLWE^{\sf phase}$ where the amplitude $f(e) := \exp\left( -\pi \frac{ e^2 }{ (\alpha q)^2} \right)$ and the phase $\theta$ is not efficiently computable, with $m = 2^{o(n)}$ samples and in time complexity $T$. Then there exists a quantum algorithm that solves $\GAP\SVP_{\gamma}$ and $\SIVP_\gamma$, where $\gamma\in \tilde{O}(n/\alpha)$, in time $\poly(n, m, T)$. 
\end{theorem}

The informal statement above omits the distribution of the unknown phase term $\theta(\ary{y})$ and some other details. All those details can be found in the statement of \Cref{thm:Regev_MainThm}. 
Morally, \Cref{thm:Regev_MainThm} says there is a quantum reduction from worst-case lattice problems to $\QLWE$ with Gaussian amplitude with unknown phase, with arbitrarily many samples. Let us remark that the distribution of the unknown phase is \emph{known} to be Gaussian with \emph{small} width. 

We also provide a quantum reduction directly from classical LWE to $\QLWE^{\sf phase}$ in \Cref{thm:EDCP_MainThm}. 
This reduction goes through the (extrapolated) dihedral coset problem, originally used in \cite{DBLP:journals/siamcomp/Regev04,DBLP:conf/pkc/BrakerskiKSW18}. 
It achieves slightly worse parameters compared to \Cref{thm:Regev_MainThm}, but is much simpler to describe. 
For more details we refer the readers to \Cref{sec:EDCP}.
Although the result appears to be qualitatively similar to \Cref{thm:Regev_MainThm}, as it also says $\QLWE$ with Gaussian amplitude with small unknown phase is as hard as classical LWE; the algorithm used in the reduction is very different, therefore it might offer a different approach for potential improvements. 

Compared to the reduction of \cite{DBLP:conf/pkc/BrakerskiKSW18}, while their reduction can be seen as converting classical LWE to $\QLWE$ samples with non-negligible probability of failure (meaning that the amplitude of $\QLWE$ samples does not follow an expected shape; the event of failure is not efficiently detectable), our reduction converts classical LWE to $\QLWE$ samples which always have a Gaussian amplitude, but the phase has a small unknown term. Strictly speaking, our result is incomparable with the result in \cite{DBLP:conf/pkc/BrakerskiKSW18}. But we hope that maybe some quantum algorithms in future can handle $\QLWE$ samples with small unknown phase better than $\QLWE$ samples with failure. 

\paragraph{Is the unknown phase an inherent barrier?}
Overall, we show subexponential time quantum algorithm for $\QLWE$ with known amplitudes. We also reduce worst-case lattice problems, or classical LWE, to $\QLWE$ with unknown phase. Readers may wonder whether the \emph{unknown phase} would be an inherent barrier for getting a subexponential quantum algorithm for solving standard LWE. We think it is a wonderful question for future work. On one hand, the unknown phase appeared in our reduction is not completely random (if the unknown phase were completely random then the quantum LWE samples are as useless as classical LWE samples) -- we know they are small and follow Gaussian distributions, so there are some hope of further utilizing the unknown phases (such as trying to use the filtering technique in \cite{DBLP:conf/eurocrypt/ChenLZ22} to guess the phase).  On the other hand, all our attempts of utilizing the unknown phase failed. The reasons of failures are rather technical and entangled with the details of our algorithms, so we refer readers to the main body for discussions therein. We think it is worth to understand whether the unknown phase is an inherent obstacle towards solving LWE in quantum subexponential time, or there is a hope of finding innovative methods in handling the unknown phase appeared in our reductions.


\paragraph{Organization.}
The rest of this paper is organized as follows. In \Cref{sec:prelim} we provide the background of quantum computation and lattice problems. In \Cref{sec:kuper} we provide our sub-exponential time quantum algorithm for $\QLWE$ with completely known amplitude. In \Cref{sec:osampling} we provide our polynomial time quantum algorithm for $\QLWE$ with complex Gaussian amplitudes, and the application to oblivious LWE sampling. In \Cref{sec:EDCP} we provide the reduction from classical LWE to $\QLWE^{\sf phase}$ via extrapolated DCP. In \Cref{sec:regev} we provide the reduction from worst-case lattice problem to $\QLWE^{\sf phase}$ via quantizing Regev's reduction. We choose to present the reduction from classical LWE to $\QLWE^{\sf phase}$ via extrapolated DCP first, since this reduction is relatively easier to follow. 
\Cref{sec:kuper}, \Cref{sec:osampling}, \Cref{sec:EDCP}, \Cref{sec:regev} are all in fact self-contained and independent, containing their own overview if necessary, so readers can start from any section without reading the others.

\section{Preliminaries}\label{sec:prelim}

\paragraph{Notations and terminology.} 
Let $\C, \R, \R_{\ge 0}, \Z, \N, \N^+$ be the set of complex numbers, real numbers, non-negative real numbers, integers,  natural numbers (non-negative integers), and positive integers. 
Denote $\Z/q\Z$ by $\Z_q$.
By default we represent the elements of $\Z_q$ by elements in $(-q/2, q/2 ]\cap \Z$.
For any integer $q\geq 2$, let $\omega_q = e^{2 \pi {\rm i} / q}$ denote the primitive $q$-th root of unity. 
The rounding operation $\rd{a}: \R \to \Z$ rounds a real number $a$ to its nearest integer (if $a\in \Z+0.5$, we round it to $a+0.5$).
For positive integer $q$ the rounding operation $\rd{a}_q: \R\to q\Z$ rounds a real number $a$ to its nearest integer which is a multiple of $q$. 
For $n\in\N^+$, let $[n] := \set{1, 2, \cdots, n}$. 

A vector in $\R^n$ (represented in column form by default) is written as a bold lower-case letter, e.g. $\ary{v}$. For a vector $\ary{v}$, the $i^{th}$ component of $\ary{v}$ will be denoted by $v_i$. A matrix is written as a bold capital letter, e.g. $\mat{A}$. The $i^{th}$ column vector of $\mat{A}$ is denoted $\ary{a}_i$. 

For $x \in \R$ and $q \in \N^+$, let $x \bmod q$ be the unique real number $z \in (-q/2, q/2 ]$ such that $x - z$ is a multiple of $q$. For a vector $\ary{v} \in \R^n$, we do $\bmod$ coordinate-wise, i.e. the $i^{\text{th}}$ coordinate of $\ary{v} \bmod q$ is given by $v_i \bmod q$. To avoid ambiguity, we give $\bmod$ lower precedence than addition/subtraction. For example, $\ary{a} + \ary{b} \bmod q$ means $(\ary{a} + \ary{b}) \bmod q$.

The length of a vector is the $\ell_p$-norm $\|\ary{v}\|_p := (\sum v_i^p)^{1/p}$, or the infinity norm given by its largest entry $\|\ary v\|_{\infty} := \max_i\{|v_i|\}$. The $\ell_p$ norm of a matrix is the norm of its longest column: $\|\mat{A}\|_p := \max_i \|\ary{a}_i\|_p$. 
By default we use $\ell_2$-norm unless explicitly mentioned. 
Let $\ary{x}\in\C^n$, then $\|\ary{x}\|_\infty\leq \|\ary{x}\|_2 \leq \|\ary{x}\|_1$.
Let $B^n_p$ denote the open unit ball in $\R^n$ in the $\ell_p$ norm.

When a variable $v$ is drawn uniformly random from the set $S$ we denote as $v\la \mathcal U(S)$. 
When a function $f$ is applied on a set $S$, it means $f(S) := \sum_{x\in S}f(x)$. 
\begin{definition}[Statistical distance]
    For two distributions over $\R^n$ with probability density functions $f_1$ and $f_2$, we define the statistical distance between them as 
    $$D(f_1, f_2) = \frac{1}{2}\int_{\R^n}|f_1(\ary{x}) - f_2(\ary{x})| d\ary{x}.$$
\end{definition}

We say two distributions (respectively, quantum states) are $\epsilon$-close to each other if their statistical distance (respectively, trace distance by default) is at most $\epsilon$. We say two pure (unnormalized) states $\ket{\phi}$ and $\ket{\psi}$ are $\epsilon$-close in $\ell_2$ distance if $\|\ket{\phi} - \ket{\psi}\| \le \epsilon \max(\|\ket{\phi}\|, \|\ket{\psi}\|)$.

\paragraph{Fourier transform. }
The Fourier transform of a function $h: \R^n \to \C$ is defined to be
\begin{equation*}
    \hat{h}(\ary{w}) = \int_{\R^n} h(\ary{x}) \exp( -2\pi {\rm i} \ipd{\ary x}{\ary w} ) d\ary{x}.
\end{equation*}

Define the convolution of two functions as $f * g(\ary{y}) = \int_{\R^n} f(\ary{x})g(\ary{y}-\ary{x}) d\ary{x}$. Then $\widehat{f * g} = \hat{f}\cdot \hat{g}$ and $\widehat{f \cdot g} = \hat{f} * \hat{g}$. 

We recall some formulas about Fourier transform (cf.~\cite[P.100,~Proposition~2.2.11]{grafakos2008classical}). 
If $h$ is defined by $h(\ary x) = g(\ary x+\ary v)$ for some function $g: \R^n \to \C$ and vector $\ary v\in\R^n$, then
\iffull
\begin{equation}\label{eqn:Fourier_shift1}
    \hat{h}(\ary w) = \hat{g}(\ary w)\cdot \exp( 2\pi i \ipd{\ary v}{\ary w} ).
\end{equation}
\fi
If $h$ is defined by $h(\ary x) = g(\ary x) \exp( 2\pi i \ipd{\ary x}{\ary v} )$ for some function $g: \R^n \to \C$ and vector $\ary v\in\R^n$, then
\iffull
\begin{equation}\label{eqn:Fourier_shift2}
    \hat{h}(\ary w) = \hat{g}(\ary w - \ary v).
\end{equation}
\fi

\iffull
As a corollary of Eqns.~\eqref{eqn:Fourier_shift1} and~\eqref{eqn:Fourier_shift2}, 
if $h$ is defined by $h(\ary x) = f(\ary x +\ary v) \exp( 2\pi i \ipd{\ary x}{\ary z} )$ for some function $f: \R^n \to \C$ and vectors $\ary v$, $\ary z\in\R^n$, 
then we define $g(\ary x) := f(\ary x +\ary v)$, so $h(\ary x) = g(\ary x) \exp( 2\pi i \ipd{\ary x}{\ary z} )$. 
Therefore $\hat{g}(\ary w) = \hat{f}(\ary w)\cdot \exp( 2\pi i \ipd{\ary v}{\ary w} )$, and
\[
    \hat{h}(\ary w) = \hat{g}(\ary w - \ary z) = \hat{f}(\ary w - \ary z)\cdot \exp( 2\pi i \ipd{\ary v}{\ary w - \ary z} ).
\]
\fi

\subsection{Lattices}\label{sec:lattice}

An $n$-dimensional lattice $\Lattice$ of rank $k\leq n$ is a discrete additive subgroup of $\R^n$. Given $k$ linearly independent basis vectors $\mat{B} =\set{\ary{b}_1, \cdots, \ary{b}_k\in \R^n}$, the lattice generated by $\mat{B}$ is 
\[ \Lattice(\mat{B}) = \Lattice(\ary{b}_1, \cdots, \ary{b}_k) = \set{  \sum\limits_{i=1}^{k} x_i\cdot \ary{b}_i , x_i\in \Z }.\] 
By default we work with full-rank lattices unless explicitly mentioned. 

The minimum distance $\lambda_1(\Lattice)$ of a lattice $\Lattice$ is the length (in the $\ell_2$ norm by default) of its shortest nonzero vector: $\lambda_1(\Lattice) = \min_{\ary{x}\in \Lattice\setminus{\set{\ary{0}}}} \|\ary{x}\|$. More generally, the $i^{th}$ successive minimum $\lambda_i(\Lattice)$ is the smallest radius $r$ such that $\Lattice$ contains $i$ linearly independent vectors of norm at most $r$. We write $\lambda_1^p$ as the minimum distance in the $\ell_p$ norm.

For a point $\ary y\in \R^n$, its distance to $\Lattice$ is given by $\text{dist}(\ary y, \Lattice) = \min_{\ary x\in \Lattice}\{\|\ary y - \ary x\|\}$. 
Define ``the ball around lattice" as $B_{\Lattice}(r) = \{\ary x\in \R^n: \text{dist}(\ary x, \Lattice) < r\}$. For $\ary y\in B_{\Lattice}(\lambda_1(\Lattice) / 2)$, the (unique) closest vector to $\ary y$ in $\Lattice$ is given by $\kappa_{\Lattice}(\ary y) = \mathrm{argmin}_{\ary x\in \Lattice}\{\|\ary y - \ary x\|\}$. For convenience, we omit the $\lambda_1(\Lattice) / 2$ term and define $B_{\Lattice}$ as $B_{\Lattice}(\lambda_1(\Lattice) / 2)$, over which $\kappa_{\Lattice}$ is uniquely defined.

The dual of a lattice $\Lattice\in\R^n$ is defined as
\[ \Lattice^*:= \set{  \ary{y}\in \R^n: \ipd{\ary{y}}{\ary{x}}\in\Z \text{ for all } \ary{x}\in \Lattice  }. \]
If $\mat{B}$ is a basis of a full-rank lattice $\Lattice$, then $\mat{B}^{-T}$ is a basis of $\Lattice^*$. The determinant of a full-rank lattice $\Lattice(\mat{B})$ is $\det(\Lattice(\mat{B})) = |\det(\mat{B})|$.

\begin{lemma}[Poisson Summation Formula]\label{lemma:poisson1}
    For any lattice $\Lattice$ and any Schwartz function $f: \R^n \to \C$, we have $f(\Lattice) = \det(\Lattice^*)\hat{f}(\Lattice^*)$. 
\end{lemma}

\paragraph{Gaussians and lattices.}
For any $s > 0$, define the Gaussian function on $\R^n$ with parameter $s$ following the convention in~\cite{MicciancioRegev07}
\[
    \forall \ary{x}\in \R^n, ~\rho_{s}(\ary{x}) = \exp(-\pi \|\ary{x}\|^2/s^2).
\]
For any $\ary{c}\in\R^n$, define $\rho_{s, \ary{c}}(\ary{x}):= \rho_{s}(\ary{x} - \ary{c})$. 
The subscripts $s$ and $\ary{c}$ are taken to be 1 and $\ary{0}$ (respectively) when omitted.
Note that the standard deviation of $\rho_s$ is $s/\sqrt{2\pi}$. 
The Fourier transform for Gaussian satisfies $\hat{\rho_s} = s^n \rho_{1/s}$.
From Poisson summation formula we have
$\rho_s(\Lattice) = s^n\cdot \det(\Lattice^*)\cdot \rho_{1/s}(\Lattice^*)$. 

For a full-rank, symmetric, positive definite $n\times n$ matrix $\mat \Sigma$, define the Gaussian function on $\R^n$ with parameter $\sqrt{\mat \Sigma}$ following the convention in \cite{DBLP:conf/eurocrypt/MicciancioP12}
\[
    \forall \ary{x}\in \R^n, ~
	  \rho_{\sqrt{\mat\Sigma}}(\ary{x}) = \exp(-\pi\cdot \ary{x}^T\mat\Sigma^{-1}\ary x).
\]

\iffull
For any $\ary{c}\in\R^n$, real $s > 0$, and $n$-dimensional lattice $\Lattice$, define the discrete Gaussian distribution $D_{\Lattice+\ary{c}, s}$ as
\[ \forall \ary{x}\in \Lattice+\ary{c}, ~ D_{\Lattice+\ary{c}, s}(\ary{x}) = \frac{\rho_{s}(\ary{x})}{\rho_{s}(\Lattice+\ary{c})}.\]
Similarly, for a full-rank, symmetric, positive definite $n\times n$ matrix $\mat \Sigma$, define the discrete Gaussian distribution $D_{\Lattice+\ary{c}, \sqrt{\mat \Sigma}}$ as
\[ \forall \ary{x}\in \Lattice+\ary{c}, ~ D_{\Lattice+\ary{c}, \sqrt{\mat \Sigma}}(\ary{x}) = \frac{\rho_{\sqrt{\mat \Sigma}}(\ary{x})}{\rho_{\sqrt{\mat \Sigma}}(\Lattice+\ary{c})}.\]
\fi

\iffull
The following Gaussian tail bound over lattices is due to Banaszczyk.
\begin{lemma}[Lemma~1.5~of~\cite{banaszczyk1993new}]\label{lemma:Bana93}
    For any $n$-dimensional lattice $\Lattice$, and $r\geq \frac{1}{\sqrt{2\pi}}$, $\ary{c}\in\R^n$,
    \begin{equation}
        \begin{split}
            &\rho( \Lattice \setminus B^n(r\sqrt{n}) ) < \left( r\sqrt{2\pi e}\cdot e^{-\pi r^2}\right)^n\rho(\Lattice),  \\
            &\rho( (\Lattice -\ary{c}) \setminus  B^n( r\sqrt{n} ) )< 2\left( r\sqrt{2\pi e}\cdot e^{-\pi r^2}\right)^n\rho(\Lattice).
        \end{split}
    \end{equation}
\end{lemma}

As a direct corollary, for $r > \frac{C}{\sqrt{2\pi}}\alpha\sqrt{n}$ with $C > 1$ be a constant, we have that
\[\rho_\alpha( (\Lattice -\ary{c}) \setminus B^n(r) )< 2^{-\Omega(n)}\rho_\alpha(\Lattice).\]

\begin{lemma}[Lemma~2.10~\cite{banaszczyk1995inequalities}]\label{lemma:Bana95}
For any $n$-dimensional lattice $L$, $\ary{c}\in\R^n$, $r>0$, one has 
\[ \rho( (L - \ary{c})\setminus r \cB^n_\infty) < \left( 2n \cdot e^{-\pi r^2}\right) \rho(L).\]
\end{lemma}

\begin{lemma}[Claim~8.1~\cite{DBLP:conf/stoc/0001S17}]\label{lemma:gaussianintsum}
    For any $n\geq 1$, $s>0$,    \[  s^n(1+2e^{-\pi s^2})^n\leq \rho_s(\Z^n)\leq s^n(1+(2+1/s)e^{-\pi s^2})^n. \]
\end{lemma}
\fi

\paragraph{Smoothing parameter. } 
We recall the definition of smoothing parameter for Gaussian over lattices and some useful facts.
\begin{definition}[Smoothing parameter~\cite{MicciancioRegev07}]
    For any lattice $\Lattice$ and positive real $\epsilon > 0$, the smoothing parameter $\eta_\epsilon(\Lattice)$ is the smallest real $s > 0$ such that $\rho_{1/s}(\Lattice^*\setminus\{\ary{0}\}) \leq \epsilon$.
\end{definition}

\iffull
\begin{lemma}[Lemma~2.12~\cite{DBLP:journals/jacm/Regev09}]\label{lemma:smoothingMR07}
    For any $n$-dimensional lattice $\Lattice$, and any real $\epsilon>0$, 
    \[ \eta_\epsilon(\Lattice)\leq \lambda_n(\Lattice) \cdot \sqrt{\ln(2n(1+1/\epsilon))/\pi} . \]
\end{lemma}

\begin{lemma}[Claim~2.13~\cite{DBLP:journals/jacm/Regev09}]\label{lemma:smoothingRegev09}
    For any $n$-dimensional lattice $\Lattice$, and any real $\epsilon>0$, 
    \[ \eta_\epsilon(\Lattice)\ge \sqrt{\frac{\ln 1/\epsilon}{\pi}}\frac{1}{\lambda_1(\Lattice^*)} . \]
\end{lemma}

\begin{lemma}[Claim~3.8~of~\cite{DBLP:journals/jacm/Regev09}]\label{lemma:regevclaim3.8}
    For any $n$-dimensional lattice $\Lattice$, $\ary{c}\in\R^n$, $\epsilon>0$, and $r\geq \eta_\epsilon(\Lattice)$ 
    \[ \rho_{r}(\Lattice+\ary{c})\in r^n\det(\Lattice^*)(1\pm \epsilon) . \]
\end{lemma}
\fi

\paragraph{$q$-ary lattices. }
Given $n < m\in\N$ and a modulus $q\geq 2$, for $\mat{A}\in\Z_q^{n\times m}$, define $q$-ary lattices as
\[
    \begin{split}
    \Lattice_q(\mat{A}) =& \set{ \ary{x}\in\Z^m: \exists \ary{s}\in\Z^n \text{ such that }\ary{x} = \mat{A}^T\cdot \ary{s} + q\Z^m }; \\
    \Lattice_q^\perp(\mat{A}) =& \set{ \ary{x}\in\Z^m: \mat{A}\cdot \ary{x} = \ary{0} \pmod{q} }.
    \end{split}
\]
Those two lattices are dual of each other up to a factor of $q$, i.e., $\Lattice_q(\mat{A}) = q\cdot \Lattice_q^\perp(\mat{A})^*$.

\iffull
\begin{lemma}\label{lemma:linfty}
    Let $q \ge 2, m \ge 2n\log_2 q$, then for all but at most $q^{-0.16n}$ fraction of $\mat{A}\in \Z_q^{n\times m}$, we have
    \[ \lambda_1^\infty(\Lattice_q(\mat{A})) \ge \frac{q}{4}.   \]
\end{lemma}

\begin{proof}
The lemma is proven when $q$ is a prime in~\cite[Lemma~5.3]{DBLP:conf/stoc/GentryPV08}. Here we extend the proof to a general $q$.
 
For any fixed non-zero $\ary{s}\in\Z^n_q$, wlog assuming $s_1$ is a non-zero entry of $\ary{s}$. 
Then for any $\ary{a}\in\Z_q^n$, $y := \ipd{\ary{a}}{\ary{s}}\mod q $ can be written as $y = s_1 a_1 + v\mod q$ for some $v\in\Z_q$. We observe that for any $q\in\N$, for any $v\in\Z_q$, for any non-zero $s_1\in\Z_q$, 
$$  \Pr_{a_1\in\Z_q}[~ s_1 a_1 +v \bmod q \in(-q/4, q/4)\cap \Z ~] \leq 2/3, $$
where the equality holds when $q \in 3^k \cdot\N$ for some $k\geq 1$, $s_1\in (q/3)\cdot \Z/q\Z$, $s_1\neq 0$, and for some $v\in\Z_q$ (for example, when $q = 15$, $s_1 = 5$, and $v = 2$).

Therefore, over the randomness of $\mat{A}\in \Z_q^{n\times m}$, the probability that $\mat{A}^T \ary{s} = \ary{y} \bmod q$ for some $\ary{y}\in \Z^m$ such that $\|\ary{y}\|_{\infty}<q/4$ is at most $(2/3)^m \leq (3/2)^{-2n\log_2 q} \leq q^{-1.16n}$. Applying a union bound over all $\ary{s}\in\Z_q^n$ completes the proof of~\Cref{lemma:linfty}.
\end{proof}
\paragraph{Lattice problems.}

We have formally defined the LWE, $\QLWE$, $\LWEstate$, and $\QLWE^{\sf phase}$ problems in the introduction. Now let us recall the definitions for other lattice problems.

The shortest vector problem (SVP) asks to find a lattice vector of length $\lambda_1$. 
More generally, let $\gamma(n)\geq 1$ be an approximation factor, we consider the approximation version of SVP and its close variants. 

\begin{definition}[Approximate SVP]
    Given a basis $\mat{B}$ of an $n$-dimensional lattice $\Lattice$, the $\SVP_\gamma$ problem asks to output a non-zero lattice vector $\mat{B}\ary{x}$, $\ary{x}\in\Z^n\setminus{\set{\ary{0}}}$, such that $\|\mat{B}\ary{x}\|\leq \gamma(n)\cdot \lambda_1(\Lattice)$.
\end{definition}

\begin{definition}[GapSVP]
    Given a basis $\mat{B}$ of an $n$-dimensional lattice $\Lattice$ and a number $d>0$, the $\GAP\SVP_\gamma$ problem asks to decide whether $\lambda_1(\Lattice)\leq d$ or $\lambda_1(\Lattice)> d\cdot\gamma(n)$.
\end{definition}

\begin{definition}[Shortest independent vector problem (SIVP)]
    Given a basis $\mat{B}$ of an $n$-dimensional lattice $\Lattice$, the $\SIVP_\gamma$ problem asks to output a set of $n$ linearly independent vectors of length at most $\gamma(n)\cdot \lambda_n(\Lattice)$.
\end{definition}

\begin{definition}[Discrete Gaussian Sampling Problem (DGS)]\label{def:DGS}
    Given a basis $\mat{B}$ of an $n$-dimensional lattice $\Lattice$ and a parameter $s>0$, the $\DGS_s$ problem asks to output a vector whose distribution is statistically close to $D_{\Lattice, s}$.  
\end{definition}

\begin{definition}[Quantum Discrete Gaussian Sampling Problem ($\QDGS$)]\label{def:QDGS}
    Given a basis $\mat{B}$ of an $n$-dimensional lattice $\Lattice$ and a parameter $s>0$, the $\QDGS_s$ problem asks to output a quantum state that is $2^{-\Omega(n)}$-close to $\sum_{\ary{v}\in \Lattice} \rho_{s}(\ary{v})\ket{\ary{v}}$ in trace distance.
\end{definition}

If a quantum algorithm solves $\QDGS_s$, then it immediately solves $\DGS_{s/\sqrt{2}}$ by simply measuring the quantum state. Let us also recall the relationships among $\DGS$,  $\GAP\SVP$, and $\SIVP$. 

\begin{lemma}[Lemma~3.20 of \cite{DBLP:journals/jacm/Regev09}]\label{lemma:GapSVP2DGS}
    For any $\gamma = \gamma(n)\geq 1$, there exists a polynomial time reduction from $\GAP\SVP_{100\sqrt{n}\gamma}$ for $\Lattice$ to $\DGS_{\sqrt{n}\gamma/\lambda_1(\Lattice^*)}$ for $\Lattice^*$. 
\end{lemma}

\begin{lemma}[Lemma~3.17 of \cite{DBLP:journals/jacm/Regev09}]\label{lemma:SIVP2DGS}
    For any $\gamma>\omega(\sqrt{\log n})$, there exists a polynomial time reduction from $\SIVP_{2\sqrt{n}\gamma}$ for $\Lattice$ to $\DGS_{\gamma\lambda_n(\Lattice)}$ for $\Lattice$. 
\end{lemma}
\fi

\subsection{Quantum computation}

We assume readers are familiar with the basic concepts of quantum computation. All the quantum background we need in this paper are available in standard textbooks of quantum computation, e.g., \cite{DBLP:books/daglib/0046438}. When writing a quantum state such as $\sum_{x \in S } f(x)\ket{x}$, we typically omit the normalization factor except when needed. 

The trace distance between two quantum states $\rho$ and $\sigma$ is defined as $\td(\rho, \sigma) := \frac{1}{2}\tr|\rho-\sigma|$.
Note that when $\rho$ and $\sigma$ commute they are diagonal in the same basis,
\iffull 
\[ \rho = \sum_{i} r_i \ket{i}\bra{i}, ~~~ \sigma = \sum_{i} s_i \ket{i}\bra{i},   \]
\fi
for some orthonormal basis $\ket{i}$, then 
$\td(\rho, \sigma) = \frac{1}{2}\tr\left|\sum_i(r_i - s_i) \ket{i}\bra{i}\right| = \frac{1}{2}\sum_i|r_i - s_i|$.

The trace distance is preserved under unitary transformations, and is contractive under trace-preserving operations.

\begin{lemma}\label{lemma:difftotd}
    Let $\ket{\phi}$, $\ket{\psi}$ be un-normalized vectors s.t. $\|\ket{\phi}\|\geq \mu$ and $\|\ket{\phi} - \ket{\psi}\|\leq \epsilon$. Then 
    \[ \td\left( \frac{1}{\|\ket{\phi}\|}\ket{\phi}, \frac{1}{\|\ket{\psi}\|}\ket{\psi} \right) = \sqrt{1 - \left( \frac{|\bk{\phi}{\psi}|}{\|\ket{\phi}\|\|\ket{\psi}\|} \right)^2}\leq O\left(\sqrt{\frac{\epsilon}{\mu}}\right) .\]
\end{lemma}

\begin{lemma}[Gentle measurement~\cite{Winter_1999}]\label{lemma:gentle_measurement}
Let $\rho$ be a quantum state and let $(\mat \Pi, \mat I - \mat \Pi)$ be a two-outcome projective measurement such that $\tr\left(\mat \Pi \rho\right) \ge 1 - \epsilon$. Let $\rho' = \frac{\mat \Pi \rho \mat \Pi}{\tr\left(\mat \Pi \rho\right)}$ be the state after applying the measurement and post-selecting on getting the first outcome. Then $\td(\rho, \rho') \le 2\sqrt{\epsilon}$. 
\end{lemma}

\iffull
We need the following lemma about the trace distance between discrete Gaussian states.
\begin{lemma}\label{lemma:dist_between_gaussian_state}
When $q > 2\sqrt{n}\max(\beta_1, \beta_2)$ and $R \ge \frac{2\sqrt{n}}{\min(\beta_1, \beta_2)}$, the trace distance between \[\ket{\phi_1} = \sum_{e \in \Z_{qR}/R}\rho_{\beta_1}(e)\ket{e} \qquad \text{and} \qquad \ket{\phi_2} = \sum_{e \in \Z_{qR}/R}\rho_{\beta_2}(e)\ket{e}\] is at most $\sqrt{\frac{(\beta_1 - \beta_2)^2}{\beta_1^2 + \beta_2^2}}(1 + 2^{-\Omega(n)})$.
\end{lemma}

\begin{proof}
\iffull
\begin{align*}
\frac{\bk{\phi_1}{\phi_2}}{\|\ket{\phi_1}\|\|\ket{\phi_2}\|} =& \sum_{e \in \Z_{qR}/R}\rho_{\frac{\beta_1\beta_2}{\sqrt{\beta_1^2 + \beta_2^2}}}(e) \left/ \sqrt{\sum_{e \in \Z_{qR}/R}\rho_{\beta_1/\sqrt{2}}(e)\sum_{e \in \Z_{qR}/R}\rho_{\beta_2/\sqrt{2}}(e)} \right. \\
=& \sum_{e \in \Z/R}\rho_{\frac{\beta_1\beta_2}{\sqrt{\beta_1^2 + \beta_2^2}}}(e) \left/ \sqrt{\sum_{e \in \Z/R}\rho_{\beta_1/\sqrt{2}}(e)\sum_{e \in \Z/R}\rho_{\beta_2/\sqrt{2}}(e)}\right. (1 + 2^{-\Omega(n)}) \\
=& \frac{\beta_1\beta_2}{\sqrt{\beta_1^2 + \beta_2^2}} \cdot \frac{\sqrt{2}}{\sqrt{\beta_1\beta_2}} \sum_{e \in R\Z}\rho_{\frac{\sqrt{\beta_1^2 + \beta_2^2}}{\beta_1\beta_2}}(e) \left/ \sqrt{\sum_{e \in R\Z}\rho_{\sqrt{2}/\beta_1}(e)\sum_{e \in R\Z}\rho_{\sqrt{2}/\beta_2}(e)}\right. (1 + 2^{-\Omega(n)})\\
=& \frac{\sqrt{2\beta_1\beta_2}}{\sqrt{\beta_1^2 + \beta_2^2}}(1 + 2^{-\Omega(n)}),
\end{align*}
\fi

\ifllncs
\begin{align*}
&~\frac{\bk{\phi_1}{\phi_2}}{\|\ket{\phi_1}\|\|\ket{\phi_2}\|} \\
=&~ \sum_{e \in \Z_{qR}/R}\rho_{\frac{\beta_1\beta_2}{\sqrt{\beta_1^2 + \beta_2^2}}}(e) \left/ \sqrt{\sum_{e \in \Z_{qR}/R}\rho_{\beta_1/\sqrt{2}}(e)\sum_{e \in \Z_{qR}/R}\rho_{\beta_2/\sqrt{2}}(e)} \right. \\
=&~ \sum_{e \in \Z/R}\rho_{\frac{\beta_1\beta_2}{\sqrt{\beta_1^2 + \beta_2^2}}}(e) \left/ \sqrt{\sum_{e \in \Z/R}\rho_{\beta_1/\sqrt{2}}(e)\sum_{e \in \Z/R}\rho_{\beta_2/\sqrt{2}}(e)}\right. (1 + 2^{-\Omega(n)}) \\
=&~ \frac{\sqrt{2\beta_1\beta_2}}{\sqrt{\beta_1^2 + \beta_2^2}}\sum_{e \in R\Z}\rho_{\frac{\sqrt{\beta_1^2 + \beta_2^2}}{\beta_1\beta_2}}(e) \left/ \sqrt{\sum_{e \in R\Z}\rho_{\sqrt{2}/\beta_1}(e)\sum_{e \in R\Z}\rho_{\sqrt{2}/\beta_2}(e)}\right. (1 + 2^{-\Omega(n)})\\
=&~ \frac{\sqrt{2\beta_1\beta_2}}{\sqrt{\beta_1^2 + \beta_2^2}}(1 + 2^{-\Omega(n)}),
\end{align*}
\fi
where we use the Poisson summation formula and Banaszczyk's tail bound.

So their trace distance is at most
\[
\sqrt{1 - \left(\frac{\left|\bk{\phi_1}{\phi_2}\right|}{\|\ket{\phi_1}\|\|\ket{\phi_2}\|}\right)^2} \le \sqrt{\frac{\left(\beta_1 - \beta_2\right)^2}{\beta_1^2 + \beta_2^2}}(1 + 2^{-\Omega(n)}). \qedhere
\]
\end{proof}
\fi

We use the following quantum algorithms:

\begin{lemma}[Quantum Fourier Transform (QFT)~\cite{DBLP:Kitaev95}]
	Let $q\geq 2$ be an integer. The following unitary operator $\QFT_q$ can be implemented by $\poly(\log q)$ elementary quantum gates. When $\QFT_q$ is applied on a quantum state $\ket{\phi} := \sum_{x\in\Z_q} f(x)\ket{x}$, we have  
    \iffull
	\[\QFT_q \ket{\phi} = \sum_{y\in\Z_q} \sum_{x\in\Z_q} \frac{1}{\sqrt{q}}\cdot \omega_q^{xy}\cdot f(x)  \ket{y}. \]
    \fi
\end{lemma}

We use the following lemma~\cite{ozols2013quantum} to change the amplitude of a state.
\begin{lemma}[Quantum rejection sampling]\label{lemma:quantumstatepre}
Let $f: D\to \C$ be a normalized amplitude of some quantum state $\ket{\phi_f}:=\sum_{x\in D} f(x)\ket{x}$. 
Let $\gamma: D\to [0,1]$ be a polynomial time computable function. 
There is a quantum algorithm that takes as input $\ket{\phi_f}$, outputs a state $\sum_{x\in D} \frac{1}{\sqrt{M}}\gamma(x) f(x) \ket{x}$ with probability $M$ where $M = \sum_{x\in D}\gamma^2(x) |f(x)|^2$.
\end{lemma}

In this paper we are interested in preparing the Gaussian state $\ket{ \sigma_{n, R} }:= \sum_{\ary{x}\in\Z^n\cap B^n(R\sqrt{n}) } \rho_R(\ary{x}) \ket{\ary{x}}$ for some radius $R\leq 2^{\poly(n)}$. 
Given \Cref{lemma:Bana93}, there is a $2^{-\Omega(n)}$ mass in the tail of $\rho_R(\ary{x})$ outside $B^n(R\sqrt{n})$. 
This means for any integer $P\in (2R\sqrt{n}, 2^{\poly(n)})$, $\ket{ \sigma_{n, R} }$ is $2^{-\Omega(n)}$ close to $\sum_{\ary{x}\in\Z_P^n } \rho_R(\ary{x}) \ket{\ary{x}}$. 
This also means we can prepare $\ket{ \sigma_{n, R} }$ by generating $n$ independent samples of one-dimensional Gaussian state $\ket{ \sigma_{1, R} }$, which can be done efficiently using~\cite{grover2002creating}.

\begin{lemma}[Gaussian state preparation]\label{lemma:Gaussianstateprep}
    Let $n, R\in \N$ such that $1\leq R\leq 2^{n^c}$ for some constant $c\geq 0$. 
    Then we can create a $\poly(n)$ size unitary $U$ that maps $\ket{\ary{0}}$ to a state within trace distance $2^{-\Omega(n)}$ from $\ket{\sigma_{n, R}}$ with $2n\upperrounding{n^c\cdot\log n}$ qubits. 
\end{lemma}

\section{Quantum Sub-exponential Time Algorithm for $\QLWE$}\label{sec:kuper}

In this section, we provide a quantum sub-exponential time algorithm designed to solve $\QLWE$ instances with specific amplitudes. More precisely, we consider scenarios where the discrete Fourier transform of these amplitudes has at least $2^{-\sqrt{n}\log q}$ mass on two distinct points. Our approach is built upon two key steps: (1) generate a DCP state for each quantum $\QLWE$ sample, resulting in a sub-exponential collection of DCP states; (2) employ the Kuperberg sieve technique \cite{DBLP:journals/siamcomp/Kuperberg05} on the DCP states to successfully recover the secret vector. Formally, we state the main theorem of this section as follows:

\begin{theorem}[Main theorem]\label{thm:subexp_alg_gen}
    For any efficiently computable normalized amplitude function $f:\Z\to \C$ such that there exists two distinct points $j_1$ and $j_2$ from $\Z_q$ with ${\rm gcd} (j_1 - j_2, q) = 1$ and $|\DFT_q(f)(j_1)|$ and $|\DFT_q(f)(j_2)|$ are both greater than $2^{-\sqrt{n}\log q}$ ($\DFT_q(f)$ is the discrete Fourier transform of $f$, defined as $\DFT_q(f)(j) = \frac{1}{\sqrt q}\sum_{e\in \Z} f(e)\omega_q^{je}$ for $j\in \Z_q$), there exists a quantum algorithm that, given $\ell = 2^{\Theta(\sqrt{n}\log q)}$ samples of vector $\ary a\gets \mathcal{U}(\Z_q^n)$ and $\QLWE$ state of form
    \[
        \QLWE = \sum_{e \in \Z} f(e) \ket{\ipd{\ary a}{\ary s} + e \bmod q},
    \]
    finds the secret vector $\ary s\in \Z_q^n$ within a time complexity of $2^{\Theta(\sqrt{n}\log q)}$.
\end{theorem}

Note that when $q$ is prime, then ${\rm gcd} (j_1 - j_2, q) = 1$ is always satisfied. So the condition ${\rm gcd} (j_1 - j_2, q) = 1$ is only needed to be taken care of when $q$ is composite.

To prove \Cref{thm:subexp_alg_gen}, we begin by introducing the Kuperberg sieve algorithm \cite{DBLP:journals/siamcomp/Kuperberg05}.

\begin{lemma}[Kuperberg sieve]\label{thm:kuperberg_sieve}
    Let $\ary s\in \Z_q^n$ be a secret vector. There exists a quantum algorithm that given $\ell^* = 2^{\Theta(\sqrt{n}\log q)}$ samples of 
    \[
        \ary a\gets \mathcal{U}(\Z_q^n), \quad
        \ket{\psi_{\ary a}} = \ket{0} + \omega_q^{\ipd{\ary a}{\ary s}}\ket{1},
    \]
    finds out the secret vector $\ary s$ in time $2^{\Theta(\sqrt{n}\log q)}$.
\end{lemma}

Now we present the proof of \Cref{thm:subexp_alg_gen} here.

\begin{proof}[Proof of \Cref{thm:subexp_alg_gen}]
    Suppose we have $\ell = \ell^* \cdot 2^{4\sqrt{n}\log q} = 2^{\Theta(\sqrt{n}\log q)}$ instances of $\QLWE$ states. Our quantum algorithm proceeds as follows:
    \begin{enumerate}
        \item[1.] Compute $\DFT_q(f)(j)$ for each $j \in \Z_q$ and find two distinct points $j_1$ and $j_2$ from $\Z_q$ with ${\rm gcd} (j_1 - j_2, q) = 1$ and $|\DFT_q(f)(j_1)|$ and $|\DFT_q(f)(j_2)|$ are both greater than $2^{-\sqrt{n}\log q}$. If it fails, abort.
        \item[2.] Apply QFT to any $\QLWE$ state, resulting in the state
        \[
            \QFT_q\cdot\QLWE = \frac{1}{\sqrt q}\sum_{j\in \Z_q}\sum_{e \in \Z} f(e) \omega_q^{j(\ipd{\ary a}{\ary s} + e)} \ket{j} = \sum_{j\in \Z_q}\omega_q^{\ipd{j\cdot \ary a}{\ary s}}\DFT_q(f)(j)\ket{j},
        \]
        \item[3.] 
        Define $\gamma(j): \Z_q\to [0, 1]$ as
        \[
            \gamma(j) = \frac{\min\{|\DFT_q(f)(j_1)|, |\DFT_q(f)(j_2)|\}}{|\DFT_q(f)(j)|} \text{ for } j = j_1, j_2
        \]
        and $\gamma(j) = 0$ otherwise, apply quantum rejection sampling (\Cref{lemma:quantumstatepre}) to obtain the state
        \begin{equation}\label{eqn:sub_exp_temp1}
            \frac{\DFT_q(f)(j_1)}{|\DFT_q(f)(j_1)|}\omega_q^{\ipd{j_1 \cdot \ary a}{\ary s}}\ket{j_1} + \frac{\DFT_q(f)(j_2)}{|\DFT_q(f)(j_2)|}\omega_q^{\ipd{j_2\cdot \ary a}{\ary s}}\ket{j_2},
        \end{equation}
        with probability
        \[
            M = \sum_{j\in \Z_q}\gamma^2(j)|\DFT(f)(j)|^2 = 2(\min\{|\DFT_q(f)(j_1)|, |\DFT_q(f)(j_2)|\})^2 > 2^{-2\sqrt{n}\log q}.
        \]
        \item[4.] For the states that have been successfully transformed to \Cref{eqn:sub_exp_temp1}, apply a unitary operation
        \[
            U: \ket{j_1} \to \frac{\overline{\DFT_q(f)(j_1)}}{|\DFT_q(f)(j_1)|}\ket{0}, \ket{j_2} \to \frac{\overline{\DFT_q(f)(j_2)}}{|\DFT_q(f)(j_2)|}\ket{1},
        \]
        this results in the state
        \[
            \ket{\psi_{(j_2 - j_1)\ary a}} = \ket{0} + \omega_q^{\ipd{(j_2 - j_1)\ary a}{\ary s}}\ket{1},
        \]
        where $(j_2 - j_1)\ary a$ is a known vector in $\Z_q^n$ that is uniformly random by the assumption that ${\rm gcd}(j_1 - j_2, q) = 1$.
        \item[5.] Select $\ell^*$ such states obtained in step 3 and apply the Kuperberg sieve algorithm to recover the secret vector $\ary s\in \Z_q^n$.
    \end{enumerate}
    It is evident that step 1 runs in time $O(\poly(n) q \log q) \le 2^{O(\sqrt{n}\log q)} $ and transforming any $\QLWE$ state to a DCP-like state requires a time complexity of $2^{O(\sqrt{n}\log q)}$. Therefore, the run time of our quantum algorithm is constrained by both the quantity of $\QLWE$ states and the application of the Kuperberg sieve, which both exhibit a complexity of $2^{\Theta(\sqrt{n}\log q)}$. In step 2, the count of states successfully transformed to \Cref{eqn:sub_exp_temp1} will be at least $M^2\ell = \ell^*$ with a probability exponentially close to 1. This concludes the proof.
\end{proof}

\begin{corollary}\label{cor:subexp_alg}
    Suppose $m, n, q$ are LWE parameters. There exists a quantum algorithm that, given $2^{\Theta(\sqrt{n}\log q)}$ samples of vector $\ary a\gets \mathcal{U}(\Z_q^n)$ and $\QLWE$ state of form
    \[
        \QLWE = \sum_{e \in \Z} \rho_{\sigma}(e) \exp(2\pi {\rm i}\cdot ce /q) \ket{\ipd{\ary a}{\ary s} + e \bmod q},
    \]
    where the Gaussian width $\sigma$ satisfies $\sigma = \Omega(\sqrt{n}), \sigma \le q$, and $c$ is an arbitrary known number that can be different for different samples, solves the secret vector $\ary s\in \Z_q^n$ within a time complexity of $2^{\Theta(\sqrt{n}\log q)}$.
\end{corollary}

\iffull
\begin{proof}
    Let us define
    \[
        N = \sum_{e\in \Z}\rho_{\sigma}^2(e) = \sum_{e\in \Z}\frac{\sigma}{\sqrt{2}}\rho_{\sqrt{2}/\sigma}(e) \approx \frac{\sigma}{\sqrt{2}},
    \]
    where the final approximation holds under the assumption that $\sigma = \Omega(n)$. In this case, the summation of $\rho_{\sqrt{2}/\sigma}$ is concentrated at $\rho_{\sqrt{2}/\sigma}(0)$, with exponentially small weight elsewhere.
    
    In the given problem scenario,
    \[
        f(e) = \frac{1}{\sqrt{N}}\rho_\sigma(e)\exp(2\pi {\rm i}\cdot ce /q),
    \]
    thus
    \[
        \begin{split}
            \DFT_q(f)(j) & = \frac{1}{\sqrt{qN}}\sum_{e\in \Z}\rho_\sigma(e)\exp(2\pi {\rm i}\cdot (j + c)e/q) \\
            & =_{(*)} \frac{1}{\sqrt{qN}}\sum_{e \in \Z} \sigma\rho_{1 / \sigma} \left(e - \frac{j + c}{q}\right) \\
            & \approx_{(**)} \frac{1}{\sqrt{qN}}\sigma\rho_{1 / \sigma} \left(\rd{\frac{j + c}{q}} - \frac{j + c}{q}\right) \\
            & = \frac{\sigma}{\sqrt{qN}} \rho_{q / \sigma}(j + c - \rd{j + c}_q),
        \end{split}
    \]
    here $(*)$ is from the Poisson summation formula, $(**)$ holds due to the assumption that $\sigma = \Omega(\sqrt{n})$. In this case, the summation $\sum_{j\in \Z_q}\rho_{1 / \sigma}\left(e - \frac{j + c}{q}\right)$ is concentrated at $\rho_{1 / \sigma}\left(\rd{\frac{j + c}{q}} - \frac{j + c}{q}\right)$, with exponentially small weight elsewhere.

    Define $j_1 = \lowerrounding{-c} \bmod q, j_2 = (\lowerrounding{-c} + 1) \bmod q$, we can establish that $|j + c - \rd{j + c}_q| \le 1$ holds for both $j = j_1$ and $j = j_2$. This implies that
    \[
        |\DFT_q(f)(j)| \ge \sqrt{\frac{\sqrt{2\sigma}}{q}} \rho_{q / \sigma}(1)(1 - 2^{-\Omega(n)}) \ge \sqrt{\frac{\sqrt{2}}{q}} e^{-\pi} (1 - 2^{-\Omega(n)}) \gg 2^{-\sqrt{n}\log q},
    \]
    for both $j = j_1, j_2$.
    
    As a result, we can deduce the validity of the original statement by straightforwardly applying \Cref{thm:subexp_alg_gen}.
\end{proof}
\fi

\begin{remark}
    Readers may be curious about why our sub-exponential algorithm cannot handle $\QLWE^{\sf phase}$ instances with an unknown phase. Informally speaking, when the phase term of $\QLWE^{\sf phase}$ samples is unknown, we can only obtain a DCP-like state with varying weights in the superposition. Although the ratio of weights on $\ket{0}$ and $\ket{1}$ can be bounded by an inverse polynomial, this ratio tends to become extremely large during the sieving step of Kuperberg's algorithm. As a result, the final state collapses into either $\ket{0}$ or $\ket{1}$, and the information about $\ary s$ is entirely lost.
\end{remark}

\section{Complex Gaussian and Oblivious LWE Sampling}\label{sec:osampling}

In this section, we provide a quantum polynomial time algorithm that can obliviously sample LWE instances. Our approach follows the reduction from oblivious LWE sampling to the $\LWEstate$ problem introduced in \cite{debris2024quantum}. However, our method for solving the $\LWEstate$ problem is fundamentally different from the approach in \cite{debris2024quantum}; it employs complex Gaussian amplitudes, a technique first introduced in \cite{chen2024quantum}, along with a center-finding trick. Not only does our algorithm improve upon the sample complexity of the core quantum algorithm in \cite{debris2024quantum}, but it is also relatively straightforward and easy to comprehend.

\subsection{Center finding}

We first show an efficient quantum algorithm for finding the center $c \bmod t$ of a complex Gaussian state $\ket{\phi_c} = \sum_{x\in \Z}\rho_r(x)e^{-\frac{\pi i x^2}{t}}\ket{x + c}$.
Note that although the state $\ket{\phi_c}$ appears to be supported over $\Z$, it is within $2^{-\Omega(n)}$ trace distance from $\sum_{x\in [c-r\sqrt{n}, c+r\sqrt{n}]\cap \Z}\rho_r(x)e^{-\frac{\pi i x^2}{t}}\ket{x + c}$, so it can be stored within $O(\log |c| + \log(r\sqrt{n}))$ many qubits.

\begin{theorem}[Center finding]\label{thm:find_center}
    Let $n, t$ be positive integers, $r$ be a positive real with $r \ge 30 t n \log n$. Then there exists a polynomial time quantum algorithm that, given a complex Gaussian state
    \[
        \ket{\phi_c} = \sum_{x\in \Z}\rho_r(x)e^{-\frac{\pi i x^2}{t}}\ket{x + c}
    \]
    with unknown center $c\in \Z$, finds $c \bmod t$ with probability $1 - 1/n$.
\end{theorem}

Before proving the theorem, we first observe that for any $c_1, c_2\in \Z$ with $c_1 \not \equiv c_2 \bmod t$, the complex Gaussian states $\ket{ \phi_{c_1} }$, $\ket{ \phi_{c_2} }$ are nearly orthogonal and, therefore, almost perfectly distinguishable. We prove this for the special case of $\ket{ \phi_{0} }$ v.s. $\ket{ \phi_{c} }$, and the proof generalizes to any pair of distinct centers $\ket{ \phi_{c_1} }$, $\ket{ \phi_{c_2} }$:
\begin{equation}
    \begin{split}
    \frac{\abs{\bk{ \phi_0 }{ \phi_c }}}{\|\ket{\phi_0}\|\|\ket{\phi_c}\|}
    & = \abs{\sum_{x\in\Z} \rho_r(x)\rho_r(x-c) e^{-\frac{\pi i (x - c)^2 - \pi i x^2}{t}}} \left/ \sum_{x\in\Z} \rho_r(x)^2 \right. \\
    & = \rho_{r / \sqrt{2}}(c/2)\abs{ \sum_{x\in\Z} \rho_{r/\sqrt{2}}(x-c/2) e^{\frac{2\pi i cx}{t}}} \left/ \sum_{x\in\Z} \rho_{r / \sqrt{2}}(x) \right. \\
    & =_{(*)} \rho_{r/\sqrt{2}}(c/2) \abs{ \sum_{y\in\Z}\rho_{\sqrt{2}/r}(y - c/t)e^{2\pi i(y - c/t)\cdot c/2}} \left/ \sum_{y\in\Z} \rho_{\sqrt{2} / r}(y) \right. \\
    & \approx_{(**)} \rho_{r/\sqrt{2}}(c/2) \rho_{\sqrt{2}/r}(\rd{c/t} - c/t),
    \end{split}
\end{equation}
where $(*)$ is from the Poisson summation formula, $(**)$ holds due to the assumption that $r/t \in \Omega(\sqrt{n})$. In this case, the summation $\sum_{y\in\Z}\rho_{\sqrt{2}/r}(y - c/t)$ (with some phase term that does not affect the norm) is concentrated at $\rho_{\sqrt{2}/r}(\rd{c/t} - c/t)$, and the summation $\sum_{y\in\Z}\rho_{\sqrt{2}/r}(y)$ is concentrated at $\rho_{\sqrt{2}/r}(0) = 1$, with exponentially small weight elsewhere.

Therefore, when $c\not \equiv 0\bmod t$, the overlap of $\ket{\phi_0}$ and $\ket{\phi_c}$ can be bounded by
\[
    \left|\frac{\bk{ \phi_0 }{ \phi_c }}{\|\ket{\phi_0}\|\|\ket{\phi_c}\|}\right| \le \rho_{r/\sqrt2}(c/2)\rho_{\sqrt2/r}(1/t) \left(1 + 2^{-\Omega(n)}\right) = 2^{-\Omega(n)}.
\]

Using a similar argument, for $c_1\not\equiv c_2\bmod t$, the overlap of $\ket{\phi_{c_1}}$ and $\ket{\phi_{c_2}}$ is bounded by $2^{-\Omega(n)}$. 
Keeping the above observation in mind, the algorithm of center finding is straightforward: measure the lower order bits (for modulus $t$) of $\ket{\phi_c}$ under the basis $\left\{\ket{\psi_d}\right\}_{d\in \Z_{t}}$, and output the measurement result $d\in \Z_{t}$, where
\[
    \ket{\psi_d} = \sum_{x\in \Z_{t}} e^{-\frac{\pi i(x - d)^2}{t}}\ket{x}.
\]
We'll first show that $\left\{\ket{\psi_d}\right\}_{d\in \Z_{t}}$ is a basis that can be efficiently prepared, and then show the measurement result $d$ equals the center $c$ with high probability. 

For $d, d'\in \Z_{t}$, the overlap of $\ket{\psi_d}$ and $\ket{\psi_{d'}}$ is
\[
    \begin{split}
        \frac{\bk{\psi_{d'}}{\psi_d}}{\|\ket{\psi_{d'}}\|\|\ket{\psi_d}\|} & = \frac{1}{t}\sum_{x\in \Z_{t}} e^{-\left(\frac{\pi i(x - d)^2}{t} - \frac{\pi i (x - d')^2}{t}\right)} \\
        & = \frac{1}{t}e^{\frac{\pi i (d'^2 - d^2)}{t}} \sum_{x\in \Z_{t}}e^{\frac{2\pi i(d - d')x}{t}} \\
        & = \begin{cases}
            1 & d=d'  \\
            0 & d\neq d'
        \end{cases}.
    \end{split}
\]

Therefore $\{\ket{\psi_d}\}_{d\in\Z_t}$ is indeed a basis. To show $\left\{\ket{\psi_d}\right\}_{d\in \Z_{t}}$ can be efficiently prepared, we use the following lemma:
\begin{lemma}
There is a $\poly(\log t)$ size quantum circuit that implements the unitary transformation $\sum_{d\in\Z_t} \ket{d} \bra{\psi_d}$. 
\end{lemma}
\begin{proof}
For any $d\in \Z_t$, we transform $\ket{d} \to \ket{\psi_d} = \sum_{x\in \Z_{t}} e^{-\frac{\pi i(x - d)^2}{t}}\ket{x}$ as follows:
\iffull
\begin{equation*}
 \ket{d} \mapsto_{(1)} e^{-\frac{\pi i d^2}{t} } \ket{d} 
 \mapsto_{QFT_{\Z_t}} \sum_{x\in\Z_t} e^{\frac{2\pi i xd}{t}} e^{-\frac{\pi i d^2}{t}} \ket{x} 
 \mapsto_{(2)}  \sum_{x\in\Z_t} e^{-\frac{\pi i x^2}{t}} e^{\frac{2\pi i xd}{t}} e^{-\frac{\pi i d^2}{t}} \ket{x} = \sum_{x\in \Z_{t}} e^{-\frac{\pi i(x - d)^2}{t}}\ket{x}, 
\end{equation*}
where (1), (2) use the phase kick-back trick.
\else
\begin{align*}
 \ket{d} &\mapsto_{(1)} e^{-\frac{\pi i d^2}{t} } \ket{d} 
 \mapsto_{QFT_{\Z_t}} \sum_{x\in\Z_t} e^{\frac{2\pi i xd}{t}} e^{-\frac{\pi i d^2}{t}} \ket{x}\\
 &\mapsto_{(2)}  \sum_{x\in\Z_t} e^{-\frac{\pi i x^2}{t}} e^{\frac{2\pi i xd}{t}} e^{-\frac{\pi i d^2}{t}} \ket{x} = \sum_{x\in \Z_{t}} e^{-\frac{\pi i(x - d)^2}{t}}\ket{x}, 
\end{align*}
where (1), (2) use the phase kick-back trick.
\fi
\end{proof}

We then calculate the probability of obtaining $c_1 := c \bmod t$ when measuring the lower order bits of $\ket{\phi_c}$ under the basis $\left\{\ket{\psi_d}\right\}_{d\in \Z_{t}}$. We write the term $x + c$ in the summation of $\ket{\phi_c}$ by $kt + y$, with summation on $k\in \Z$ and $y\in \Z_{t}$:
\[\ket{\phi_c} = \sum_{k \in \Z} \sum_{y \in \Z_t}\rho_r(kt + y - c)e^{-\frac{\pi i (kt + y - c)^2}{t}}\ket{k}_{\reg{A}}\ket{y}_{\reg{B}}\]
where registers $\reg{A}$ and $\reg{B}$ store the higher order bits and the lower order bits of $x + c$, respectively.

We observe that $\ket{\phi_c}$ can also be written in the following form:
\begin{align*}
\ket{\phi_c} &= \sum_{k \in \Z} \sum_{y \in \Z_t}\rho_r(kt + y - c)e^{-\frac{\pi i (kt + y - c)^2}{t} + \frac{\pi i (y - c_1)^2}{t}} \cdot e^{-\frac{\pi i (y - c_1)^2}{t}} \ket{k}_{\reg{A}}\ket{y}_{\reg{B}}\\
&= \sum_{k \in \Z} \sum_{y \in \Z_t}\rho_r(kt + y - c)e^{-\frac{\pi i \left(k^2t^2 + (y - c)^2\right)}{t} + \frac{\pi i (y - c_1)^2}{t}} \cdot e^{-\frac{\pi i (y - c_1)^2}{t}} \ket{k}_{\reg{A}}\ket{y}_{\reg{B}}\\
&= \sum_{k \in \Z}e^{-\pi i k^2t} \sum_{y \in \Z_t}\rho_r(kt + y - c)e^{2\pi i \frac{(c - c_1)y}{t} + \frac{\pi i (c_1^2 - c^2)}{t}} \cdot e^{-\frac{\pi i (y - c_1)^2}{t}} \ket{k}_{\reg{A}}\ket{y}_{\reg{B}}\\
&= \sum_{k \in \Z}e^{-\pi i k^2t + \frac{\pi i (c_1^2 - c^2)}{t}} \ket{k}_{\reg{A}}\sum_{y \in \Z_t}\rho_r(kt + y - c) e^{-\frac{\pi i (y - c_1)^2}{t}} \ket{y}_{\reg{B}}
\end{align*}

The probability of obtaining $c_1$ when measuring $\reg{B}$ in basis $\{\ket{\psi_d}\}_{d \in \Z_t}$ can be computed as first measuring $\reg{A}$ in computational basis and getting some outcome $k$, and then obtaining $c_1$ when measuring $\reg{B}$ in basis $\{\ket{\psi_d}\}_{d \in \Z_t}$:
\[
    \begin{split}
        \Pr(d = c_1) 
        & = \sum_{k\in \Z} \bigg(\sum_{y\in \Z_{t}} \rho_r(kt + y - c)\bigg)^2 \left/ t\sum_{x\in \Z}\rho_r(x)^2 \right. \\
        & = \sum_{k\in \Z}\frac{\sum_{y\in \Z_{t}} \rho_r(kt + y - c)^2}{\sum_{x\in \Z}\rho_r(x)^2} \cdot \frac{\left(\sum_{y\in \Z_{t}} \rho_r(kt + y - c)\right)^2}{t\sum_{y\in \Z_{t}} \rho_r(kt + y - c)^2}\\
        & =_{(1)} \sum_{k\in \Z}\frac{\sum_{y\in \Z_{t}} \rho_r(kt + y - c_1)^2}{\sum_{x\in \Z}\rho_r(x)^2} \cdot \frac{\left(\sum_{y\in \Z_{t}} \rho_r(kt + y - c_1)\right)^2}{t\sum_{y\in \Z_{t}} \rho_r(kt + y - c_1)^2}\\
        & \ge \sum_{-\frac{r\log n}{t} \leq k \leq \frac{r \log n}{t}, k\in \Z }\frac{\left(\sum_{y\in \Z_{t}} \rho_r(kt + y - c_1)\right)^2} {t\sum_{y\in \Z_{t}} \rho_r(kt + y - c_1)^2} \cdot \frac{\sum_{y\in \Z_{t}} \rho_r(kt + y - c_1)^2}{\sum_{x\in \Z}\rho_r(x)^2 }\\
        & \ge_{} \sum_{-\frac{r\log n}{t} \leq k \leq \frac{r \log n}{t}, k\in \Z}\frac{\min_{y \in \Z_{t}}\rho_r(kt + y - c_1)^2} {\max_{y \in \Z_{t}}\rho_r(kt + y - c_1)^2} \cdot \frac{\sum_{y\in \Z_{t}} \rho_r(kt + y - c_1)^2}{\sum_{x\in \Z}\rho_r(x)^2 }\\
        & \ge_{(2)} e^{-\frac{8\pi t \log n}{r}}\sum_{-\frac{r\log n}{t} \leq k \leq \frac{r \log n}{t}, k\in \Z} \cdot \frac{\sum_{y\in \Z_{t}} \rho_r(kt + y - c_1)^2}{\sum_{x\in \Z}\rho_r(x)^2 }\\
        & \ge_{(3)} 1 - \frac{8\pi t \log n}{r} - \negl(n)
    \end{split}
\]
where $(1)$ uses $c \equiv c_1 \pmod t$, $(2)$ uses for all $-\frac{r\log n}{t} \leq k \leq \frac{r \log n}{t}$, $g(k):= \frac{\min_{y \in \Z_{t}}\rho_r(kt + y - c_1)^2}{ \max_{y \in \Z_{t}}\rho_r(kt + y - c_1)^2 } \geq e^{-\frac{8\pi t \log n}{r}} $ (to see why, observe that when $k\geq 3$, $g(k)\ge \frac{\min_{y \in \Z_{t}}\rho_r(kt +t/2-c_1)^2}{ \rho_r(kt -t/2 - c_1)^2 }\geq e^{-\frac{8\pi k t^2}{r^2}}\geq e^{-\frac{8\pi t\log n}{r}} $, same when $k\leq -3$, when $|k|\leq 2$,  $g(k)\ge e^{-\frac{32\pi t^2}{r^2}}\geq e^{-\frac{8\pi t\log n}{r}} $  ); $(3)$ is due to \Cref{lemma:Bana95}.

By our choice of parameters, $1 - \frac{8\pi t \log n}{r} - \negl(n) \ge 1 - \frac{1}{n}$, which ends the proof of \Cref{thm:find_center}.

\subsection{Constructing quantum LWE states with Gaussian amplitudes}

In this subsection, we show polynomial time quantum algorithms for solving $\QLWE$, $\LWEstate$ problems where the amplitude is Gaussian with a specific choice of phase terms. Here we will state the theorem and algorithm for $\LWEstate$, the corresponding ones for $\QLWE$ are analogous.

\begin{theorem}\label{thm:CLWE}
    Let $n, m, q, \ell$ be positive integers and $r$ be a real number. Suppose that $m = 2\ell n \cdot \omega(\log n)$, $q$ is a composite number satisfying $q = q_1q_2\cdots q_\ell$ where $q_1, q_2, \cdots, q_\ell$ are coprime, $r$ satisfies $\frac{q}{\sqrt{n}} > r > 30n\log n\cdot \max{\set{q_1, q_2, \cdots, q_\ell}}$. 
    
    There exists a quantum algorithm running in time $\poly(n, \ell, \log q)$ that, takes input $\mat{A} \gets \mathcal{U}(\Z_q^{n \times m})$, outputs a state $\rho$ such that the trace distance between $\rho$ and $\phi := \kb{\phi}{\phi}$ is negligible, where 
    \[
        \ket{\phi} = \sum_{\ary{s} \in \Z_q^n}\sum_{\ary{x} \in \Z^m}\rho_r(\ary{x})e^{-\pi i \sum_{j = 1}^{\ell} \|\ary x_j\|^2/q_j} \ket{(\mat A^T\ary{s} + \ary{x}) \bmod q}
    \]
    is a quantum LWE state with Gaussian amplitude and quadratic phase terms, with probability $1 - \negl(n)$ over the randomness of $\mat{A}$. Here we write $\ary x$ into $\ell$ blocks $\ary x = (\ary x_1, \ary x_2, \cdots, \ary x_\ell)$ with $\ary x_j\in \Z^{m/\ell}, j = 1, 2, \cdots, \ell$.
\end{theorem}

The algorithm begin by constructing the state (within trace distance $2^{-\Omega(n)}$)
\[
    \ket{\phi_0} = \sum_{\ary{s} \in \Z_q^n}\ket{\ary s}\sum_{\ary{x} \in \Z^m}\rho_r(\ary{x})e^{-\pi i \sum_{j = 1}^{\ell} \|\ary x_j\|^2/q_j} \ket{(\mat A^T\ary{s} + \ary{x}) \bmod q},
\]
by combining the complex Gaussian state preparation procedure~\cite[Lemma 2.15]{chen2024quantum} with the tail bound in~\Cref{lemma:Bana93}, then recovering $\ary s\bmod q$ via the state in the second register
\[
    \ket{\psi_{\ary s}} = \sum_{\ary{x} \in \Z^m}\rho_r(\ary{x})e^{-\pi i \sum_{j = 1}^{\ell} \|\ary x_j\|^2/q_j} \ket{(\mat A^T\ary{s} + \ary{x}) \bmod q}.
\]
and subtract it in the first register to obtain $\ket{\phi}$. At a high level, the strategy of recovering $\ary{s}\bmod q$ is to divide $\mat{A}$ into $\ell$ groups $\mat A_1, \mat A_2, \cdots, \mat A_{\ell} \in \Z_q^{n \times m/\ell}$. We use the $j^{\text{th}}$ block of $m/\ell = 2n\cdot \omega(\log n)$ samples $(\mat A_j^{T}\ary s + \ary x_j)\bmod q$ to (coherently) compute $\ary s\bmod q_j$. Then by Chinese Remainder Theorem (CRT), one can (coherently) recover $\ary s\bmod q$ via $\ary s\bmod q_j, j = 1, 2, \cdots, \ell$ and subtract it in the first register. It suffices to prove the following lemma:

\begin{lemma}\label{lem:recover_s_mod_q_j}
    For any $j = 1, 2, \cdots, \ell$, given $\mat{A}_j \leftarrow {\cal{U}}(\Z_q^{n \times m/\ell})$ and the state
    \[
        \ket{\psi_{\ary s, j}} = \sum_{\ary{x}_j \in \Z^{m/\ell}}\rho_r(\ary{x}_j)e^{-\pi i \|\ary x_j\|^2/q_j} \ket{(\mat A_j^T\ary{s} + \ary{x}_j) \bmod q},
    \]
    one can recover $\ary s\bmod q_j$ in time $\poly(n)$ with probability $1 - \negl(n)$ over the randomness of $\mat{A}_j$ and measurements.
\end{lemma}

\begin{proof}
    First observe that, when looking at the lower order bits (for modulus $q_j$) of $\ket{\psi_{\ary s, j}}$, it shares the same density matrix with the state
    \[
        \ket{\psi_{\ary s, j}'} = \sum_{\ary{x}_j \in \Z^{m/\ell}}\rho_r(\ary{x}_j)e^{-\pi i \|\ary x_j\|^2/q_j} \ket{(\mat A_j^T\ary{s} \bmod q_j + \ary{x}_j) \bmod q}.
    \]
    Therefore, when applying the center-finding technique (\Cref{thm:find_center}, which only works on the lower order bits), the measurement result of $\ket{\psi_{\ary s, j}}$ and $\ket{\psi_{\ary s, j}'}$ shares the same distribution. 

    Notice that the absolute value of each entry of $\mat A_j^T\ary{s} \bmod q_j + \ary{x}_j$ is at most $q_j/2 + rn^{1/3} \le q/2$ with overwhelming probability (by \Cref{lemma:Bana95}), when $\ary x_j$ follows the distribution of $\rho_r^2$. Thus $\ket{\psi_{\ary s, j}'}$ is $\negl(n)$-close to the state
    \[
        \ket{\psi_{\ary s, j}''} = \sum_{\ary{x}_j \in \Z^{m/\ell}}\rho_r(\ary{x}_j)e^{-\pi i \|\ary x_j\|^2/q_j} \ket{\mat A_j^T\ary{s} \bmod q_j + \ary{x}_j}.
    \]

    

    Observe that $\ket{\psi_{\ary s, j}''}$ is exactly a tensor product of complex Gaussian state $\ket{\phi_c}$ defined in \Cref{thm:find_center} for $c$ being entries of $\mat A_j^T\ary{s} \bmod q_j$. Applying the center-finding technique (\Cref{thm:find_center}) coordinate by coordinate on the state $\ket{\psi_{\ary s, j}''}$, we can obtain a measurement outcome $\tilde{\ary{y}} \in \Z_{q_j}^{m/\ell}$, which is an approximation of $\mat A_j^T\ary{s} \bmod q_j$ such that each coordinate is independently correct with probability at least $1 - \frac{1}{n}$ over the measurements (we will call it $(1-\frac{1}{n})$-approximation of $\mat A_j^T\ary{s} \bmod q_j$).

    Now let's construct a (classical) polynomial time probabilistic algorithm $\cal{B}$, which uses Gaussian elimination along with a verification process to obtain $\ary s\bmod q_j$ given $\mat A_j \leftarrow {\cal{U}}(\Z_q^{n \times m/\ell})$ and $\tilde{\ary y}$, which is supposed to be an $(1 - \frac{1}{n})$-approximation of $\mat A_j^T \ary{s} \bmod q_j$, with overwhelming probability over the randomness of $\mat A_j$ and $\tilde{y}$. Concretely, the algorithm runs
    \begin{enumerate}
        \item Divide $\mat{A}_j$ into groups of $n$ by $2n$ matrices $\mat{A}_{j, 1}, \mat{A}_{j, 2}, \cdots, \mat{A}_{j, m/(2\ell n)} \in \Z_q^{n \times 2n}$, and write $\tilde{\ary{y}}$ into blocks $\tilde{\ary{y}} = (\tilde{\ary{y}}_1, \tilde{\ary{y}}_2, \cdots, \tilde{\ary{y}}_{m/(2\ell n)})$ with each block belonging to $\Z_{q_j}^{2  n}$;
        \item For each $i = 1, 2, \cdots, m/(4\ell n)$:
        \begin{enumerate}
            \item If $\mat{A}_{j, i}$ is not full-rank (with rank $n$), go to the next iteration;
            \item Solve the linear equations $\mat{A}_{j, i}^T\ary{s} = \tilde{\ary{y}}_i \bmod q_j$ using Gaussian elimination. If the linear equations are not solvable, go to the next iteration; otherwise obtain a solution $\ary s'\in \Z_{q_j}^{n}$;
            \item Compare $\ary y' = \mat A_{j, i + m/(4\ell n)}^T \ary s' \bmod q_j$ with $\tilde{\ary y}_{i + m/(4\ell n)}$. If there are at least $0.9$ fraction of entries are matched, then output $\ary s = \ary s'$; otherwise, go to the next iteration.
        \end{enumerate}
    \end{enumerate}

    By construction, $\cal{B}$ runs in polynomial time. Now let's show its correctness.

    First, the verification process in step 2(c) ensures that the probability of outputting a wrong answer is negligible. Namely, when a wrong answer $\ary s' \neq \ary s\bmod q_j$ is obtained, each entry of $\tilde{\ary y}_{i + m/(4\ell n)}$ equals to the corresponding row of $\mat A_{j, i + m/(4\ell n)}^T \ary s \bmod q_j$ with probability at least $1 - 1/n$. Since the matrix $\mat A_j$ is chosen uniformly at random, each row of $\mat A_{j, i + m/(4\ell n)}^T \ary s \bmod q_j$ and $\ary y' = \mat A_{j, i + m/(4\ell n)}^T \ary s' \bmod q_j$ are different with probability at least $\frac{1}{2}$. Therefore, each entry of $\ary y'$ and $\tilde{\ary y}_{i + m/(4\ell n)}$ are equal with probability at most $\frac{1}{2} + \frac{1}{n} < 0.6$. By Chernoff bound, this $\ary{s'} \ne \ary{s} \bmod q_j$ only passes 2(c) with negligible probability.

    Second, the probability that the algorithm can output the correct secret $\ary s\bmod q_j$ is overwhelming. Notice that a random $n$ by $2n$ matrix is full-rank in $\Z_{q_j}$ with overwhelming probability (see \cite[Claim 3.3]{chen2024quantum}), so the failure probability in step 2(a) is negligible. Observe that none of the entries in $\tilde{\ary y}_i$ is wrong with probability at least $\left(1 - \frac{1}{n}\right)^{2n} > 0.1$, over the randomness of measurement in the center-finding process. Thus, among the $m / (4\ell n) = \omega(\log n)$ linear equations $\mat A_{j, i}^T \ary s = \tilde{\ary y_i}\bmod q_j$ for $i = 1, 2, \cdots, m/(4\ell n)$, except with negligible probability, at least one of them consists of rows without any error, which leads to the correct $\ary{s}$ in step 2(b). This correct $\ary{s}$ will pass the verification process in step 2(c) with overwhelming probability because each entry of $\ary y'$ and $\tilde{\ary y}_{i + m/(4\ell n)}$ are equal with probability at least $1 - \frac{1}{n}$ in this case.

    \Cref{lem:recover_s_mod_q_j} follows from the fact that when applying the center-finding technique, the measurement result of $\ket{\psi_{\ary s, j}}$ and $\ket{\psi_{\ary s, j}''}$ are indistinguishable. 
\end{proof}






\subsection{Oblivious LWE sampling}

While constructing a quantum LWE state with a specific choice of phase terms may not seem interesting at first glance, the problem is actually closely related to a problem called oblivious LWE sampling for a vast range of LWE parametrizations.

Roughly speaking, in oblivious LWE sampling, the sampler is asked to sample from the LWE distribution $(\mat{A}, \mat{A}^T \ary{s}+\ary{e})$, while any extractor, observing the sampler (but not disturbing the sampler's state), cannot extract the secret $\ary{s}$. 

Here we abuse the notation $\LWE_{n, m, q, \alpha}$ to represent the distribution of valid LWE samples $\mat{A}, \ary{b} = \mat{A}^T\ary{s}+\ary{e}\bmod q$, where $\mat{A}\la U(\Z_q^{n\times m})$, $\ary{s}\la U(\Z_q^n)$, and $\ary{e}\la D_{\Z, \alpha q}^m$. 
Later we will use two variants of $\LWE_{n, m, q, \alpha}$. 
Let $\LWE_{n, m, q, \leq \alpha}$ denote the case where the error is sampled from $D_{\Z, \beta q}$ for some $\beta\leq \alpha$. 
Let $\LWE_{n, m, q, \alpha}(D_{\Z^n, s})$ denote  the distribution that produces LWE samples where the secret is sampled from $D_{\Z^n, s}$ for some $s>0$.

\begin{definition}[Witness-Oblivious Quantum Samplers, \cite{debris2024quantum}]
    Let $n, m, q, \alpha$ be LWE parameters following the same definition as in~\Cref{def:LWE}. A quantum polynomial time algorithm $\cal{S}$ is called a \emph{witness-oblivious quantum sampler} for $\LWE_{n, m, q, \alpha}$ if it has the following properties:
    \begin{enumerate}
        \item Given as input the security parameter $1^\lambda$, a matrix $\mat{A} \in \Z_q^{n \times m}$, and polynomial ancillas initialized to $\ket{0}$, $\cal{S}$ outputs a pair $(\mat{A}, \ary{b})$ where $\ary{b} \in \Z_q^m$ such that for a uniformly distributed $\mat{A}$, the distribution of ${\cal{S}}(1^{\lambda}, \mat{A}, \ket{0}^{\poly(\lambda)})$ is statistically indistinguishable with the distribution $\LWE_{n, m, q, \alpha}$;
        \item For any valid quantum polynomial-time extractor $\cal{E}$ for $\cal{S}$, we have that
        \iffull
        \[\Pr\left[\ary{s} = \ary{s'} \text{ and } \ary{e} = \ary{e'}: \substack{\mat{A} \leftarrow \Z_q^{n \times m}\\ \left(\left(\mat{A}, \ary{b} = \mat{A}^T\ary{s} + \ary{e}\right), \left(\ary{s'}, \ary{e'}\right)\right) \leftarrow \langle {\cal{S}}, {\cal{E}}\rangle\left(\tau_{\cal{S}} = \left(1^{\lambda}, \mat{A}, \ket{0}^{\poly(\lambda)}\right), \tau_{\cal{E}} = \ket{0}^{\poly(\lambda)}\right)}\right] \leq \negl(\lambda)\]
        \else
        \[\Pr\left[\ary{s} = \ary{s'} ,\ary{e} = \ary{e'}: \substack{\mat{A} \leftarrow \Z_q^{n \times m}\\ \left(\left(\mat{A}, \ary{b} = \mat{A}^T\ary{s} + \ary{e}\right), \left(\ary{s'}, \ary{e'}\right)\right) \leftarrow \langle {\cal{S}}, {\cal{E}}\rangle\left( \left(1^{\lambda}, \mat{A}, \ket{0}^{\poly(\lambda)}\right), \ket{0}^{\poly(\lambda)}\right)}\right] \leq \negl(\lambda)\]
        \fi
        where $\langle {\cal{S}}, {\cal{E}}\rangle(\tau_{\cal{S}}, \tau_{\cal{E}})$ denote the joint output of the sampler $\cal{S}$ and the extractor $\cal{E}$ on the input $\tau_{\cal{S}}$ and $\tau_{\cal{E}}$, and an extractor $\cal{E}$ is called \emph{valid} for the sampler $\cal{S}$ if and only if $\cal{E}$ does not change the state of the sampler $\cal{S}$ up to negligible trace distance when $\cal{S}$ and $\cal{E}$ are running together.
    \end{enumerate}
\end{definition}

\begin{theorem}[Theorem 2, \cite{debris2024quantum}]\label{thm:oblivious_to_CLWE}
    Let $n, m, q, \alpha$ be LWE parameters. Assume the quantum hardness of $\LWE_{n, m, q, \alpha}$. If a quantum polynomial algorithm $\cal{S}$ solves $\LWEstate_{n, m, q, f}$, where the amplitude $f$ satisfies $|f|^2\propto D_{\Z, \alpha q}$, then $\cal{S}$ followed by computational measurements is a witness-oblivious quantum sampler for $\LWE_{n, m, q, \alpha}$.
\end{theorem}

\begin{corollary}\label{coro:OSAMP}
    Let $n, m, q, \ell$ be positive integers and $r$ be a real number. Suppose that $m \geq 2\ell n\cdot \omega(\log n)$, $q$ is a composite number satisfying $q = q_1q_2\cdots q_\ell$ where $q_1, q_2, \cdots, q_\ell$ are coprime, $r$ satisfies $\frac{q}{\sqrt{n}} > r > 30n\log n\cdot \max{\set{q_1, q_2, \cdots, q_\ell}}$. Assume the quantum hardness of $\LWE_{n, m, q, r/(\sqrt{2}q)}$, there exists a witness-oblivious quantum sampler for $\LWE_{n, m, q, r/(\sqrt{2}q)}$.
\end{corollary}

\begin{proof}
This is a direct corollary of \Cref{thm:CLWE} and \Cref{thm:oblivious_to_CLWE}.
\end{proof}



As noted in \cite{debris2024quantum}, once we get an oblivious $\LWE_{n,m, q, \alpha}$ sampler for certain composite modulus $q$, we can throw away additional samples to get an oblivious $\LWE_{n,m', q, \alpha}$ sampler for $m'\leq m$. We can also use the modulus switching technique in \cite{DBLP:conf/stoc/BrakerskiLPRS13} to get an oblivious LWE sampler for $\LWE_{n,m, q', \alpha'}$ where $q'< q$ is not necessarily composite. 

Here we spell out the details of the modulus switching part. To use modulus switching, we need to start from LWE samples where the length of the secret is shorter than $q$ (it is easy to modify the quantum algorithm for $\LWEstate$ so that the secret is a superposition of a general distribution). For consistency let us fix the secret to be sampled from $D_{\Z^n, s}$ and adapt from \cite[Corollary~3.2]{DBLP:conf/stoc/BrakerskiLPRS13}:
\begin{lemma}\label{lemma:modulusswitching}
For any positive integers $n, m$, $q\geq q'\geq 1$, $\alpha, \alpha'\in(0, 1)$ such that $m, \log q\in \poly(n)$, 
\[ \left( \alpha' \right)^2 \geq  \alpha^2 + \left( \frac{s\sqrt{n}}{q'} \right)^2 \cdot \omega(\log n),   \]
there is a $\poly(n)$ time classical algorithm that takes as input a sample from $\LWE_{n, m, q, \leq \alpha}(D_{\Z^n, s})$, outputs a sample within negligible statistical distance from $\LWE_{n, m, q', \leq \alpha'}(D_{\Z^n, s})$. 
\end{lemma}

As a corollary of \Cref{coro:OSAMP} and \Cref{lemma:modulusswitching}:
\begin{corollary}\label{coro:modulusswitch}
    Let $n, m, q, \ell$ be positive integers and $r$ be a real number. 
    Suppose that $m \geq 2\ell n\cdot \omega(\log n)$, $q$ is a composite number satisfying $q = q_1q_2\cdots q_\ell$ where $q_1, q_2, \cdots, q_\ell$ are coprime, $r$ satisfies $\frac{q}{\sqrt{n}} > r > 30n\log n\cdot \max{\set{q_1, q_2, \cdots, q_\ell}}$. 
    Suppose positive integers $m'\leq m, q'\leq q$ and real numbers $s$, $r'>0$ satisfies 
    \[ \left( \frac{r'}{\sqrt 2 q'} \right)^2 \geq  \left( \frac{r}{\sqrt 2 q} \right)^2 + \left( \frac{s\sqrt{n}}{q'} \right)^2 \cdot \omega(\log n),   \]
    and assume the quantum hardness of $\LWE_{n, m, q, r/(\sqrt{2}q)}(D_{\Z^n, s})$, there exists a witness-oblivious quantum sampler for $\LWE_{n, m', q', r'/(\sqrt{2}q')}(D_{\Z^n,s})$.
\end{corollary}

For example, to produce an oblivious LWE sample from $\LWE_{n, m', q', r'/(\sqrt{2}q')}(D_{\Z^n, s})$ where $q'\in \tilde{O}(n^2)$ is prime, $r'\in \tilde{O}(n^{1.5})$, $s\in O(\sqrt{n})$, we can start from an oblivious LWE sample from $\LWE_{n, m, q, r/(\sqrt{2}q)}(D_{\Z^n, s})$ where  $r\in \tilde{O}(n^{1.5})$, $q\in \tilde{O}(n^2)$ such that $q'<q<2q'$,  $q = q_1q_2q_3q_4$ where $q_1, q_2, q_3, q_4\in \tilde{O}(\sqrt{n})$ are coprime.
The sample complexity required for producing $\LWE_{n, m, q, r/(\sqrt{2}q)}(D_{\Z^n, s})$ from our quantum algorithm in \Cref{coro:OSAMP} is in $\tilde{O}(n)$. 
If we use the quantum oblivious LWE sampler in \cite{debris2024quantum}, the sample complexity required for producing $\LWE_{n, m', q', r'/(\sqrt{2}q')}(D_{\Z^n, s})$ is $\tilde{O}(n^{2.5})$ .

\section{Hardness of $\QLWE$ with Unknown Phase via Extrapolated DCP}\label{sec:EDCP}

In this section, we show how to obtain a quantum reduction from classical LWE to $\QLWE^{\mathsf{phase}}$ with Gaussian amplitude. Our reduction goes through the Extrapolated Dihedral Coset problem (EDCP), derived from a modification of the reduction by Regev \cite{DBLP:conf/focs/Regev02} and Brakerski et al. \cite{DBLP:conf/pkc/BrakerskiKSW18}. Our reduction consists of two steps:
\begin{enumerate}
    \item[Step 1] Given a classical LWE instance, our quantum reduction first generates an Extrapolated DCP state with amplitudes following a Gaussian distribution centered at an unknown value. More precisely, we establish the following theorem:
    \begin{theorem}\label{thm:EDCP_Step1}
        Let $n, m, q\in \N^+$, $\alpha = \Omega(\sqrt{n}), \beta, \gamma\in (0, 1)$ satisfy $m\geq n\log q$, $\alpha\gamma\sqrt{m} < \beta < \frac{1}{16\sqrt{m\log (\beta q)}}$. There exists a $\poly(n)$ time quantum algorithm that, given a classical LWE instance $\LWE_{n, m, q, \gamma} = (\mat{A}, \mat{A}^T\ary{s} + \ary{e})$, generates, with probability $1 - 2^{-\Omega(n)}$,  a vector $\ary y\in \Z_q^m\cap B_{\Lattice_q(\mat A)}$ and an Extrapolated DCP state of form 
        \begin{equation}\label{eqn:EDCP_form}
            \ket{\EDCP} = \sum_{j\in \Z_q}\rho_{\sigma}(j - c)\ket{j}\ket{(\ary v + j \cdot \ary s)\bmod q},
        \end{equation}
        where 
        \begin{enumerate}
            \item[1.] the vector $\ary v$ is chosen uniformly at random from $\Z_q^n$ and is unknown,
            \item[2.] the Gaussian width $\sigma = \frac{\alpha \beta q}{ \sqrt{ \alpha^2\|\ary{e}\|^2 +\beta^2 q^2 } }$,
            \item[3.] the vector $\ary y$ is sampled by first sampling $\ary x \in \Z^m \cap B^m(\lambda_1(\Lattice_q(\mat A))/2)$ with probability proportional to $\Pr(\ary x) \propto \rho_{\beta q\sqrt{\mat \Sigma / 2}}(\ary x)$ where $\mat \Sigma = \mat I_m + \frac{\alpha^2}{\beta^2 q^2}\ary e\ary e^T$, then outputting $\ary y = (\mat A^T \ary v + \ary x) \bmod q$,
            \item[4.] the center $c = - \frac{ \alpha^2\ipd{\ary x}{\ary{e}} }{ \alpha^2\|\ary{e}\|^2 + \beta^2 q^2 }$ (we don't know how to efficiently compute $c$ with success probability $1 - 2^{-\Omega(n)}$ since we don't know $\ary x$ and $\ary{e}$; we can guess $c$ correctly with non-negligible probability, but the event of guessing correctly is not efficiently checkable).
        \end{enumerate}
    \end{theorem}

    \item[Step 2] Given an Extrapolated DCP state with amplitudes following a Gaussian distribution centered at an unknown value, we adapt the quantum reduction proposed by Brakerski et al. \cite{DBLP:conf/pkc/BrakerskiKSW18} to transform it to an $\QLWE^{\sf phase}$ instance. The resulting $\QLWE^{\sf phase}$ instance will have amplitudes represented as a Gaussian distribution multiplied by an unknown phase term. More precisely, we establish the following theorem:
    \begin{theorem}\label{thm:EDCP_Step2}
        There exists an efficient quantum algorithm that, given an Extrapolated DCP state of form as \Cref{eqn:EDCP_form}, generates an $\QLWE^{\sf phase}$ state of form 
        \begin{equation}\label{eqn:SLWE_form_EDCP}
            \sum_{e \in \Z} \rho_{q/\sigma}(e) \exp(2\pi {\rm i}\cdot ce /q) \ket{(\ipd{\ary a}{\ary s} + e) \bmod q}
        \end{equation}
        along with a known vector $\ary a\gets \mathcal{U}(\Z_q^n)$. Here the parameters $\sigma, c$ correspond to those mentioned in \Cref{thm:EDCP_Step1}.
    \end{theorem}
\end{enumerate}

\iffull
In the remaining part of this section, we will provide the detailed proofs for \Cref{thm:EDCP_Step1} and \Cref{thm:EDCP_Step2}.
\fi
By combining these two theorems, we achieve a quantum reduction from classical LWE to $\QLWE^{\sf phase}$ with Gaussian amplitude. This establishes that solving the classical LWE problem is as hard as solving the $\QLWE^{\sf phase}$ with Gaussian amplitude for the phase function defined below. Formally, we define our special parameters and functions, 
and propose the main theorem for this section here:
\begin{definition}\label{def:paramsEDCP}
    Let $n, m, q\in \N^+$, $\alpha = \Omega(\sqrt{n}), \beta, \gamma\in (0, 1)$ satisfy $m\geq n\log q$, $\alpha\gamma\sqrt{m} < \beta < \frac{1}{16\sqrt{m\log (\beta q)}}$. 
    Given a matrix $\mat{A} \in \Z_q^{n \times m}$ and a vector $\ary{e}$ such that $\|\ary{e}\| \leq \sqrt{m}\gamma q$, we define:
    \begin{enumerate}
        \item A family of functions $\{f_E: \Z\to \R\}_{E \in [m\gamma^2q^2]}$ with $f_E(e) = \rho_{q/\sigma(E)}(e)$ where $\sigma(E) = \frac{\alpha \beta q}{ \sqrt{ \alpha^2 E +\beta^2q^2 } }$.
        \item A distribution $D_{\theta, \ary e, \mat A}(\ary y)$ over $\Z_q^m\cap B_{\Lattice_q(\mat A)}$ given by $\Pr(\ary y) \propto \rho_{\beta q\sqrt{\mat \Sigma / 2}}(\ary y')$ where $\mat \Sigma = \mat I_m + \frac{\alpha^2}{\beta^2 q^2}\ary e\ary e^T$ and $\ary y' = \ary y - \kappa_{\Lattice_q(\mat A)}(\ary y)$.
        \item A phase function $\theta_{\ary e, \mat A}: \Z_q^m\cap B_{\Lattice_q(\mat A)}\to \R$ with $\theta_{\ary e}(\ary y) = - \frac{ \alpha^2\ipd{\ary y'}{\ary{e}} }{ q(\alpha^2\|\ary{e}\|^2 + \beta^2 q^2) }$ (It is not efficiently computable when assuming classical LWE is hard).
    \end{enumerate}
\end{definition}

\begin{theorem}[Main theorem, from $\LWE$ to $\QLWE^{\mathsf{phase}}$]\label{thm:EDCP_MainThm}
    Following the parameters defined in \Cref{def:paramsEDCP}. Assume for any $E \in [m\gamma^2q^2]$, there exists a quantum algorithm that takes $\mat{A} \leftarrow \mathcal{U}(\Z_q^{n \times m})$ as input, with $1 - 2^{-\Omega(n)}$ probability over $\mat{A}$, solves (see \Cref{def:SLWEphase}) $\QLWE_{n, \ell, q, f_E, \theta_{\ary e, \mat{A}}, D_{\theta, \ary e, \mat{A}}}^{\mathsf{phase}}$ for every $\ary{e} \in \Z^m$ such that $\|\ary{e}\|_2^2 = E$, with $\ell = 2^{o(n)}$ and time complexity $T = 2^{o(n)}$, then there exists a quantum algorithm that takes a classical LWE sample $(\mat A, \mat A^T \ary s + \ary e)$ as input, and outputs the secret vector $\ary s\in \Z^n_q$ with success probability $1 - 2^{-\Omega(n)}$ and time complexity $O((T + \ell \cdot \poly(n, q)) \cdot m\gamma^2 q^2)$.
\end{theorem}

\begin{remark}
    For an example of parameters, 
    let $q\in \tilde{O}(n^2)$, $m\in \Omega(n\log q)$, $\|\ary{e}\|_2^2\in \tilde{O}(n)$, $\alpha\in \tilde{O}(n^{0.5})$, $\beta\in \tilde{O}(n^{-0.5})$, $\gamma\in \tilde{O}(n^{-1.5})$. 
\end{remark}

\begin{proof}
    We proceed to address the classical LWE instance $(\mat A, \mat A^T \ary s + \ary e)$ as follows:
    \begin{enumerate}
        \item Enumerate $E\in \{1, 2, \cdots,  m \gamma^2 q^2\}$ to make a guess for $\|\ary e\|^2$.
        \item\label{main_step_generate_SLWE} Apply \Cref{thm:EDCP_Step1} and \Cref{thm:EDCP_Step2} $\ell$ times to generate $\ell$ instances of $\QLWE^{\sf phase}$ in the form of \Cref{eqn:SLWE_form_EDCP}.
        \item Utilize the quantum algorithm in the assumption for $\QLWE_{n, \ell, q, f_E, \theta_{\ary e, \mat A}, D_{\theta, \ary e, \mat A}}^{\mathsf{phase}}$, with those $\ell$ $\QLWE^{\sf phase}$ instances as input, to derive a solution $\ary s'$.
        \item Employ any verification algorithm (e.g., as proposed by Regev \cite[Lemma 3.6]{DBLP:journals/jacm/Regev09}) to ascertain whether $\ary s' = \ary s$. If this condition holds, output $\ary s'$ and conclude the process.
    \end{enumerate}
    
    It can be easily verified that this algorithm operates with a runtime of $O((T + \ell \cdot \poly(q, n)) \cdot m\gamma^2 q^2 )$. Furthermore, as indicated in \Cref{thm:EDCP_Step1}, the probability of successfully generating $\ell$ $\QLWE^{\sf phase}$ states in step \ref{main_step_generate_SLWE} is exponentially close to 1. Thus, when $E = \|\ary e\|^2$ (i.e., when we guess $E$ correctly), the probability that the solution $\ary s' = \ary s$ is exponentially close to 1. Consequently, the aforementioned algorithm achieves success probability exponentially close to 1.
\end{proof}

\subsection{Reduce classical LWE to Extrapolated DCP}

Our quantum reduction from classical LWE to Extrapolated DCP follows the general design of Regev's reduction \cite{DBLP:conf/focs/Regev02} and the reduction proposed by Brakerski et al. \cite{DBLP:conf/pkc/BrakerskiKSW18}. In these reductions, the Euclidean space $\mathbb{R}^n$ is divided into grids, with each grid cell having a width that lies between the length of the error vector $\|\ary{e}\|$ and the length of the shortest vector in the lattice $\lambda_1(\Lattice_q(\mat{A}))$. The key observation is that when randomly selecting a vector $\ary{x} \in \mathbb{R}^n$, the vectors $\ary{x}, \ary{x} + \ary{e}, \cdots, \ary{x} + k \cdot \ary{e}$ will be in the same grid cell with high probability, creating a superposition in the quantum world. 

We modify the reductions in \cite{DBLP:conf/focs/Regev02,DBLP:conf/pkc/BrakerskiKSW18} by introducing Gaussian balls around all lattice points in $\Lattice_q(\mat A)$, where the radius of each ball is a quantity smaller then the length of shortest vector in the lattice $\lambda_1(\Lattice_q(\mat A))$.
Note that the reductions in \cite{DBLP:conf/focs/Regev02,DBLP:conf/pkc/BrakerskiKSW18} use Euclidean balls or cubes.

Here we give the detailed proof of \Cref{thm:EDCP_Step1}.
\begin{proof}[Proof of \Cref{thm:EDCP_Step1}]
    Following the parameters defined in \Cref{thm:EDCP_Step1}. Relying on both \Cref{lemma:linfty} and Banaszczyk's tail bound from \Cref{lemma:Bana93}, we make the assumptions that $\lambda_1(\Lattice_q(\mat A)) \ge \frac{q}{4}$ and $\|\ary e\| < \sqrt{m}\gamma q$ for the remainder of this proof; these assumptions hold with probability $1 - 2^{-\Omega(n)}$. Our quantum reduction from classical LWE to Extrapolated DCP works as follows (for simplicity, we omit the normalization factors):
    \begin{enumerate}
        \item We start by preparing the superposition using the Gaussian state sampler (see \Cref{lemma:Gaussianstateprep})
        \[
            \sum_{j\in\Z_q} \rho_\alpha(j)\ket{j} \sum_{\ary{v}\in\Z_q^n} \ket{ \ary{v} } \sum_{\ary{x}\in\Z_q^m} \rho_{\beta q}(\ary{x})\ket{\ary{x}}.
        \]
        \item We apply a unitary to compute $(j, \ary{v}, \ary{x})\to (\mat{A}^T\ary{v} - j\cdot (\mat A^T\ary s + \ary e) + \ary{x})\bmod q$ on the third register, obtaining the state
        \begin{equation}\label{eqn:EDCP_state_temp}
            \sum_{j\in\Z_q} \rho_\alpha(j)\ket{j} \sum_{\ary{v}\in\Z_q^n} \ket{ \ary{v} } \sum_{\ary{x}\in\Z_q^m} \rho_{\beta q}(\ary{x})\ket{ (\mat{A}^T \ary{v} - j\cdot (\mat A^T\ary s + \ary e) + \ary{x}) \bmod q}.
        \end{equation}
        This state approximates, with an error of $1 - 2^{-\Omega(n)}$, the same state structure with the only difference being the range of $\ary x$ in the summation, which is now $\Z^m$ rather than $\Z_q^m$. By expressing $\mat A^T\ary v - j\cdot (\mat A^T\ary s + \ary e) + \ary x$ as $\mat A^T(\ary v - j\cdot \ary s) + (\ary x - j\cdot \ary e)$, we perform a change of variables $\ary v \gets (\ary v + j\cdot \ary s) \bmod q$ and $\ary x\gets \ary x + j\cdot \ary e$, yields that the state of form \Cref{eqn:EDCP_state_temp} is $2^{-\Omega(n)}$-close to the state
        \begin{equation}\label{eqn:EDCP_state_temp0}
            \sum_{j\in\Z_q} \rho_\alpha(j)\ket{j} \sum_{\ary{v}\in\Z_q^n} \ket{ (\ary{v} + j\cdot \ary s) \bmod q } \sum_{\ary{x}\in\Z^m} \rho_{\beta q}(\ary{x} + j\cdot \ary e)\ket{ (\mat{A}^T\ary{v} + \ary{x}) \bmod q}.
        \end{equation}
        To proceed, we need the following lemma to guarantee that each vector in the support of the third register $\ary{y}:=\ary{A}^T\ary{v}+\ary{x} \bmod q$ corresponds to a unique $\ary{x}$, in order to match with the target of \Cref{thm:EDCP_Step1} and to simplify later analyses. The proof of this lemma is deferred to \Cref{sec:EDCP_state_close}.
        \begin{lemma}\label{lem:EDCP_state_close}
            Assume that $\left(\beta q\sqrt{m} + \alpha\gamma q m\right)\sqrt{\log (\beta q)} < \lambda_1(\Lattice_q(\mat A)) / 2$ and $\beta q > \sqrt{m}$, then the state in \Cref{eqn:EDCP_state_temp0} is $2^{-\Omega(n)}$-close to the state
            \begin{equation}\label{eqn:EDCP_state_temp1}
                \sum_{j\in\Z_q} \rho_\alpha(j)\ket{j} \sum_{\ary{v}\in\Z_q^n} \ket{ (\ary{v} + j\cdot \ary s) \bmod q } \sum_{\substack{\ary{x}\in\Z^m, \\ \|\ary{x}\| < \lambda_1(\Lattice_q(\mat A))/2}} \rho_{\beta q}(\ary{x} + j\cdot \ary e)\ket{ (\mat{A}^T\ary{v} + \ary{x}) \bmod q}.
            \end{equation}
        \end{lemma}

        \begin{remark}
            The condition of this lemma can be relaxed if we instead use the proof technique from Claim A.5 in~\cite{regev2023efficient}. To be more specific, the above lemma holds as long as $\beta q \sqrt{m} + \alpha\gamma q m < \lambda_1(\Lattice_q(\mat A)) / 2$ and $q/2 > \alpha\sqrt{m}$. However, the relaxed condition provides similar parameters in the overall reduction, and therefore, we do not emphasize it here. 
        \end{remark}
        
        Under the condition of \Cref{thm:EDCP_Step1}, $\left(\beta q \sqrt{m} + \alpha\gamma q m\right)\sqrt{\log (\beta q)} < 2\beta q\sqrt{m\log (\beta q)} < q/8 < \lambda_1(\Lattice_q(\mat A)) / 2$ and $\beta q > \alpha\gamma q\sqrt{m} > \sqrt{m}$, so the state in \Cref{eqn:EDCP_state_temp0} is $2^{-\Omega(n)}$-close to the state in \Cref{eqn:EDCP_state_temp1}, which can be rewritten as follows
        \begin{equation}\label{eqn:EDCP_state_temp2}
            \sum_{\substack{\ary{v}\in\Z_q^n, \ary{x}\in\Z^m, \\ \|\ary{x}\| < \lambda_1(\Lattice_q(\mat A))/2}} \left(\sum_{j\in \Z_q} \rho_\alpha(j) \rho_{\beta q}(\ary{x} + j\cdot \ary e) \ket{j}  \ket{ (\ary{v} + j\cdot \ary s) \bmod q }\right) \ket{ (\mat{A}^T\ary{v} + \ary{x}) \bmod q}
        \end{equation}

        \item\label{LWE_EDCP_step_measure} We measure the state in \Cref{eqn:EDCP_state_temp2} on the third register and denote the result as $\ary y = (\mat A^T \ary v + \ary x) \bmod q$ (note that we have $\|\ary x\| < \lambda_1(\Lattice_q(\mat A)) / 2$, so the vectors $\ary v$ and $\ary x$ are both unique), then the remaining state on the first two registers is
        \[
            \sum_{j\in\Z_q} \rho_\alpha(j) \rho_{\beta q}( \ary x + j\cdot \ary e ) \ket{j} \ket{ (\ary v + j\cdot \ary{s}) \bmod q}.
        \]
        The amplitude of this state is computed as follows
        \[
            \begin{split}
                \rho_\alpha(j) \rho_{\beta q}( \ary x + j\cdot \ary e ) & = \exp\left[-\pi \left( \frac{j^2}{\alpha^2} + \frac{j^2 \|\ary{e}\|^2 + 2j\ipd{\ary x}{\ary{e}} + \|\ary x\|^2 }{\beta^2 q^2} \right)  \right] \\
                & \propto \exp\left[-\pi \left( \frac{ ( \alpha^2\|\ary{e}\|^2 +\beta^2 q^2) j^2 + 2j\alpha^2\ipd{\ary x}{\ary{e}}  }{\alpha^2 \beta^2 q^2} \right) \right] \\
                & \propto \rho_\sigma(j - c),
            \end{split}
        \]
        where the Gaussian width $\sigma = \frac{\alpha \beta q}{ \sqrt{ \alpha^2\|\ary{e}\|^2 +\beta^2 q^2 } } \in (\alpha / \sqrt{2}, \alpha)$ and the center $c = - \frac{ \alpha^2\ipd{\ary x}{\ary{e}} }{ \alpha^2\|\ary{e}\|^2 + \beta^2 q^2 }$. Unfortunately, the center $c$ remains unknown because we have no knowledge of $\ary x$ other than that it is the error term in the LWE sample $\ary y = (\mat{A}^T \ary{v} + \ary x) \bmod q$. The analysis of the distribution of $\ary y$ and $c$ is deferred to \Cref{sec:dist_unknown_center}.
    \end{enumerate}

    It's evident that in the state of \Cref{eqn:EDCP_state_temp1}, the amplitude for every $\ary v\in \Z_q^n$ is the same, which implies that the distribution of the vector $\ary v$ is uniformly random. This completes the proof.
\end{proof}

\begin{remark}
    It seems that our reduction bears similarities to the reduction from classical LWE to G-EDCP (Extrapolated DCP with amplitudes following a Gaussian distribution) proposed by Brakerski et al. \cite{DBLP:conf/pkc/BrakerskiKSW18}. However, our reduction, compared to both Regev's reduction and Brakerski et al.'s reduction, exhibits superior success probability. In the previous reductions, the failure probability is inverse-polynomial, leading to the reduction of classical LWE to only polynomially many (Extrapolated) DCP states. In contrast, our reduction achieves $1 - 2^{-\Omega(n)}$ success probability, allowing for the construction of sub-exponentially many Extrapolated DCP states without failure, the cost is introducing an unknown center in the Gaussian distribution of the amplitude. This novel approach offers the potential for creating sub-exponential time quantum algorithms for the standard LWE problem.
\end{remark}

\iffull
\subsection{Reduce Extrapolated DCP to $\QLWE^{\sf phase}$}\label{sec:DCP_SLWE}

The second step of our quantum reduction from classical LWE to $\QLWE^{\sf phase}$ involves reducing the obtained Extrapolated DCP states to $\QLWE^{\sf phase}$. This step is an adaptation of the reduction from G-EDCP to LWE proposed by Brakerski et al. \cite{DBLP:conf/pkc/BrakerskiKSW18}. We give the detailed proof of \Cref{thm:EDCP_Step2} here.

\begin{proof}[Proof of \Cref{thm:EDCP_Step2}]
    Suppose we are given an Extrapolated DCP state with the form as \Cref{eqn:EDCP_form}. Our quantum reduction works as follows (for simplicity, we omit the normalization factors):
    \begin{enumerate}
        \item Apply Quantum Fourier Transformation on $\Z_q^n$ for the second register, obtaining the state
        \[
            \sum_{\ary{a}\in\Z_q^n} \sum_{j\in\Z_q} \omega_q^{\ipd{\ary{a}}{\ary{v} + j\cdot \ary{s}}} \rho_{\sigma}(j - c) \ket{j} \ket{ \ary{a} }.
        \]
        \item Measure the second register to get a particular measurement result $\hat {\ary a}$, which is randomly chosen from $\Z_q^n$ with a uniform distribution. By omitting the global phase term $\omega_q^{\ipd{\hat{\ary{a}}}{\ary{v}}}$, the remaining state is
        \[
            \sum_{j\in\Z_q} \omega_q^{\ipd{\hat{\ary{a}}}{j\cdot \ary{s}}} \rho_{\sigma}(j - c) \ket{j}.
        \]
        \item Apply another Quantum Fourier Transformation on $\Z_q$ and incorporate Gaussian tails of $j$ again, obtaining a state $2^{-\Omega(n)}$-close to the state
        \[
            \sum_{b\in \Z_q} \sum_{j\in\Z} \omega_q^{j(\ipd{\hat{\ary{a}}}{\ary{s}} + b)} \rho_{\sigma}(j - c) \ket{b}.
        \]
        \item Use the Poisson summation formula on the amplitude and change the summation variable to $e \gets \ipd{\hat{\ary{a}}}{\ary{s}} + b - q\cdot j$, this state can be rewritten as
        \[
            \begin{split}
                &~ \sum_{b\in \Z_q} \sum_{j\in\Z} \omega_q^{j(\ipd{\hat{\ary{a}}}{\ary{s}} + b)} \rho_{\sigma}(j - c) \ket{b} \\
                = &~ \sum_{b\in \Z_q}\sum_{j\in \Z} \sigma\rho_{1/\sigma}\left(j - \frac{\ipd{\hat{\ary{a}}}{\ary{s}} + b}{q}\right)\cdot \exp\left(-2\pi {\rm i} \cdot c\left(j - \frac{\ipd{\hat{\ary{a}}}{\ary{s}} + b}{q}\right)\right)\ket{b} \\
                \propto &~ \sum_{e\in \Z}\rho_{q/\sigma}(e)\cdot \exp(2\pi {\rm i}\cdot ce/q)\ket{(\ipd{-\hat{\ary{a}}}{\ary{s}} + e) \bmod q}
            \end{split}
        \]
    \end{enumerate}
    Finally, this state along with the classical vector $-\hat{\ary{a}}$ will be the output $\QLWE^{\sf phase}$ instance of our quantum reduction.
\end{proof}
\fi

\subsection{The distribution of unknown center}\label{sec:dist_unknown_center}

For additional technical insights, we present a more detailed analysis of the distribution of center $c$ in the Extrapolated DCP states (see \Cref{eqn:EDCP_form}) we get. To achieve this, we begin by examining the distribution of $\ary x$ after measurement on the third register of the state given in \Cref{eqn:EDCP_state_temp2}. It is evident that the probability of obtaining a specific vector $\ary x$ is proportional to
\[
    \begin{split}
        \Pr(\ary x) & \propto \sum_{j\in \Z_q}\rho_\alpha(j)^2\rho_{\beta q}(\ary x + j \cdot \ary e)^2 \\
        & \approx \sum_{j\in \Z}\rho_\alpha(j)^2\rho_{\beta q}(\ary x + j \cdot \ary e)^2 \\
        & = \sum_{j\in \Z}\exp\left[-2\pi \left( \frac{j^2}{\alpha^2} + \frac{j^2 \|\ary{e}\|^2 + 2j\ipd{\ary x}{\ary{e}} + \|\ary x\|^2 }{\beta^2 q^2} \right) \right] \\
        & = \sum_{j\in \Z}\rho_{\sigma / \sqrt{2}}(j - c)\cdot \exp\left[-2\pi \left(\frac{\|\ary x\|^2}{\beta^2 q^2} - \frac{\alpha^2\ipd{\ary x}{\ary e}^2}{\beta^2 q^2 (\alpha^2\|\ary e\|^2 + \beta^2 q^2)}\right)\right]. \\
    \end{split}
\]
We observe that $\sigma = \frac{\alpha \beta q}{ \sqrt{ \alpha^2\|\ary{e}\|^2 +\beta^2 q^2 } } > \frac{\alpha}{\sqrt{2}} = \Omega(\sqrt{n})$, which implies that almost all of the weight of $\rho_{\sqrt{2}/\sigma}$ is concentrated at $\rho_{\sqrt{2}/\sigma}(0)$ with exponentially small weight elsewhere. Using the Poisson summation formula, we get that
\[
    \sum_{j\in \Z}\rho_{\sigma / \sqrt{2}}(j - c) = \sum_{j\in \Z} \frac{\sigma}{\sqrt{2}} \rho_{\sqrt{2}/\sigma}(j) \cdot \exp(-2\pi {\rm i}\cdot cj) \in \frac{\sigma}{\sqrt{2}}(1 \pm 2^{-\Omega(n)}).
\]
Let us denote $\mat \Sigma = \left(\mat I_m - \frac{\alpha^2}{\alpha^2\|\ary e\|^2 + \beta^2 q^2}\ary e\ary e^T\right)^{-1} = \mat I_m + \frac{\alpha^2}{\beta^2 q^2}\ary e\ary e^T$, the remaining term can be written as

\iffull
\[
    \begin{split}
        \exp\left[-2\pi \left(\frac{\|\ary x\|^2}{\beta^2 q^2} - \frac{\alpha^2\ipd{\ary x}{\ary e}^2}{\beta^2 q^2(\alpha^2\|\ary e\|^2 + \beta^2 q^2)}\right)\right] & = \exp\left[-2\pi \cdot\frac{1}{\beta^2 q^2}(\ary x)^T\left(\mat I_m - \frac{\alpha^2}{\alpha^2\|\ary e\|^2 + \beta^2 q^2}\ary e\ary e^T\right)\ary x\right] \\
        & = \exp\left[-2\pi \cdot\frac{1}{\beta^2 q^2}(\ary x)^T\mat \Sigma^{-1}\ary x\right] \\
        & = \rho_{\beta q\sqrt{\mat \Sigma / 2}}(\ary x).
    \end{split}
\]
\fi

This means that the distribution of $\ary x$ follows the discrete Gaussian distribution with center $\ary 0$ and covariance matrix $\beta^2 q^2\mat \Sigma/2$. Correspondingly, the distribution of $\ary y\in \Z_q^m\cap B_{\Lattice_q(\mat A)}$ is given by $\Pr(\ary y) = \rho_{\beta q\sqrt{\mat \Sigma / 2}}(\ary x)$.

To derive the distribution of the unknown center $c = \frac{\alpha^2\ipd{\ary x}{\ary e}}{\alpha^2\|\ary e\|^2 + \beta^2 q^2}$, we 
observe that the distribution of $\ary x$ is smooth enough to be treated as a continuous Gaussian distribution since the eigenvalues of $\beta^2 q^2\mat \Sigma/2$ are $\beta^2 q^2/2$ and $(\beta^2 q^2 + \alpha^2\|\ary e\|^2)/2$. So the distribution of the unknown center $c$ can be approximated by the discrete Gaussian distribution with minimum step $\frac{\alpha^2}{\alpha^2\|\ary e\|^2 + \beta^2 q^2}$, center 0 and variance
\[
    \sigma_c^2 = \left(\frac{\alpha^2}{\alpha^2\|\ary e\|^2 + \beta^2 q^2}\right)^2\cdot \ary e^T\left(\beta^2 q^2\mat\Sigma/2\right)\ary e = \frac{\alpha^4\|\ary e\|^2}{2(\alpha^2\|\ary e\|^2 + \beta^2 q^2)}.
\]

In conclusion, we propose the following statement for the distribution of the unknown center in the Extrapolated DCP states:
\begin{theorem}
    The distribution of $c$ in \Cref{eqn:EDCP_form} of \Cref{thm:EDCP_Step1} approximately follows the discrete Gaussian distribution $D_{\frac{\alpha^2}{\alpha^2\|\ary e\|^2 + \beta^2 q^2} \Z, \sigma_c}$ where $\sigma_c = \frac{\alpha^2\|\ary e\|}{\sqrt{2(\alpha^2\|\ary e\|^2 + \beta^2 q^2)}}$.
\end{theorem}

\begin{remark}
    As readers may notice, the Gaussian width $\sigma$ of $j$ and the Gaussian width $\sigma_c$ of $c$ (the center of the distribution of $j$) satisfy $\sigma_c = \frac{\alpha\|\ary e\|}{\sqrt{2}\beta q}\sigma$. In our settings, if we assume $\beta q \gg \alpha\cdot\|\ary e\|$, then the distribution of $j$ is a discrete Gaussian distribution with a small shift. However, this shift is non-negligible, preventing our $\QLWE^{\sf phase}$ state from being exponentially close to a $\QLWE$ state without unknown phase. 
\end{remark}

\iffull
\section{Hardness of $\QLWE$ with Unknown Phase via Quantizing Regev's Iterative Reduction}\label{sec:regev}

In this section, we show how to reduce from the problem of generating discrete Gaussian states ($\QDGS$, \Cref{def:QDGS}) to a variant of $\QLWE$ with an unknown phase term on the amplitude ($\QLWE^{\sf phase}$), by modifying Regev's iterative reduction~\cite{DBLP:journals/jacm/Regev09} from the problem of generating discrete Gaussian samples ($\DGS$, \Cref{def:DGS}) to $\LWE$. Combined with the known reductions from $\GAP\SVP$ and $\SIVP$ to $\DGS$ in \Cref{lemma:GapSVP2DGS} and \Cref{lemma:SIVP2DGS}, it gives a quantum reduction from ${\GAP\SVP}_{\Tilde{O}(n^{1.5})}$ and ${\SIVP}_{\Tilde{O}(n^{1.5})}$ to $\QLWE^{\sf phase}$. 

\subsection{Overview of our reduction} 
As in \cite{DBLP:journals/jacm/Regev09}, our proof is iterative. We start from generating discrete Gaussian states with exponentially large widths (in Regev's reduction, it is classical discrete Gaussian sample with exponentially large widths; both can be done efficiently). Then, each iteration produces discrete Gaussian states (samples) with smaller widths. Repeating the iteration for polynomial number of times gives the discrete Gaussian states (samples) for the $\QDGS$ ($\DGS$) problem. We illustrate Regev's reduction in \Cref{subfig:regev_reduc}, aligned with our reduction in \Cref{subfig:q_regev_reduc}, and then explain with more details.

\iffull
\begin{figure}[htbp]
    \centering
    \begin{subfigure}[b]{0.49\textwidth}
        \centering
        \scalebox{0.75}{
        \centering
    \begin{tikzpicture}
        \def \h {1.15}
        \def \hb {1.3}
        \def \wl {3}
        \def \wm {3.5}
        \def \wr {3.8}

        \draw (0, 4*\h+\hb) rectangle (\wl, 4*\h)
            node[pos=0.5, align=center] 
            {$n^c$ samples \\ of $D_{\Lattice, r}$};
        \draw (0, 2*\h+\hb) rectangle (\wl, 2*\h)
            node[pos=0.5, align=center]
            {$n^c$ samples \\ of $D_{\Lattice, r\sqrt{n}/(\alpha q)}$};
        \draw (0, \hb) rectangle (\wl, 0)
            node[pos=0.5, align=center]
            {$n^c$ samples \\ of $D_{\Lattice, rn/(\alpha q)^2}$};
        \draw (\wl+\wm, 3*\h+\hb) rectangle (\wl+\wm+\wr, 3*\h)
            node[pos=0.5, align=center]
            {Quantum query \\ for $\CVP_{\Lattice^*, \alpha q/(\sqrt 2 r)}$};
        \draw (\wl+\wm, \h+\hb) rectangle (\wl+\wm+\wr, \h)
            node[pos=0.5, align=center]
            {Quantum query \\ for $\CVP_{\Lattice^*, (\alpha q)^2/(r\sqrt{2n})}$};

        \draw[-{Stealth[length=2mm, width=1.5mm]}] (\wl, 4*\h+0.3*\hb) -- (\wl+\wm, 3*\h+0.7*\hb);
        \draw[-{Stealth[length=2mm, width=1.5mm]}] (\wl+\wm, 3*\h+0.3*\hb) -- (\wl, 2*\h+0.7*\hb);
        \draw[-{Stealth[length=2mm, width=1.5mm]}] (\wl, 2*\h+0.3*\hb) -- (\wl+\wm, \h+0.7*\hb);
        \draw[-{Stealth[length=2mm, width=1.5mm]}] (\wl+\wm, \h+0.3*\hb) -- (\wl, 0.7*\hb);

        \path [decorate, decoration={text along path, text align=center, text={|\footnotesize|Classical, use {$\LWE$}}}] (\wl, 4*\h+0.3*\hb+0.15) -- (\wl+\wm, 3*\h+0.7*\hb+0.15);

        \path [decorate, decoration={text along path, text align=center, text={|\footnotesize|{\cite[Lemma 3.14]{DBLP:journals/jacm/Regev09}}}}] (\wl, 2*\h+0.7*\hb+0.15) -- (\wl+\wm, 3*\h+0.3*\hb+0.15);

        \path [decorate, decoration={text along path, text align=center, text={|\footnotesize|Classical, use {$\LWE$}}}] (\wl, 2*\h+0.3*\hb+0.15) -- (\wl+\wm, \h+0.7*\hb+0.15);

        \path [decorate, decoration={text along path, text align=center, text={|\footnotesize|{\cite[Lemma 3.14]{DBLP:journals/jacm/Regev09}}}}] (\wl, 0.7*\hb+0.15) -- (\wl+\wm, \h+0.3*\hb+0.15);

        \node at (\wl+0.5*\wm, 4*\h+\hb+0.3) {~{\Huge $\vdots$}~};
        \node at (\wl+0.5*\wm, 0.3) {~{\Huge $\vdots$}~};
    \end{tikzpicture}
        }
        \caption{Two iterations in Regev's reduction \cite{DBLP:journals/jacm/Regev09}.}
        \label{subfig:regev_reduc}
    \end{subfigure}
    \hfill
    \begin{subfigure}[b]{0.49\textwidth}
        \centering
        \scalebox{0.75}{
        \centering
    \begin{tikzpicture}
        \def \h {1.15}
        \def \hb {1.3}
        \def \wl {3}
        \def \wm {3.5}
        \def \wr {3.8}

        \draw (0, 4*\h+\hb) rectangle (\wl, 4*\h)
            node[pos=0.5, align=center] 
            {$3m^2n^2$ states \\ of $\ket{D_{\Lattice, r}}$};
        \draw (0, 2*\h+\hb) rectangle (\wl, 2*\h)
            node[pos=0.5, align=center]
            {$3m^2n^2$ states \\ of $\ket{D_{\Lattice, r\sqrt{n}/(\alpha q)}}$};
        \draw (0, \hb) rectangle (\wl, 0)
            node[pos=0.5, align=center]
            {$3m^2n^2$ states \\ of $\ket{D_{\Lattice, rn/(\alpha q)^2}}$};
        \draw (\wl+\wm, 3*\h+\hb) rectangle (\wl+\wm+\wr, 3*\h)
            node[pos=0.5, align=center]
            {Quantum query \\ for $\CVP_{\Lattice^*, \alpha q/r}$};
        \draw (\wl+\wm, \h+\hb) rectangle (\wl+\wm+\wr, \h)
            node[pos=0.5, align=center]
            {Quantum query \\ for $\CVP_{\Lattice^*, (\alpha q)^2/(r\sqrt{n})}$};

        \draw[-{Stealth[length=2mm, width=1.5mm]}] (\wl, 4*\h+0.3*\hb) -- (\wl+\wm, 3*\h+0.7*\hb);
        \draw[-{Stealth[length=2mm, width=1.5mm]}] (\wl+\wm, 3*\h+0.3*\hb) -- (\wl, 2*\h+0.7*\hb);
        \draw[-{Stealth[length=2mm, width=1.5mm]}] (\wl, 2*\h+0.3*\hb) -- (\wl+\wm, \h+0.7*\hb);
        \draw[-{Stealth[length=2mm, width=1.5mm]}] (\wl+\wm, \h+0.3*\hb) -- (\wl, 0.7*\hb);

        \path [decorate, decoration={text along path, text align=center, text={|\footnotesize|Use {$\QLWE^{\sf phase}$,}}}] (\wl+0.15, 4*\h+0.3*\hb+0.5) -- (\wl+\wm+0.15, 3*\h+0.7*\hb+0.5);
        \path [decorate, decoration={text along path, text align=center, text={|\footnotesize|({\Cref{thm:DGStoCVP}})}}] (\wl, 4*\h+0.3*\hb+0.15) -- (\wl+\wm, 3*\h+0.7*\hb+0.15);

        \path [decorate, decoration={text along path, text align=center, text={|\footnotesize|{\cite[Lemma 3.14]{DBLP:journals/jacm/Regev09}}}}] (\wl, 2*\h+0.7*\hb+0.15) -- (\wl+\wm, 3*\h+0.3*\hb+0.15);

        \path [decorate, decoration={text along path, text align=center, text={|\footnotesize|Use {$\QLWE^{\sf phase}$,}}}] (\wl+0.15, 2*\h+0.3*\hb+0.5) -- (\wl+\wm+0.15, \h+0.7*\hb+0.5);
        \path [decorate, decoration={text along path, text align=center, text={|\footnotesize|({\Cref{thm:DGStoCVP}})}}] (\wl, 2*\h+0.3*\hb+0.15) -- (\wl+\wm, \h+0.7*\hb+0.15);

        \path [decorate, decoration={text along path, text align=center, text={|\footnotesize|{\cite[Lemma 3.14]{DBLP:journals/jacm/Regev09}}}}] (\wl, 0.7*\hb+0.15) -- (\wl+\wm, \h+0.3*\hb+0.15);

        \node at (\wl+0.5*\wm, 4*\h+\hb+0.3) {~{\Huge $\vdots$}~};
        \node at (\wl+0.5*\wm, 0.3) {~{\Huge $\vdots$}~};
    \end{tikzpicture}
        }
        \caption{Two iterations in our reduction that quantizes Regev's reduction.}
        \label{subfig:q_regev_reduc}
    \end{subfigure}
    \caption{The correspondence between Regev's reduction (from $\DGS$ to $\LWE$) and our reduction (from $\QDGS$ to $\QLWE^{\sf phase}$). }
    \label{fig:reduc_roadmap}
\end{figure}
\fi

\ifllncs
\begin{figure}[p]
    \centering
    \begin{subfigure}[b]{\textwidth}
        \centering
        \scalebox{1}{
        \centering
    \begin{tikzpicture}
        \def \h {1.15}
        \def \hb {1.3}
        \def \wl {3}
        \def \wm {3.5}
        \def \wr {3.8}

        \draw (0, 4*\h+\hb) rectangle (\wl, 4*\h)
            node[pos=0.5, align=center] 
            {$n^c$ samples \\ of $D_{\Lattice, r}$};
        \draw (0, 2*\h+\hb) rectangle (\wl, 2*\h)
            node[pos=0.5, align=center]
            {$n^c$ samples \\ of $D_{\Lattice, r\sqrt{n}/(\alpha q)}$};
        \draw (0, \hb) rectangle (\wl, 0)
            node[pos=0.5, align=center]
            {$n^c$ samples \\ of $D_{\Lattice, rn/(\alpha q)^2}$};
        \draw (\wl+\wm, 3*\h+\hb) rectangle (\wl+\wm+\wr, 3*\h)
            node[pos=0.5, align=center]
            {Quantum query \\ for $\CVP_{\Lattice^*, \alpha q/(\sqrt 2 r)}$};
        \draw (\wl+\wm, \h+\hb) rectangle (\wl+\wm+\wr, \h)
            node[pos=0.5, align=center]
            {Quantum query \\ for $\CVP_{\Lattice^*, (\alpha q)^2/(r\sqrt{2n})}$};

        \draw[-{Stealth[length=2mm, width=1.5mm]}] (\wl, 4*\h+0.3*\hb) -- (\wl+\wm, 3*\h+0.7*\hb);
        \draw[-{Stealth[length=2mm, width=1.5mm]}] (\wl+\wm, 3*\h+0.3*\hb) -- (\wl, 2*\h+0.7*\hb);
        \draw[-{Stealth[length=2mm, width=1.5mm]}] (\wl, 2*\h+0.3*\hb) -- (\wl+\wm, \h+0.7*\hb);
        \draw[-{Stealth[length=2mm, width=1.5mm]}] (\wl+\wm, \h+0.3*\hb) -- (\wl, 0.7*\hb);

        \path [decorate, decoration={text along path, text align=center, text={|\footnotesize|Classical, use {$\LWE$}}}] (\wl, 4*\h+0.3*\hb+0.15) -- (\wl+\wm, 3*\h+0.7*\hb+0.15);

        \path [decorate, decoration={text along path, text align=center, text={|\footnotesize|{\cite[Lemma 3.14]{DBLP:journals/jacm/Regev09}}}}] (\wl, 2*\h+0.7*\hb+0.15) -- (\wl+\wm, 3*\h+0.3*\hb+0.15);

        \path [decorate, decoration={text along path, text align=center, text={|\footnotesize|Classical, use {$\LWE$}}}] (\wl, 2*\h+0.3*\hb+0.15) -- (\wl+\wm, \h+0.7*\hb+0.15);

        \path [decorate, decoration={text along path, text align=center, text={|\footnotesize|{\cite[Lemma 3.14]{DBLP:journals/jacm/Regev09}}}}] (\wl, 0.7*\hb+0.15) -- (\wl+\wm, \h+0.3*\hb+0.15);

        \node at (\wl+0.5*\wm, 4*\h+\hb+0.3) {~{\Huge $\vdots$}~};
        \node at (\wl+0.5*\wm, 0.3) {~{\Huge $\vdots$}~};
    \end{tikzpicture}
        }
        \caption{Two iterations in Regev's reduction \cite{DBLP:journals/jacm/Regev09}.}
        \label{subfig:regev_reduc}
    \end{subfigure}
    \hfill \\ \quad \\
    \begin{subfigure}[b]{\textwidth}
        \centering
        \scalebox{0.8}{
            \begin{tikzpicture}
        \def \h {1.15}
        \def \hb {1.3}
        \def \wl {3}
        \def \wm {3.5}
        \def \wr {3.8}

        \draw (0, 4*\h+\hb) rectangle (\wl, 4*\h)
            node[pos=0.5, align=center] 
            {$3m^2n^2$ states \\ of $\ket{D_{\Lattice, r}}$};
        \draw (0, 2*\h+\hb) rectangle (\wl, 2*\h)
            node[pos=0.5, align=center]
            {$3m^2n^2$ states \\ of $\ket{D_{\Lattice, r\sqrt{n}/(\alpha q)}}$};
        \draw (0, \hb) rectangle (\wl, 0)
            node[pos=0.5, align=center]
            {$3m^2n^2$ states \\ of $\ket{D_{\Lattice, rn/(\alpha q)^2}}$};
        \draw (\wl+\wm, 3*\h+\hb) rectangle (\wl+\wm+\wr, 3*\h)
            node[pos=0.5, align=center]
            {Quantum query \\ for $\CVP_{\Lattice^*, \alpha q/r}$};
        \draw (\wl+\wm, \h+\hb) rectangle (\wl+\wm+\wr, \h)
            node[pos=0.5, align=center]
            {Quantum query \\ for $\CVP_{\Lattice^*, (\alpha q)^2/(r\sqrt{n})}$};

        \draw[-{Stealth[length=2mm, width=1.5mm]}] (\wl, 4*\h+0.3*\hb) -- (\wl+\wm, 3*\h+0.7*\hb);
        \draw[-{Stealth[length=2mm, width=1.5mm]}] (\wl+\wm, 3*\h+0.3*\hb) -- (\wl, 2*\h+0.7*\hb);
        \draw[-{Stealth[length=2mm, width=1.5mm]}] (\wl, 2*\h+0.3*\hb) -- (\wl+\wm, \h+0.7*\hb);
        \draw[-{Stealth[length=2mm, width=1.5mm]}] (\wl+\wm, \h+0.3*\hb) -- (\wl, 0.7*\hb);

        \path [decorate, decoration={text along path, text align=center, text={|\footnotesize|Use {$\QLWE^{\sf phase}$,}}}] (\wl+0.15, 4*\h+0.3*\hb+0.5) -- (\wl+\wm+0.15, 3*\h+0.7*\hb+0.5);
        \path [decorate, decoration={text along path, text align=center, text={|\footnotesize|({\Cref{thm:DGStoCVP}})}}] (\wl, 4*\h+0.3*\hb+0.15) -- (\wl+\wm, 3*\h+0.7*\hb+0.15);

        \path [decorate, decoration={text along path, text align=center, text={|\footnotesize|{\cite[Lemma 3.14]{DBLP:journals/jacm/Regev09}}}}] (\wl, 2*\h+0.7*\hb+0.15) -- (\wl+\wm, 3*\h+0.3*\hb+0.15);

        \path [decorate, decoration={text along path, text align=center, text={|\footnotesize|Use {$\QLWE^{\sf phase}$,}}}] (\wl+0.15, 2*\h+0.3*\hb+0.5) -- (\wl+\wm+0.15, \h+0.7*\hb+0.5);
        \path [decorate, decoration={text along path, text align=center, text={|\footnotesize|({\Cref{thm:DGStoCVP}})}}] (\wl, 2*\h+0.3*\hb+0.15) -- (\wl+\wm, \h+0.7*\hb+0.15);

        \path [decorate, decoration={text along path, text align=center, text={|\footnotesize|{\cite[Lemma 3.14]{DBLP:journals/jacm/Regev09}}}}] (\wl, 0.7*\hb+0.15) -- (\wl+\wm, \h+0.3*\hb+0.15);

        \node at (\wl+0.5*\wm, 4*\h+\hb+0.3) {~{\Huge $\vdots$}~};
        \node at (\wl+0.5*\wm, 0.3) {~{\Huge $\vdots$}~};

       \draw[rounded corners] (\wl+\wm+0.5*\wr, 5*\h+2) rectangle (\wl+\wm+2.3*\wr, 5*\h-0.2)
            node[pos=0.5, align=left]
            {1. Guess the parameters for the\\ $\QLWE^{\sf phase}$ instance. (\Cref{lemma:SLWE2CVPmodq})\\2. Generate $\QLWE^{\sf phase}$ instance\\ with the correct guess. (\Cref{lemma:generate_samples})};

        \draw[densely dashed, ->] (\wl+\wm+0.5*\wr,5*\h+0.9) .. controls (\wl+\wm+0.5*\wr-0.2, 5*\h+0.9) and (\wl+\wm, 5*\h+\hb-0.5) .. (\wl+0.5*\wm+0.8, 4*\h+1);
        
    \end{tikzpicture}
        }
        \caption{Two iterations in our reduction that quantizes Regev's reduction.}
        \label{subfig:q_regev_reduc}
    \end{subfigure}
    \caption{The correspondence between Regev's reduction (from $\DGS$ to $\LWE$) and our reduction (from $\QDGS$ to $\QLWE^{\sf phase}$). }
    \label{fig:reduc_roadmap}
\end{figure}
\fi

In order to quantize Regev's iterative reduction, we focus on quantizing the only classical step in the reduction -- solving $\CVP$. Roughly speaking, given a $\CVP$ instance $\ary{x}$, Regev~\cite{DBLP:journals/jacm/Regev09} utilizes LWE oracle to solve $\CVP$ by feeding it with samples $\ary{a}:=\Lattice^{-1}\ary{v} \bmod q$ and $\ipd{\ary{x}}{\ary{v}} \bmod q$ where $\ary{v} \leftarrow D_{\Lattice, r}$, which are close to the LWE sample $\ary{a}$ and $\ipd{\ary{a}}{\ary{s}} + e \bmod q$ where $\ary{s} = (\Lattice^*)^{-1}\kappa_{\Lattice^*}(\ary{x})\bmod q$ and $e$ is sampled from the Gaussian distribution. To quantize this step, a natural idea is to replace the classical $\ary{v}$ with a superposition state of Gaussian samples $\sum_{\ary{v} \in \Lattice}\rho_r(\ary{v})\ket{\ary{v}}$, measure $\ary{a}=\Lattice^{-1}\ary{v} \bmod q$, and compute $\ipd{\ary{x}}{\ary{v}} \bmod q$ in another register, hoping that the register contains an $\QLWE$ state. However, we should be careful to make sure that the $\ary{v}$ register does not collapse to a classical $\ary{v}$. Our solution is to measure the $\ary{v}$ register in Fourier basis, which can ensure that each $\ary{v} \in q\Lattice + \Lattice \ary{a}$ appears in the amplitude of the $\ipd{\ary{x}}{\ary{v}}$ register. But it also inevitably introduces a phase term that we are unable to compute efficiently from the measurement results. The above discussion ignores the Gaussian distribution to smooth the error distribution. More details can be found in~\Cref{sec:generating_slwe_samples_regev}.

The above described reduction leads us to requiring an $\QLWE^{\sf phase}$ oracle with specific parameter. We first formally define our special
parameters and functions, and propose the main theorem for this section here:

\begin{definition}\label{def:SLWEparamsRegev}
    Let $\Lattice$ be an $n$-dimensional integer lattice. Given parameters $q, R\in \N^+$ such that $R\in 2^{\poly(n)}$, $\alpha, r\in\R^+$ and a vector $\ary x\in \R^n$ such that $\dist(\ary{x}, \Lattice^*) \le \lambda_1(\Lattice^*)/2$, we define
    \begin{enumerate}
        \item An amplitude function $f: \Z_{qR}/R\to \R$ with $f(e) = \rho_{\sqrt2\alpha q}(e)$ which is completely known.
        \item A family of distribution $D_\theta^{(r,\ary{x})}$ over $\Z_R^n \cap R \cdot B_{(q\Lattice)^*}$ parameterized by $(r, \ary{x})$ and given by $\Pr(\ary y) \propto \rho_{\sqrt{\mat\Sigma/2}}\left(\ary{z}\left(\ary{y}\right)\right)$ where $\mat \Sigma := \frac{\ary{I}}{r^2}+\frac{\ary{x}'{\ary{x}'}^T}{2\alpha^2q^2-r^2\|\ary{x}'\|^2}$, $\ary{x}':=\ary{x} - \kappa_{\Lattice^*}(\ary{x})$ and $\ary{z}(\ary{y}) := \ary y/R - \kappa_{(q\Lattice)^*}(\ary y/R)$.
        \item A family of phase function $\theta^{(r, \ary x)}: \Z_R^n \cap R \cdot B_{(q\Lattice)^*}\to \R$ parameterized by $(r, \ary x)$ with $\theta^{(r,\ary{x})}(\ary{y}) = \frac{r^2\ipd{\ary{x}'}{\ary{z}(\ary{y})}}{2\alpha^2q^2}$, where $\ary{x}'=\ary{x} - \kappa_{\Lattice^*}(\ary{x})$ and $\ary{z}(\ary{y}) = \ary y/R - \kappa_{(q\Lattice)^*}(\ary y/R)$. This function is not known to be efficiently computable since it requires to solve approximate $\CVP$.  
    \end{enumerate}
\end{definition}

\begin{theorem}[Main theorem, from $\QDGS$ to $\QLWE^{\sf phase}$]\label{thm:Regev_MainThm}
Let $\Lattice$ be an $n$-dimensional integer lattice.
Let $\epsilon = \epsilon(n)$ be a negligible function such that $\epsilon(n) < 2^{-n}$, $q = q(n) > 10n$ be an integer of at most $\poly(n)$ bits, $\alpha \in (0, \frac{1}{5\sqrt{n}})$ such that $\alpha q > 2\sqrt{n}$, $R = R(n)$ be an exponentially large integer such that $R > \max\{2^{2n +2}n\lambda_n(\Lattice)^2, \frac{2^{4n + 1}\sqrt{2}n\lambda_n(\Lattice^*)\lambda_n(\Lattice)}{\alpha q}, 2^{3n}\lambda_n(\Lattice^*)\}$. Let $r_0 > 4\sqrt{n}\eta_{\epsilon}(\Lattice)/\alpha$ be the width parameter of the $\QDGS$ problem. 

Assume there exists quantum algorithms that can solve $\QLWE^{\sf phase}_{n,m,q,f, \theta^{(r,\ary{x})}, D_{\theta}^{(r,\ary{x})}}$ for any choice of pair $(r,\ary{x})$ such that $\alpha q r_0/\sqrt{n} < r < 2^{2n}\sqrt{2}\lambda_n(\Lattice)$, $\ary{x}\in \Lattice^*/R$ and $\dist(\ary{x},\Lattice^*)\le \alpha q / r$, with $m = 2^{o(n)}$ samples and in time complexity $T$. Then there exists a quantum algorithm that can generate a state that is $2^{-\Omega(n)}$-close to the discrete Gaussian state $\ket{D_{\Lattice, r_0}} = \sum_{\ary{v}\in \Lattice}\rho_{r_0}(\ary{v})\ket{\ary{v}}$ in time complexity $O((m^4+m^3T)\poly(n))$. 
\end{theorem}

Then the $\QDGS$ problem is easily reduced to either $\GAP\SVP$ or $\SIVP$. The connection to $\GAP\SVP$ and $\SIVP$ is a Corollary of  Theorem~\ref{thm:Regev_MainThm} and Lemmas~\ref{lemma:smoothingMR07}, \ref{lemma:smoothingRegev09}, \ref{lemma:GapSVP2DGS}, and~\ref{lemma:SIVP2DGS}.

\begin{corollary}
    Under the same assumption used in \Cref{thm:Regev_MainThm}, there exists quantum algorithms for solving $\GAP\SVP_\gamma$ and $\SIVP_\gamma$ for $\gamma\in\tilde{O}(n/\alpha)$ in time complexity $\poly(n, m, T)$.
\end{corollary}

\begin{remark}\label{remark:theta_upper_bound_regev}
    Readers may think the assumption of~\Cref{thm:Regev_MainThm} looks too strong because the family of phase functions $\{\theta^{(r,\ary{x})}\}$ is a very large family. However, we know the absolute value of $\theta^{(r,\ary{x})}(\ary{y})$ is small with high probability, when $(r,\ary{x})$ follows the setting in \Cref{thm:Regev_MainThm} and $\ary{y}$ is sampled from the corresponding distribution $D_\theta^{(r,\ary{x})}$. This is because when $\ary{y}\leftarrow D_\theta^{(r,\ary{x})}$, $\ary{z}(\ary{y}) = \ary{y}/R - \kappa_{(q\Lattice)^*}(\ary{y}/R)$ follows the distribution $\rho_{\sqrt{\mat \Sigma/2}}(\ary{z}(\ary{y}))$ over support $\Z_R^n/R \cap B_{(q\Lattice)^*}$ and thus has $\ell_2$ norm at most $\frac{\sqrt{n}\alpha q}{r\sqrt{2\alpha^2 q^2 - r^2\|\ary{x'}\|^2}}$ with $1 - 2^{-\Omega(n)}$ probability. Then $\left|\theta^{(r,\ary{x})}(\ary{y})\right| \le \frac{\sqrt{n}}{2\alpha q}$ with $1 - 2^{-\Omega(n)}$ probability. As our algorithm in~\Cref{sec:kuper} can solve the problem in sub-exponential time if $\theta^{(r,\ary{x})}(\ary{y})$ is always zero, there might be a way to solve the problem when $\theta^{(r,\ary{x})}(\ary{y})$ is close to zero.
\end{remark}

In what follows, we will display the idea of our proof for \Cref{thm:Regev_MainThm}, which will be an expansion of the proof idea described earlier with \Cref{fig:reduc_roadmap}. As is described before, our proof is iterative. We start from generating discrete Gaussian states with exponentially large widths. Then, equipped with an $\QLWE^{\sf phase}$ solver, we iteratively generate discrete Gaussian states with smaller widths in each step. Repeating the iterative step for polynomial number of times gives the discrete Gaussian states we desired. Formally, our initialization and iterative steps are:  

\begin{theorem}[The initialization step, {\cite[Lemma 3.12]{DBLP:journals/jacm/Regev09}}]\label{thm:initialization}
There exists an efficient quantum algorithm that given any $n$-dimensional integer lattice $\Lattice$ and width $r>2^{2n}\sqrt{2}\lambda_n(\Lattice)$, output a state that is $2^{-\Omega(n)}$-close to the state $\ket{D_{\Lattice,r} }= \sum_{\ary{v}\in \Lattice}\rho_r(\ary{v})$. 
\end{theorem}

\begin{theorem}[The iterative step]\label{thm:iterative}
Let $\Lattice$ be an $n$-dimensional integer lattice, $q > 2$ be an integer. Define the parameters $\epsilon \in (0, 2^{-n})$, $\alpha \in (0, \frac{1}{5\sqrt{n}})$, $r > 4 q\eta_{\epsilon}(\Lattice)$, and a precision parameter $R > \max\{2\sqrt{n}r\sqrt{\log r}, \frac{2\sqrt{n}}{\alpha q}, \frac{2^{2n + 1}nr\lambda_n(\Lattice^*)}{\alpha q}, 2^{3n}\lambda_n(\Lattice^*)\}$ as an integer. 

Assume that there exists a quantum algorithm that solves $\QLWE^{\sf phase}_{n,m,q,f, \theta^{(r,\ary{x})}, D_{\theta}^{(r,\ary{x})}}$ for any $\ary x\in \Lattice^*/R$ with ${\rm dist}(\ary x, \Lattice^*) < \alpha q/r$ in time complexity $T$. Then there exists a quantum algorithm that, given $3m^2n^2$ discrete Gaussian states $\ket{D_{\Lattice, r}} = \sum_{\ary{v} \in \Lattice}\rho_{r}(\ary{v})\ket{\ary{v}}$, produces $3m^2n^2$ discrete Gaussian states that are $2^{-\Omega(n)}$-close to $\ket{D_{\Lattice, r\sqrt{n}/\alpha q}}$, in time complexity $O((m^4 + m^3 T)\poly(n))$.
\end{theorem}

The iterative step consists of two steps:
\begin{enumerate}
\item[Step 1] Given a $\CVP$ instance, we can use a collection of discrete Gaussian states $\ket{D_{\Lattice, r}}$ to construct an $\QLWE^{\sf phase}$ instance. Solving the $\QLWE^{\sf phase}$ instance will in return solve the $\CVP$ problem. More precisely, we show the following theorem:
\begin{theorem}\label{thm:DGStoCVP}
    Let $\Lattice$ be an $n$-dimensional integer lattice, define the parameters $\epsilon \in (0, 2^{-n})$, $\alpha \in (0, \frac{1}{5\sqrt{n}})$, $r > 4 q\eta_{\epsilon}(\Lattice)$, and a precision parameter $R > \max\{2\sqrt{n}r\sqrt{\log r}, \frac{2\sqrt{n}}{\alpha q}, \frac{2^{2n + 1}nr\lambda_n(\Lattice^*)}{\alpha q}\}$ as an integer. 
    
    Assume that there exists an quantum algorithm that solves $\QLWE^{\sf phase}_{n,m,q,f, \theta^{(r,\ary{x})}, D_{\theta}^{(r,\ary{x})}}$ for any $\ary x\in \Lattice^*/R$ with ${\rm dist}(\ary x, \Lattice^*) < \alpha q/r$ in time complexity $T$. Then there exists a quantum algorithm that, given $3m^2n^2$ discrete Gaussian states $\ket{D_{\Lattice, r}} = \sum_{\ary{v} \in \Lattice}\rho_{r}(\ary{v})\ket{\ary{v}}$, answers quantum query to $\CVP_{\Lattice^*, \alpha q / r}$ on the support $\Lattice^*/R$ (denoted by $\ket{\ary x, \ary y} \to \ket{\ary x, \ary y + \kappa_{\Lattice^*}(\ary x)}$ with $\ary{x} \in \Lattice^*/R$ such that $\dist(\ary{x}, \Lattice^*) \le \alpha q/r$) up to exponentially small error, with exponentially small disturbance to the states $\ket{D_{\Lattice, r}}$, and in time $O((m^2 + mT)\poly(n))$.
\end{theorem}
\item[Step 2] (Same as the quantum step in Regev's reduction) A query to the $\CVP$ oracle can help to generate a discrete Gaussian state with a smaller width. More precisely:
\begin{theorem}[{\cite[Lemma 3.14]{DBLP:journals/jacm/Regev09}}]\label{thm:CVPtoDGS}
There exists an efficient quantum algorithm that, given any $n$-dimensional lattice $\Lattice$, a number $d<\lambda_1(\Lattice^*)/2$ and an integer $R > 2^{3n}\lambda_n(\Lattice^*)$, outputs $\ket{ D_{\Lattice, \sqrt{n}/d} } = \sum_{\ary{v} \in \Lattice}\rho_{\sqrt{n}/d}(\ary{v})\ket{\ary{v}}$, with only one quantum query on the second register of state
\[
    \sum_{\ary{x} \in \Lattice^*/R, \|\ary{x}\| \le d}\rho_{d/\sqrt{n}}(\ary{x})\ket{\ary x, \ary{x} \bmod \mathcal{P}(\Lattice^*)},
\]
to the $\CVP_{\Lattice^*, d}$ oracle, which is on the support $\Lattice^*/R$. 
\end{theorem}
\end{enumerate}

\iffull
The full picture of the proof for the main reduction \Cref{thm:Regev_MainThm} is illustrated in \Cref{fig:regev_detailed}.  
\begin{figure}[htbp]
\centering
    \begin{tikzpicture}
        \def \h {1.15}
        \def \hb {1.3}
        \def \wl {3}
        \def \wm {3.5}
        \def \wr {3.8}

        \draw (0, 4*\h+\hb) rectangle (\wl, 4*\h)
            node[pos=0.5, align=center] 
            {$3m^2n^2$ states \\ of $\ket{D_{\Lattice, r}}$};
        \draw (0, 2*\h+\hb) rectangle (\wl, 2*\h)
            node[pos=0.5, align=center]
            {$3m^2n^2$ states \\ of $\ket{D_{\Lattice, r\sqrt{n}/(\alpha q)}}$};
        \draw (0, \hb) rectangle (\wl, 0)
            node[pos=0.5, align=center]
            {$3m^2n^2$ states \\ of $\ket{D_{\Lattice, rn/(\alpha q)^2}}$};
        \draw (\wl+\wm, 3*\h+\hb) rectangle (\wl+\wm+\wr, 3*\h)
            node[pos=0.5, align=center]
            {Quantum query \\ for $\CVP_{\Lattice^*, \alpha q/r}$};
        \draw (\wl+\wm, \h+\hb) rectangle (\wl+\wm+\wr, \h)
            node[pos=0.5, align=center]
            {Quantum query \\ for $\CVP_{\Lattice^*, (\alpha q)^2/(r\sqrt{n})}$};

        \draw[-{Stealth[length=2mm, width=1.5mm]}] (\wl, 4*\h+0.3*\hb) -- (\wl+\wm, 3*\h+0.7*\hb);
        \draw[-{Stealth[length=2mm, width=1.5mm]}] (\wl+\wm, 3*\h+0.3*\hb) -- (\wl, 2*\h+0.7*\hb);
        \draw[-{Stealth[length=2mm, width=1.5mm]}] (\wl, 2*\h+0.3*\hb) -- (\wl+\wm, \h+0.7*\hb);
        \draw[-{Stealth[length=2mm, width=1.5mm]}] (\wl+\wm, \h+0.3*\hb) -- (\wl, 0.7*\hb);

        \path [decorate, decoration={text along path, text align=center, text={|\footnotesize|Use {$\QLWE^{\sf phase}$,}}}] (\wl+0.15, 4*\h+0.3*\hb+0.5) -- (\wl+\wm+0.15, 3*\h+0.7*\hb+0.5);
        \path [decorate, decoration={text along path, text align=center, text={|\footnotesize|({\Cref{thm:DGStoCVP}})}}] (\wl, 4*\h+0.3*\hb+0.15) -- (\wl+\wm, 3*\h+0.7*\hb+0.15);

        \path [decorate, decoration={text along path, text align=center, text={|\footnotesize|{\cite[Lemma 3.14]{DBLP:journals/jacm/Regev09}}}}] (\wl, 2*\h+0.7*\hb+0.15) -- (\wl+\wm, 3*\h+0.3*\hb+0.15);

        \path [decorate, decoration={text along path, text align=center, text={|\footnotesize|Use {$\QLWE^{\sf phase}$,}}}] (\wl+0.15, 2*\h+0.3*\hb+0.5) -- (\wl+\wm+0.15, \h+0.7*\hb+0.5);
        \path [decorate, decoration={text along path, text align=center, text={|\footnotesize|({\Cref{thm:DGStoCVP}})}}] (\wl, 2*\h+0.3*\hb+0.15) -- (\wl+\wm, \h+0.7*\hb+0.15);

        \path [decorate, decoration={text along path, text align=center, text={|\footnotesize|{\cite[Lemma 3.14]{DBLP:journals/jacm/Regev09}}}}] (\wl, 0.7*\hb+0.15) -- (\wl+\wm, \h+0.3*\hb+0.15);

        \node at (\wl+0.5*\wm, 4*\h+\hb+0.3) {~{\Huge $\vdots$}~};
        \node at (\wl+0.5*\wm, 0.3) {~{\Huge $\vdots$}~};

       \draw[rounded corners] (\wl+\wm+0.5*\wr, 5*\h+2) rectangle (\wl+\wm+2.3*\wr, 5*\h-0.2)
            node[pos=0.5, align=left]
            {1. Guess the parameters for the\\ $\QLWE^{\sf phase}$ instance. (\Cref{lemma:SLWE2CVPmodq})\\2. Generate $\QLWE^{\sf phase}$ instance\\ with the correct guess. (\Cref{lemma:generate_samples})};

        \draw[densely dashed, ->] (\wl+\wm+0.5*\wr,5*\h+0.9) .. controls (\wl+\wm+0.5*\wr-0.2, 5*\h+0.9) and (\wl+\wm, 5*\h+\hb-0.5) .. (\wl+0.5*\wm+0.8, 4*\h+1);
        
    \end{tikzpicture}
\caption{Illustration of two iterations of the reduction algorithm.}
\label{fig:regev_detailed}
\end{figure}
\fi
\ifllncs
The full picture of the proof for the main reduction \Cref{thm:Regev_MainThm} is illustrated in \Cref{subfig:q_regev_reduc}.
\fi
The proof starts with the initial step (\Cref{thm:initialization}) then applies the iterative step (\Cref{thm:iterative}) for $\poly(n)$ times. The iterative step consists of two parts, first the construction of $\CVP$ oracle (\Cref{thm:DGStoCVP}) with the help of discrete Gaussian states from the previous iteration and the help of an $\QLWE^{\sf phase}$ solver, and then the generation of a narrower discrete Gaussian state through one query to the $\CVP$ oracle (\Cref{thm:CVPtoDGS}). Since the discrete Gaussian states generated in the previous iteration are only disturbed by an exponentially small amount upon each query of the $\CVP$ oracle, we can reuse them for $3m^2n^2$ times to construct $3m^2n^2$ narrower discrete Gaussian states for use of the next iteration.

It is left to prove \Cref{thm:DGStoCVP}, to which the remainder of this section will be devoted. To prove it, we start by generating an $\QLWE^{\sf phase}$ instance but with an unknown Gaussian width, instead of the fixed and known width $\sqrt{2}\alpha q$ in the $\QLWE^{\sf phase}$ solver, as displayed in \Cref{lemma:generate_samples} in \Cref{sec:generating_slwe_samples_regev}. We then address and resolve the issue of the unknown width in order to solve the $\sf CVP$ instance using the $\QLWE^{\sf phase}$ oracle, in the proof of \Cref{lemma:SLWE2CVPmodq} in \Cref{sec:completion_iterative}. Finally, we note that the procedure in \Cref{lemma:SLWE2CVPmodq} answers the $\CVP$ quantum query with $1-2^{-\Omega(n)}$ probability. Therefore using the idea of gentle measurement, we can answer each $\CVP$ quantum query with exponentially small disturbance to the states $\ket{D_{\Lattice, r}}$, as discussed in \Cref{sec:answerquery}.

\subsection{Generating the $\QLWE^{\sf phase}$ samples}\label{sec:generating_slwe_samples_regev}

In this subsection, we show how to create the $\QLWE^{\sf phase}$ instance for \Cref{thm:DGStoCVP} but with an unknown Gaussian width. Given a $\CVP_{\Lattice^*, \alpha q / r}$ instance, the idea is to replace the classical Gaussian samples from $D_{\Lattice,r}$ in~\cite{DBLP:journals/jacm/Regev09} (that helps to produce the $\LWE$ instance) with a superposition state of Gaussian samples $\ket{D_{\Lattice,r}} := \sum_{\ary{v}\in\Lattice}\rho_r(\ary{v})\ket{\ary{v}}$ that helps to produce the $\QLWE^{\sf phase}$ instance.

\begin{theorem}\label{lemma:generate_samples}
Let $\Lattice$ be an $n$-dimensional integer lattice, define the parameters $\epsilon \in (0, 2^{-n})$, $\alpha \in (0, \frac{1}{5\sqrt{n}})$, $\sigma \in [\alpha q, \sqrt{2}\alpha q]$, $r > 4q\eta_{\epsilon}(\Lattice)$, and a precision parameter $R > 2\sqrt{n}r \sqrt{\log r}$ as an integer. Given a $\CVP_{\Lattice^*, \alpha q / r}$ instance $\ary{x} \in \Lattice^*/R$ and a state $\ket{D_{\Lattice, r}}$ as input, there exists an efficient quantum algorithm that generates a random vector $\ary a \gets \mathcal U(\Z_q^n)$ and a state $2^{-\Omega(n)}$-close to the following state
\[
    \gamma_{t}^{\ary a} = \sum_{\substack{\ary{y} \in \Z_R^n \cap R\cdot B_{(q\Lattice)^*}}}\rho_{\sqrt{\mat\Sigma / 2}}(\ary{z}(\ary{y}))\kb{\ary{y}}{\ary{y}} \otimes \ket{\psi^{\ary{a}, \ary{y}}_{t}}\bra{\psi^{\ary{a}, \ary{y}}_{t}}
\]
where $t = \sqrt{\sigma^2+r^2\|\ary{x}'\|^2}$, $\ary{s} = (\Lattice^*)^{-1}\kappa_{\Lattice^*}(\ary{x})\bmod q$, $\ary{x'} = \ary{x} - \kappa_{\Lattice^*}(\ary{x})$, $\ary{z}(\ary{y}) = \ary y/R - \kappa_{(q\Lattice)^*}(\ary{y}/R)$, $\mat\Sigma = \frac{\mat I_n}{r^2} + \frac{\ary{x}'\ary{x}'^T}{\sigma^2}$, and the state $\ket{\psi^{\ary{a}, \ary{y}}_{t}}$ is an $\QLWE^{\sf phase}$ state
\begin{equation}\label{def:psitay}
    \ket{\psi^{\ary{a}, \ary{y}}_{t}}:=\sum_{u\in\Z_{qR}/R}\rho_{t}(u)\exp\left(2\pi {\rm i}\cdot u\frac{r^2\ipd{\ary x'}{\ary z(\ary{y})}}{t^2}\right)\ket{\ipd{\ary{s}}{\ary{a}}+u\bmod q}.
\end{equation}
\end{theorem}

\begin{proof}
From now on, when it is clear from the context, we use $\ary{z}$ to denote the value $\ary{z}(\ary{y})$, which actually depends on $\ary{y}$.

Here is the procedure of our quantum algorithm. For simplicity, we ignore the normalization factors when writing down superposition states. 

\begin{enumerate}
    \item\label{item:qlwe_single_instance_init} Prepare the initial state 
    $$\ket{D_{\Lattice, r}} \otimes \sum_{e\in \Z_{qR}/R}\rho_{\sigma}(e)\ket{e},$$
    which is $2^{-\Omega(n)}$-close to  $\sum_{\ary{v}\in \Lattice, \|\ary{v}\| \le \sqrt{n}r}\rho_{r}(\ary{v})\ket{\ary{v}} \otimes \sum_{e\in \Z_{qR}/R}\rho_{\sigma}(e)\ket{e}$ by Banaszczyk's Gaussian tail bound.
    \item\label{item:measure_a} 
    Measure $\ary{a}:=\Lattice^{-1}\ary{v} \bmod q$ to get an outcome $\ary{a}$ and a result state $2^{-\Omega(n)}$-close to  
    $$\sum_{\ary{v}\in q\Lattice + \Lattice \ary{a}, \|\ary{v}\| \le \sqrt{n}r}\rho_{r}(\ary{v})\ket{\ary{v}} \otimes \sum_{e\in\Z_{qR}/R}\rho_{\sigma}(e)\ket{e}$$
    According to \Cref{lemma:regevclaim3.8}, when $r/\sqrt 2>q\eta_\epsilon(\Lattice)$,  
    the distribution of $\ary{a}$ is $2^{-\Omega(n)}$-close to uniform.

    \item\label{item:add_xv_to_e} Apply a unitary to add the inner product $\ipd{\ary{x}}{\ary{v}} \bmod q$ to the last register and get \begin{equation}\label{eqn:ql+lawithphase0}
        \sum_{\ary{v}\in q\Lattice + \Lattice\ary{a}, \|\ary{v}\| \le \sqrt{n} r} \rho_{r}(\ary{v})\ket{\ary{v}} \otimes \sum_{e\in\Z_{qR}/R} \rho_{\sigma}(e)\ket{\ipd{\ary{s}}{\ary{a}} + \ipd{\ary{x'}}{\ary{v}}+e \bmod q}.
    \end{equation}
    Note that since we assumed $\ary{x} \in \Lattice^*/R$, the second register always has its value in $\Z_{qR}/R$. 

    Intuitively, since $u:= \ipd{\ary{x'}}{\ary{v}} + e \le \alpha q\sqrt{n} + \sigma \sqrt{n} < q/2$ with high probability and $R > 2\sqrt{n}r$, \Cref{eqn:ql+lawithphase0} should be $2^{-\Omega(n)}$-close to the following state displayed in~\Cref{eqn:ql+lawithphase0.5}. The proof is deferred to~\Cref{sec:proof_change_variable}.
    \begin{lemma}\label{lem:change_variable}
    Suppose that $q/2 > \alpha q\sqrt{n} + \sigma\sqrt{n}$, $R > 2\sqrt{n}r \sqrt{\log r}$ and $r > \sqrt{2} q \eta_{\epsilon}(\Lattice)$ for $\epsilon < 2^{-n}$, then the state in~\Cref{eqn:ql+lawithphase0} is $2^{-\Omega(n)}$-close to the state
    \begin{equation}\label{eqn:ql+lawithphase0.5}
    \sum_{\ary{v}\in q\Lattice + \Lattice\ary{a}} \rho_{r}(\ary{v})\ket{\ary{v} \bmod R} \otimes \sum_{u\in\Z_{qR}/R} \rho_{\sigma}(u - \ipd{\ary{x'}}{\ary{v}})\ket{\ipd{\ary{s}}{\ary{a}} + u \bmod q}.
    \end{equation}
    \end{lemma}
    \begin{remark}
        Similar to \Cref{lem:EDCP_state_close}, the condition of this lemma can be relaxed to $q/2>\alpha q\sqrt{n}+\sigma\sqrt{n}$ and $R>2\sqrt{n} r$ if we instead use the proof technique from Claim A.5 in~\cite{regev2023efficient}, with a slight modification on the lemma statement. Either way, we have the same asymptotic values for the parameters in the end.  
    \end{remark}

    By the assumption that $\sigma \le \sqrt{2}\alpha q$, we have that $\alpha q\sqrt{n} + \sigma\sqrt{n} \le (1 + \sqrt{2})\alpha q\sqrt{n} < q/2$. Therefore, according to the lemma mentioned above, we can obtain a state that is $2^{-\Omega(n)}$-close to the state displayed in~\Cref{eqn:ql+lawithphase0.5}.

    \item\label{item:qft}
    Recall that $\omega_R = e^{2 \pi {\rm i} / R}$. Applying $\QFT_R$ to the first register, we can get a state $2^{-\Omega(n)}$-close to 
    \begin{equation}\label{eqn:ql+lawithphase1}
        \sum_{\ary{y}\in\Z_R^n} \sum_{\ary{v}\in q\Lattice+\Lattice\ary{a}} \rho_{r}(\ary{v}) \cdot\omega_R^{\ipd{\ary{v}}{\ary{y}}} \ket{\ary{y}} \otimes \sum_{u\in\Z_{qR}/R} \rho_{\sigma}(u - \ipd{\ary{x'}}{\ary{v}})\ket{\ipd{\ary{s}}{\ary{a}} + u \bmod q}.
    \end{equation}

    We show that the state in \Cref{eqn:ql+lawithphase1} is $2^{-\Omega(n)}$-close to the following state displayed in \Cref{eqn:ql+lawithphase2}. The proof is deferred to~\Cref{sec:proof_qftstate_rewrite}.

    \begin{lemma}\label{lem:qftstate_rewrite}
    Suppose that $\epsilon < 2^{-n}, R > 2\sqrt{n}r\sqrt{\log r}, r > 4q \eta_\epsilon(\Lattice)$ and $\sigma \ge \alpha q$, then the state in~\Cref{eqn:ql+lawithphase1} is $2^{-\Omega(n)}$-close to the state 
    \begin{equation}\label{eqn:ql+lawithphase2}
        \begin{split}
        & \sum_{\substack{\ary{y}\in \Z_R^n \cap R\cdot B_{(q\Lattice)^*}}}\rho_{\sqrt{\mat\Sigma}}(\ary{z})\exp\left(2\pi {\rm i} \ipd{\Lattice\ary{a}}{\kappa_{(q\Lattice)^*}(\ary{y}/R)}\right)\ket{\ary{y}} \\
        & \qquad \qquad \otimes \sum_{u\in\Z_{qR}/R} \rho_{t}(u)\exp\left(2\pi {\rm i}\cdot u\frac{r^2\ipd{\ary{x'}}{\ary{z}}}{t^2}\right)\ket{\ipd{\ary{s}}{\ary{a}}+u\bmod q},
        \end{split}
    \end{equation}
    where $t, \ary{s}, \ary{z}, \mat\Sigma$ are specified in~\Cref{lemma:generate_samples}.
\end{lemma}

    By the assumption that $\sigma \in [\alpha q, \sqrt{2}\alpha q]$, we have that $\alpha q\sqrt{n} + \sigma\sqrt{n} < (1 + \sqrt{2})\alpha q\sqrt{n} < (1 + \sqrt{2})q/5 < q/2$. Therefore, according to \Cref{lem:qftstate_rewrite}, we can obtain a state that is $2^{-\Omega(n)}$-close to the state (recall the definition of $\ket{\psi^{\ary{a},\ary{y}}_{t}}$ in \Cref{def:psitay})
    \[
        \sum_{\substack{\ary{y} \in \Z_R^n \cap R\cdot B_{(q\Lattice)^*}}}\rho_{\sqrt{\mat\Sigma}}(\ary{z})\exp\left(2\pi {\rm i} \ipd{\Lattice\ary{a}}{\kappa_{(q\Lattice)^*}(\ary{y}/R)}\right)\ket{\ary{y}}\otimes\ket{\psi^{\ary{a},\ary{y}}_{t}}
    \]

    \item\label{item:measure_y} Measure the first register to get a state $2^{-\Omega(n)}$-close to the state $\gamma_t^{\ary a}$.\footnote{We measure the first register only for a better illustration of the residual state. In fact, we'll use the gentle measurement principle to defer all measurements, as described later in \Cref{sec:answerquery}. }
\end{enumerate}
Finally, combining the measurement result $\ary{a}$ from step \ref{item:measure_a} and the state $\gamma_t^{\ary a}$ from step \ref{item:measure_y} gives the desired sample.
\end{proof}

\begin{remark}
    An important caveat is that we should not discard the $\ary{y}$ register and hope we can solve the problem given only $\ary{a}$ and the $\QLWE^{\sf phase}$ state $\ket{\psi_{\sqrt{2}\alpha q}^{\ary{a}, \ary{y}}}$ whose error amplitude is Gaussian with a small phase, because such a solver is so strong that it solves $\LWE$ directly. This is because such a solver utilizes no information about $\ary{y}$ (the seed that generates $\theta$). So it should output $\ary{s}$ given samples $\ary{a}$ and the second register after step~\ref{item:add_xv_to_e} (Note that in the actual procedure the solver is given the second register after step~\ref{item:measure_y}, but step~\ref{item:qft} and step~\ref{item:measure_y} are both local operations acting on the first register, which should not influence the state of the second register).
    This leads to an algorithm that outputs $\ary{s}$ given samples $\left(\ary{a}, \ipd{\ary{s}}{\ary{a}} + e \bmod q\right)$ where $\ary{a} \leftarrow \mathcal{U}_{\Z_q^n}$ and $e = \ipd{\ary{x'}}{\ary{v}}$ for $\ary{v}$ distributed according to $D_{q\Lattice + \Lattice\ary{a}, r/\sqrt{2}}$ and a fixed vector $\|\ary{x'}\| \le \frac{\alpha q}{r}$. As the distribution of $e$ is close to a Gaussian distribution, this will give a surprising method to solve $\LWE$.

    This also explains why we describe the auxiliary information $\ary{y}$ carefully in~\Cref{def:SLWEphase} and ~\Cref{def:SLWEparamsRegev} instead of discarding $\ary{y}$ and strengthening the solver in~\Cref{thm:Regev_MainThm} to solve $\QLWE^{\sf phase}_{n,m,q,f,\theta,D}$ for any unknown $\theta$ such that $|\theta|\le \frac{\sqrt{n}}{2\alpha q}$ (see~\Cref{remark:theta_upper_bound_regev}) and an arbitrary distribution $D$.

    Given that we must use the information of $\ary{y}$, one may hope to compute $\ary{z}(\ary{y})$ to learn something about the phase $\theta^{(r, \ary{x})}(\ary{y}) = \frac{r^2\ipd{\ary{x}'}{\ary{z}(\ary{y})}}{2\alpha^2q^2}$. However, $\ary{z}(\ary{y})$ has $\ell_2$ norm roughly $\frac{\sqrt{n}}{r}$ (see~\Cref{remark:theta_upper_bound_regev}). So computing $\ary{z}(\ary{y})$ from $\ary{y}$ is a $\CVP_{\Lattice^*, \frac{q\sqrt{n}}{r}}$ instance, which is even harder than the goal of this iteration (a quantum query to $\CVP_{\Lattice^*, \alpha q/r}$). How to utilize $\ary{y}$ requires more attempts and is an important step towards our ultimate goal of solving standard $\LWE$ efficiently.
\end{remark}

\subsection{Dealing with the unknown Gaussian width}\label{sec:completion_iterative}

Now that we know how to generate $(\ary{a}, \gamma_t^{\ary{a}})$ which resembles an $\QLWE^{\sf phase}$ instance. However, the $\QLWE^{\sf phase}$ solver requires instances with an error distribution of the fixed width $\sqrt{2}\alpha q$. To bridge this gap, we experiment with different values of $\sigma$ to obtain a suitable width that is sufficiently close to $\sqrt{2}\alpha q$. Equipped with the $\QLWE^{\sf phase}$ solver, we can in turn solve the $\CVP$ problem. We realize this idea in the proof of the following theorem: 

\begin{theorem}\label{lemma:SLWE2CVPmodq}
Let $\Lattice$ be an $n$-dimensional integer lattice, define the parameters $\epsilon \in (0, 2^{-n})$, $\alpha \in (0, \frac{1}{5\sqrt{n}})$, $r > 4 q\eta_{\epsilon}(\Lattice)$, and a precision parameter $R > \max\{2\sqrt{n}r\sqrt{\log r}, \frac{2\sqrt{n}}{\alpha q}\}$ as an integer. 

Assume that there exists a quantum algorithm $\mathcal{A}$ that, given $m$ samples of independently uniformly random vector $\ary a\in \Z_q^n$ and state $\gamma^{\ary{a}}_{\sqrt{2}\alpha q}$, solves the secret vector $\ary s$ in time complexity $T$ with probability $1 - 2^{-\Omega(n)}$. Then there exists a quantum algorithm that, given a $\CVP_{\Lattice^*, \alpha q / r}$ instance $\ary{x} \in \Lattice^*/R$ and $3m^2n$ states $\ket{D_{\Lattice, r}}$ as input, outputs $\ary{s} = (\Lattice^*)^{-1}\kappa_{\Lattice^*}(\ary{x})\bmod q$ with probability $1 - 2^{-\Omega(n)}$ in time $O((m^2+mT)\poly(n))$.
\end{theorem}

\begin{proof}
Let $\sigma' = \sqrt{2\alpha^2 q^2 - r^2\|\ary{x}'\|^2}$ and $\sigma_i = \alpha q \left(1 + (\sqrt{2} - 1)\frac{i}{2m}\right), t_i = \sqrt{\sigma_i^2+r^2\|\ary{x}'\|^2}$ for $i = 0, 1, \cdots, 2m$. In this condition, we have $\sigma_i\in [\alpha q, \sqrt{2}\alpha q]$. Since $r\|\ary{x}'\|<\alpha q$, there must exist an index $0 \le j < 2m$ such that $\sigma_j < \sigma' \le \sigma_{j+1}$.

Then $\sigma_j$ is a suitable value of $\sigma$ to generate samples for the quantum algorithm $\mathcal{A}$. Formally,
\begin{lemma}\label{lemma:solveprob1/2}
    The quantum algorithm $\mathcal{A}$ can solve $\ary s$ with probability at least $1 / 2$, when given $m$ independent samples of vector $\ary a\in \Z_q^n$ and state $\gamma_{t_j}^{\ary a}$.
\end{lemma}

Before the proof of the lemma, let's first show how to construct a $\CVP$ algorithm based on the lemma. Here is the procedure of our quantum algorithm.
\begin{enumerate}
    \item[1.] (Generate classical LWE samples for verification of the solution) Apply \Cref{lemma:generate_samples} to $n$ states $\ket{D_{\Lattice, r}}$ with $\sigma = \sigma_0$ to obtain $n$ samples of vectors $\ary{a}\in \Z_q^n$ and states $\gamma_{t_0}^{\ary{a}}$. Then, measure the second register of $\gamma_{t_0}^{\ary{a}}$ to obtain $n$ classical LWE samples of the form $\ipd{\ary{a}}{\ary{s}} + u \bmod q$.
    \item[2.] Enumerate $\sigma$ from the set $\{\sigma_i: i\in \{0, 1, \cdots, 2m - 1\}\}$. 
    \item[3.] For each $\sigma = \sigma_i$, apply \Cref{lemma:generate_samples} to $mn$ states $\ket{D_{\Lattice, r}}$ to obtain $mn$ samples of vectors $\ary{a}\in \Z_q^n$ and states $\gamma_{t_i}^{\ary{a}}$ with a precision of $1 - 2^{-\Omega(n)}$.
    \item[4.] Utilize the quantum algorithm $\mathcal{A}$ on a group of $m$ samples of vectors $\ary{a}\in \Z_q^n$ and states $\gamma_{t_i}^{\ary{a}}$ to derive a solution $\ary{s}'$.
    \item[5.] Employ any verification process (e.g., as proposed by Regev in \cite[Lemma 3.6]{DBLP:journals/jacm/Regev09}) with the assistance of the $n$ classical LWE samples obtained in step 1 to check whether $\ary{s}' = \ary{s}$. If this condition holds, output $\ary{s}'$ and conclude the process.
\end{enumerate}

This procedure will use a maximum of $2m^2n + n<3m^2n$ samples of $\ket{D_{\Lattice, r}}$ and operates with a time complexity of $O((m^2 + mT)\poly(n))$. Moreover, it will output the correct $\ary{s}$ with a probability of at least $1 - 2^{-\Omega(n)}$ when $\sigma = \sigma_j$, by \Cref{lemma:solveprob1/2}.
\end{proof}

\Cref{lemma:solveprob1/2} follows from the fact that $\gamma_{t_j}^{\ary{a}}$ is close to $\gamma_{\sqrt{2}\alpha q}^{\ary{a}}$ as $t_j$ is close to $\sqrt{2}\alpha q$. For completeness, let's provide the formal proof here.

\begin{proof}[Proof of~\Cref{lemma:solveprob1/2}]
    We label the $m$ samples to be $\{(\ary a_i, \gamma_{t_j}^{\ary a_i})\}_{i \in [m]}$. By assumption, the quantum algorithm $\mathcal{A}$ can solve $\ary s$ with probability at least $1 - 2^{-\Omega(n)}$, when given $m$ samples $\{(\ary a_i, \gamma_{\sqrt{2} \alpha q}^{\ary a_i})\}_{i \in [m]}$. The trace distance between the states in these two different types of samples is given by
    \iffull
    \[
        \begin{split}
            \td\left(\bigotimes_{i = 1}^m\gamma_{t_j}^{\mat a_i}, \bigotimes_{i = 1}^m\gamma_{\sqrt{2}\alpha q}^{\mat a_i}\right) & \le \sum_{i = 1}^m \td\left(\gamma_{t_j}^{\ary a_i}, \gamma_{\sqrt{2}\alpha q}^{\ary a_i}\right) \\
            & \le \sum_{i = 1}^m \td\left(\ket{D_{q\Lattice + \Lattice \ary a_i, r}} \sum_{e \in \Z_{qR}/R}\rho_{\sigma_j}(e)\ket{e}, \ket{D_{q\Lattice + \Lattice \ary a_i, r}} \sum_{e \in \Z_{qR}/R}\rho_{\sigma'}(e)\ket{e}\right) + 2^{-\Omega(n)} \\
            & = \sum_{i = 1}^m \td\left(\sum_{e \in \Z_{qR}/R}\rho_{\sigma_j}(e)\ket{e}, \sum_{e \in \Z_{qR}/R}\rho_{\sigma'}(e)\ket{e}\right) + 2^{-\Omega(n)} \\
            & \le_{(*)} m\sqrt{\frac{(\sigma' - \sigma_j)^2}{\sigma_j^2 + \sigma'^2}}(1 + 2^{-\Omega(n)}) + 2^{-\Omega(n)} \\
            & \le m\cdot \sqrt{\frac{(\alpha q / (2m))^2}{2(\alpha q)^2}}(1 + 2^{-\Omega(n)}) + 2^{-\Omega(n)} \\
            & \le \frac{1}{2\sqrt{2}} + 2^{-\Omega(n)},
        \end{split}
    \]
    \fi
    \ifllncs
    \[
        \begin{split}
            &~\td\left(\bigotimes_{i = 1}^m\gamma_{t_j}^{\mat a_i}, \bigotimes_{i = 1}^m\gamma_{\sqrt{2}\alpha q}^{\mat a_i}\right) \\
            \le&~ \sum_{i = 1}^m \td\left(\ket{D_{q\Lattice + \Lattice \ary a_i, r}} \sum_{e \in \Z_{qR}/R}\rho_{\sigma_j}(e)\ket{e}, \ket{D_{q\Lattice + \Lattice \ary a_i, r}} \sum_{e \in \Z_{qR}/R}\rho_{\sigma'}(e)\ket{e}\right) + 2^{-\Omega(n)} \\
            \le&_{(*)}~ m\sqrt{\frac{(\sigma' - \sigma_j)^2}{\sigma_j^2 + \sigma'^2}}(1 + 2^{-\Omega(n)}) + 2^{-\Omega(n)} \\
            \le&~ \frac{1}{2\sqrt{2}} + 2^{-\Omega(n)},
        \end{split}
    \]
    \fi
where $(*)$ is according to \Cref{lemma:dist_between_gaussian_state} and $\sigma_j, \sigma' \in [\alpha q, \sqrt 2 \alpha q), R > \frac{2\sqrt{n}}{\alpha q}$.

Therefore, the quantum algorithm $\mathcal{A}$ will output $\ary s$ when given $\{(\ary a_i, \gamma_{t_j}^{\ary a_i})\}_{i \in [m]}$ with probability at least
\[1 - \frac{1}{2\sqrt{2}} - 2^{-\Omega(n)} > \frac12. \qedhere\]
\end{proof}

One last gap between \Cref{lemma:SLWE2CVPmodq} and the theorem we need to prove (\Cref{thm:DGStoCVP}) is that, in \Cref{lemma:SLWE2CVPmodq} the algorithm only outputs $\ary{s} = (\Lattice^*)^{-1}\kappa_{\Lattice^*}(\ary{x})\bmod q$, which answers a modulo version of $\CVP_{\Lattice^*, \alpha q/r}$ for instance $\ary{x} \in \Lattice^*/R$. However, as \cite{DBLP:journals/jacm/Regev09} shows, $\CVP$ is efficiently reducible to its modulo version, which closes this final gap. Formally, 

\begin{theorem}[{\cite[a slight modification of Lemma 3.5]{DBLP:journals/jacm/Regev09}}]\label{lemma:cvp_mod_q_suffices}
Let $\Lattice$ be an $n$-dimensional integer lattice, define distance parameter $d \in (0, \lambda_1(\Lattice)/2)$, and integers $q > 2, R > \frac{2^{2n + 1}n\lambda_n(\Lattice)}{d}$. Assume there exists an algorithm $\mathcal{A}$ that, on input $\ary{x} \in \Lattice/R$ with the guarantee $\dist(\ary{x}, \Lattice) \le d$, outputs $\ary{s} = \Lattice^{-1}\kappa_{\Lattice}(\ary{x})\bmod q$ with probability $1 - 2^{-\Omega(n)}$, then there exists a polynomial-time algorithm that, on input $\ary{x} \in \Lattice/R$ with guarantee $\dist(\ary{x}, \Lattice) \le d$, outputs $\kappa_{\Lattice}(\ary{x})$ with probability $1 - 2^{-\Omega(n)}$ using at most $n$ calls to $\mathcal{A}$.
\end{theorem}

\iffull
\begin{proof}
It is merely the same as the proof of Lemma 3.5 in~\cite{DBLP:journals/jacm/Regev09}, but with a slight modification. Compute a sequence $\ary x_1, \ary x_2, \cdots$ where $\ary x_1$ is the input $\ary{x}$ and $\ary x_{i + 1}$ is given by the following: for $\ary{x}_i\in \Lattice/R$, call $\mathcal{A}$ to compute $\ary{s}_i = \Lattice^{-1}\kappa_{\Lattice}(\ary{x}_i)\bmod q$, and then compute $\ary{y}_i:=(\ary{x}_i - \Lattice\ary{s}_i)/q$. Now that $\ary{y}_i\in \Lattice/(Rq)$ and $\dist(\ary{y}_i, \Lattice)=\dist(\ary{x}_i,\Lattice)/q$, we can apply Babai’s nearest plane algorithm~\cite{DBLP:journals/combinatorica/Babai86} to find a point $\ary{x}_{i+1} \in \Lattice/R$ such that $\dist(\ary{x}_{i+1}, \ary{y}_i) \le 2^n \dist(\ary{y}_{i}, \Lattice/R) \le 2^{n - 1} n\lambda_n(\Lattice/R) < d/2^{n + 2}$. Then $\dist(\ary{x}_{i + 1}, \Lattice) \le \dist(\ary{y}_i, \Lattice) + \dist(\ary{x}_{i+1}, \ary{y}_i) < \dist(\ary{x}_i, \Lattice)/q + d/2^{n+2}$. 

Thus $\dist(\ary{x}_i, \Lattice) \le \frac{1}{q^{i - 1}}(\dist(\ary{x}_1, \Lattice) - dq/(2^{n + 2}(q - 1))) + dq/(2^{n + 2}(q - 1)) \le \frac{d}{q^{i - 1}} + \frac{d}{2^{n + 1}}$. Then $\kappa_{\Lattice}(\ary{x}_{i+1}) = \kappa_{\Lattice}(\ary{y}_i)$ because $\dist(\ary{y}_i, \Lattice) + \dist(\ary{x}_{i+1}, \ary{y}_i) < \dist(\ary{x}_i, \Lattice)/q + d/2^{n + 2} < d < \lambda_1(\Lattice)/2$. 

After $n$ steps, we have a point $\ary{x}_{n + 1}$ such that $\dist(\ary{x}_{n + 1}, \Lattice) \le \frac{d}{q^{n}} + \frac{d}{2^{n + 1}} < \frac{d}{2^n}$. Hence we can apply Babai's nearest plane algorithm~\cite{DBLP:journals/combinatorica/Babai86} to recover $\kappa_{\Lattice}(\ary{x}_{n + 1})$. 

Note that $\kappa_{\Lattice}(\ary{x}_{i}) = q\kappa_{\Lattice}(\ary{y}_{i}) + \Lattice \ary{s}_i = q\kappa_{\Lattice}(\ary{x}_{i + 1}) + \Lattice \ary{s}_i$. We can recover $\kappa_{\Lattice}(\ary{x}_{1})$ step by step. Notice that each call to $\mathcal{A}$ has failure probability $2^{-\Omega(n)}$ and we use $n$ calls to $\mathcal{A}$, so our algorithm has failure probability at most $2^{-\Omega(n)}$, which completes the proof.
\end{proof}
\fi
\ifllncs
The detailed proof is deferred to \Cref{sec:Omit_them14}.
\fi
As a direct corollary of \Cref{lemma:SLWE2CVPmodq} and \Cref{lemma:cvp_mod_q_suffices}, we can reduce $\CVP$ to $\QLWE^{\sf phase}$ given a sufficient number of states $\ket{D_{\Lattice, r}}$. Formally,

\begin{corollary}\label{coro:SLWE2CVP}
Let $\Lattice$ be an $n$-dimensional integer lattice, define the parameters $\epsilon \in (0, 2^{-n})$, $\alpha \in (0, \frac{1}{5\sqrt{n}})$, $r > 4 q\eta_{\epsilon}(\Lattice)$, and a precision parameter $R > \max\{2\sqrt{n}r\sqrt{\log r}, \frac{2\sqrt{n}}{\alpha q}, \frac{2^{2n + 1}nr\lambda_n(\Lattice^*)}{\alpha q}\}$ as an integer. Assume that there exists a quantum algorithm $\mathcal{A}$ that, given $m$ samples of uniformly random vector $\ary a\in \Z_q^n$ and state $\gamma^{\ary{a}}_{\sqrt{2}\alpha q}$, solves the secret vector $\ary s$ in time complexity $T$. Then there exists an algorithm that given a $\CVP_{\Lattice^*, \alpha q / r}$ instance $\ary{x} \in \Lattice^*/R$ and $3m^2n^2$ states $\ket{D_{\Lattice, r}}$ as input, outputs $\kappa_{\Lattice^*}(\ary{x})$ with probability $1 - 2^{-\Omega(n)}$ in time $O((m^2+mT)\poly(n))$.
\end{corollary}

\subsection{Answering the quantum $\CVP$ query with small disturbance on $\ket{D_{\Lattice, r}}$}\label{sec:answerquery}

We are now ready to conclude the proof of \Cref{thm:DGStoCVP}. \Cref{coro:SLWE2CVP} provides a method to answer classical queries for $\CVP_{\Lattice^*, \alpha q / r}$ on instance $\ary{x} \in \Lattice^*/R$ with $1 - 2^{-\Omega(n)}$ probability. By deferred measurement principle and gentle measurement principle, we expect that it can help us to answer quantum query $\ket{\ary x, \ary y} \to \ket{\ary x, \ary y + \kappa_{\Lattice^*}(\ary x)}$ with $\dist(\ary{x}, \Lattice^*) \le \alpha q/r$, using $3m^2n^2$ states $\ket{D_{\Lattice, r}}$ while introducing only an exponentially small disturbance on them. However, it is important to note that the gentle measurement principle is primarily suitable for producing measurement results (which is classical) rather than answering quantum query. Therefore, we shall provide a formal proof for using gentle measurement principle to answer quantum query:

\begin{proof}[Proof of \Cref{thm:DGStoCVP}]

The existence of the quantum algorithm $\mathcal{A}$ in \Cref{coro:SLWE2CVP} is a direct consequence of the assumption that a quantum algorithm can solve the $\QLWE^{\sf phase}_{n,m,q,f, \theta^{(r,\ary{x})}, D_{\theta}^{(r,\ary{x})}}$ problem. Specifically, let's recall the expression
\[
    \gamma_{\sqrt{2}\alpha q}^{\ary a} = \sum_{\substack{\ary{y} \in \Z_R^n \cap R\cdot B_{(q\Lattice)^*}}}\rho_{\sqrt{\mat\Sigma / 2}}(\ary{z}(\ary{y}))\kb{\ary{y}}{\ary{y}} \otimes \ket{\psi^{\ary{a}, \ary{y}}_{\sqrt{2}\alpha q}}\bra{\psi^{\ary{a}, \ary{y}}_{\sqrt{2}\alpha q}}.
\]
The quantum algorithm $\mathcal{A}$ proceeds by first measuring the result of $\ket{\ary y}$ to obtain a specific $\ary{y}$ and an $\QLWE^{\sf phase}$ state $\ket{\psi^{\ary{a}, \ary{y}}_{\sqrt{2}\alpha q}}$. The parameters and functions of this state corresponds to the parameters and functions defined in \Cref{def:SLWEparamsRegev}. Subsequently, $\mathcal{A}$ applies the $\QLWE^{\sf phase}_{n,m,q,f, \theta^{(r,\ary{x})}, D_{\theta}^{(r,\ary{x})}}$ solver to compute $\ary{s}$.

We defer all the measurements in the algorithm in~\Cref{coro:SLWE2CVP} including those in the $\QLWE^{\sf phase}$ solver $\mathcal{A}$ to obtain a unitary $U:\ket{\ary{x}}\ket{D_{\Lattice,r}}^{\otimes \left(3m^2n^2\right)}\ket{0^{\sf Aux}} \rightarrow \ket{\ary x} \ket{\phi_\ary{x}}$, where the first register of $\ket{\phi_\ary{x}}$ contains the solution $\kappa_{\Lattice^*}(\ary{x})$, and the size of $U$ is $O((m^2 + mT)\poly(n))$. As the algorithm in \Cref{coro:SLWE2CVP} outputs  $\kappa_{\Lattice^*}(\ary{x})$ with probability at least $1 - 2^{-\Omega(n)}$ whenever $\dist(\ary{x}, \Lattice^*) \le \alpha q/r$ and $\ary{x} \in \Lattice^*/R$, by gentle measurement principle~\cite{Winter_1999}, there exists state $\ket{\phi_{\ary x}^{\mathsf{Aux}}}$ such that the state $\ket{\phi_{\ary x}}$ is $2^{-\Omega(n)}$-close to the state $\ket{\kappa_{\Lattice^*}(\ary x)}\ket{\phi_{\ary x}^{\mathsf{Aux}}}$ in $\ell_2$-norm. Our quantum algorithm answers the query $\ket{\ary x, \ary y} \to \ket{\ary x, \ary y + \kappa_{\Lattice^*}(\ary x)}$ with $\dist(\ary{x}, \Lattice^*) \le \alpha q/r$ using the following procedure:
\begin{enumerate}
    \item Prepare the initial state \[\ket{\ary{x}, \ary{y}} \ket{D_{\Lattice,r}}^{\otimes \left(3m^2n^2\right)}\ket{0^{\sf Aux}}.\]
    \item Apply $U$ to the registers containing $\ket{\ary{x}}\ket{D_{\Lattice, r}}^{\otimes (3m^2n^2)}\ket{0^{\sf Aux}}$ to get a state $2^{-\Omega(n)}$-close to the state \[\ket{\ary x, \ary y}\ket{\kappa_{\Lattice^*}(\ary x)}\ket{\phi_{\ary x}^{\mathsf{Aux}}}\]
    in $\ell_2$-norm.
    \item Apply a unitary to add the value of $\kappa_{\Lattice^*}(\ary x)$ to $\ary y$ to get a state $2^{-\Omega(n)}$-close to the state \[\ket{\ary x, \ary y + \kappa_{\Lattice^*}(\ary x)}\ket{\kappa_{\Lattice^*}(\ary x)}\ket{\phi_{\ary x}^{\mathsf{Aux}}}\]
    in $\ell_2$-norm.
    \item Apply $U^\dagger$ to the registers containing $\ket{\ary x}\ket{\kappa_{\Lattice^*}(\ary x)}\ket{\phi_{\ary x}^{\mathsf{Aux}}}$ to get a state $2^{-\Omega(n)}$-close to the state \[\ket{\ary{x}, \ary{y} + \kappa_{\Lattice^*}(\ary x)} \ket{D_{\Lattice,r}}^{\otimes \left(3m^2n^2\right)}\ket{0^{\sf Aux}}\]
    in $\ell_2$-norm.
\end{enumerate}
Therefore, we can answer the quantum query up to exponentially small error with exponentially small disturbance on the states $\ket{D_{\Lattice, r}}$ in time $O((m^2 + mT)\poly(n))$, which ends up the proof of \Cref{thm:DGStoCVP}.
\end{proof}
\else
\fi

\iffull
\small
\bibliography{q_ref}
\normalsize
\else
\bibliography{q_ref}
\fi

\appendix

\section{Appendix}

\subsection{Upper bounds on Gaussian tails}

\begin{lemma}\label{thm:additive_err_tail}
    Let $\Lattice\subseteq \R^n$ be a lattice, $\ary u\in \R^n$ be a fixed vector, $\epsilon \in (0, 1)$ be a small error parameter, $\sigma$ be a positive real number with $\sigma > 2\eta_{\epsilon}(\Lattice)$. Then we have
    \[
        \sum_{\ary x\in \Lattice^*}\rho_{1/\sigma}(\ary x - \ary u) < \rho_{1/\sigma}(\ary u - \kappa_{\Lattice^*}(\ary u)) + \epsilon.
    \]
    If $\ary{u}$'s closest vectors in $\Lattice^*$ are not unique, then $\kappa_{\Lattice^*}(\ary{u})$ can be an arbitrary one of the closest vectors. 
\end{lemma}

\begin{proof}
    It suffices to prove that
    \[
        \sum_{\ary x\in {\Lattice^*}\setminus \{\kappa_{\Lattice^*}(\ary u)\}}\rho_{1/\sigma}(\ary x - \ary u) < \epsilon.
    \]

    For any $\ary x\in \Lattice^*$, we have that
    \[
        \begin{split}
            \|\ary x - \ary u\|^2 & \ge \frac{1}{2}\left(\|\ary x - \ary u\|^2 + \|\kappa_{\Lattice^*}(\ary u) - \ary u\|^2\right)\qquad\text{(by the definition of }\kappa_{\Lattice^*}(\ary u)\text{)} \\
            & \ge \frac{1}{4}\left(\|\ary x - \ary u\| + \|\kappa_{\Lattice^*}(\ary u) - \ary u\|\right)^2 \\
            & \ge \frac{1}{4}\|\ary x - \ary u - (\kappa_{\Lattice^*}(\ary u) - \ary u)\|^2 \qquad \text{(triangle inequality)} \\
            & = \frac{1}{4}\|\ary x - \kappa_{\Lattice^*}(\ary u)\|^2,
        \end{split}
    \]
    so
    \[
        \begin{split}
            \sum_{\ary x\in {\Lattice^*}\setminus \{\kappa_{\Lattice^*}(\ary u)\}}\rho_{1/\sigma}(\ary x - \ary u) & \le \sum_{\ary x\in \Lattice^*\setminus \{\kappa_{\Lattice^*}(\ary u)\}}\rho_{2/\sigma}(\ary x - \kappa_{\Lattice^*}(\ary u)) \\
            & = \sum_{\ary x\in \Lattice^*\setminus \{\ary 0\}}\rho_{2/\sigma}(\ary x) \\
            & < \epsilon,
        \end{split}
    \]
    as desired.
\end{proof}

\begin{lemma}\label{thm:mul_err_tail}
    Let $\Lattice\subseteq \R^n$ be a lattice, $\ary u\in \R^n$ be a fixed vector, $\epsilon \in (0, 1)$ be a small error parameter, $\sigma$ be a positive real number with $\sigma > 2\sqrt{2}\eta_{\epsilon}(\Lattice)$. Then we have
    \[
        \sum_{\ary x\in \Lattice^*}\rho_{1/\sigma}(\ary x - \ary u) < \rho_{1/\sigma}(\kappa_{\Lattice^*}(\ary u) - \ary{u}) + \epsilon \cdot \rho_{\sqrt{2}/\sigma}(\kappa_{\Lattice^*}(\ary u) - \ary{u}).
    \]
    If $\ary{u}$'s closest vectors in $\Lattice^*$ are not unique, then $\kappa_{\Lattice^*}(\ary{u})$ can be an arbitrary one of the closest vectors. 
\end{lemma}

\begin{proof}
    It suffices to prove that
    \[
        \sum_{\ary x\in {\Lattice^*}\setminus\{ \kappa_{\Lattice^*}(\ary u)\}}\rho_{1/\sigma}(\ary x - \ary u) < \epsilon \cdot \rho_{\sqrt{2}/\sigma}(\kappa_{\Lattice^*}(\ary u) - \ary{u}).
    \]
    
    For any $\ary x\in \Lattice^*$, we have that
    \iffull 
    \[
        \begin{split}
            \|\ary x - \ary u\|^2 -   \frac{1}{2}\|\kappa_{\Lattice^*}(\ary u) - \ary u\|^2& \ge \frac{1}{4}\left(\|\ary x - \ary u\|^2 + \|\kappa_{\Lattice^*}(\ary u) - \ary u\|^2\right)\qquad\text{(by the definition of }\kappa_{\Lattice^*}(\ary u)\text{)} \\
            & \ge \frac{1}{8}\left(\|\ary x - \ary u\| + \|\kappa_{\Lattice^*}(\ary u) - \ary u\|\right)^2 \\
            & \ge \frac{1}{8}\|\ary x - \ary u - (\kappa_{\Lattice^*}(\ary u) - \ary u)\|^2 \qquad \text{(triangle inequality)} \\
            & = \frac{1}{8}\|\ary x - \kappa_{\Lattice^*}(\ary u)\|^2,
        \end{split}
    \]
    \fi
    so
    \[
        \begin{split}
            \sum_{\ary x\in {\Lattice^*}\setminus \{\kappa_{\Lattice^*}(\ary u)\}}\rho_{1/\sigma}(\ary x - \ary u) & \le \rho_{\sqrt{2}/\sigma}(\kappa_{\Lattice^*}(\ary u) - \ary{u})\sum_{\ary x\in \Lattice^*\setminus \{\kappa_{\Lattice^*}(\ary u)\}}\rho_{2\sqrt{2}/\sigma}(\ary x - \kappa_{\Lattice^*}(\ary u)) \\
            & = \rho_{\sqrt{2}/\sigma}(\kappa_{\Lattice^*}(\ary u) - \ary{u})\sum_{\ary x\in \Lattice^*\setminus \{\ary 0\}}\rho_{2\sqrt{2}/\sigma}(\ary x) \\
            & < \epsilon \cdot \rho_{\sqrt{2}/\sigma}(\kappa_{\Lattice^*}(\ary u) - \ary{u}),
        \end{split}
    \]
    as desired.
\end{proof}

\subsection{Proof of \Cref{lem:EDCP_state_close}}\label{sec:EDCP_state_close}

\begin{proof}[Proof of \Cref{lem:EDCP_state_close}]
    Denote the (unnormalized) state in~\Cref{eqn:EDCP_state_temp0} as $\ket{\Phi}$, and the (unnormailized) state in~\Cref{eqn:EDCP_state_temp1} as $\ket{\Phi'}$. Then we have
    \iffull
    \[
        \begin{split}
            &\|\ket{\Phi} - \ket{\Phi'}\|^2\\
            =& \sum_{\ary{v}\in \Z_q^n}\sum_{j\in \Z_q}\rho_\alpha(j)^2\left\|\sum_{\substack{\ary{x}\in \Z^m,\\ \|\ary{x}\| \ge \lambda_1(\Lattice_q(\mat A))/2}}\rho_{\beta q}(\ary{x} + j\cdot \ary e)\ket{ (\mat{A}^T\ary{v} + \ary{x}) \bmod q}\right\|^2 \\
            \le&_{(1)} q^n\sum_{j\in \Z_q}\rho_\alpha(j)^2\left(\sum_{\substack{\ary{x}\in \Z^m,\\ \|\ary{x}\| \ge \lambda_1(\Lattice_q(\mat A))/2}}\rho_{\beta q}(\ary{x} + j\cdot \ary e)\right)^2 \\
            \le&_{(2)} q^n\sum_{\substack{j\in \Z_q, \\ |j| < \alpha \sqrt{m\log {(\beta q)}}}}\rho_\alpha(j)^2\left(\sum_{\substack{\ary{x}\in \Z^m,\\ \|\ary{x}\| \ge {\beta q}\sqrt{m\log {(\beta q)}}}}\rho_{\beta q}(\ary{x})\right)^2 + q^n\sum_{\substack{j \in \Z_q, \\|j| \ge \alpha \sqrt{m\log {(\beta q)}}}}\rho_\alpha(j)^2\left(\sum_{\ary{x}\in \Z^m}\rho_{\beta q}(\ary{x})\right)^2,
        \end{split}
    \]
    \fi
    where in (1) we absorb the summation over $\ary{v}\in\Z_q^n$ in $q^n$, and use the fact that $\rho$ is non-negative; in (2), the first term uses the following: when $\|\ary x\| \ge \lambda_1(\Lattice_q(\mat A)) / 2$ and $|j| < \alpha \sqrt{m\log (\beta q)}$, we have $\|\ary x + j\cdot \ary e\| > \lambda_1(\Lattice_q(\mat A)) / 2 - \alpha\sqrt{m\log (\beta q)}\cdot \gamma q \sqrt{m} > \beta q\sqrt{m\log (\beta q)}$, the second term uses for any $\ary{c}\in\R^n$, $\sum_{\ary{x}\in \Z^m}\rho_{\beta q}(\ary{x}+\ary{c})\leq \sum_{\ary{x}\in \Z^m}\rho_{\beta q}(\ary{x})$.
    
    From Banaszczyk's tail bound (see \Cref{lemma:Bana93}), we have that
    \[
        \sum_{\substack{\ary{x}\in \Z^m,\\ \|\ary{x}\| \ge \beta q\sqrt{m\log (\beta q)}}}\rho_{\beta q}(\ary{x}) < \beta^{-3m} q^{-3m} \sum_{\ary{x}\in \Z^m}\rho_{\beta q}(\ary x), \quad \sum_{\substack{j \in \Z_q, \\|j| \ge \alpha \sqrt{m\log (\beta q)}}}\rho_{\alpha/\sqrt{2}}(j) < \beta^{-6m} q^{-6m}\sum_{j \in \Z_q}\rho_{\alpha/\sqrt{2}}(j).
    \]
    
    Notice that $\beta q > \sqrt{m}$, so from~\Cref{lemma:gaussianintsum}, $\sum_{\ary{x}\in \Z^m}\rho_{\beta q}(\ary{x}) < (1 + 2^{-\Omega(m)}) \beta^{m} q^m$. Therefore,
    \[
        \begin{split}
            \|\ket{\Phi} - \ket{\Phi'}\|^2 & <2 q^n (\beta q)^{-6m} \sum_{j\in \Z_q}\rho_\alpha(j)^2 \left(\sum_{\ary{x}\in \Z^m}\rho_{\beta q}(\ary{x})\right)^2 \\
            & < 2(1 + 2^{-\Omega(n)})q^n (\beta q)^{-4m} \sum_{j\in \Z_q}\rho_\alpha(j)^2
        \end{split}
    \]

    On the other hand, notice that from~\Cref{lemma:gaussianintsum}, $\sum_{\ary{x}\in \Z^m}\rho_{\beta q/\sqrt{2}}(\ary{x}) > (\beta q/\sqrt{2})^{m}$, we have that
    \[
        \begin{split}
            \|\ket{\Phi}\|^2 & \ge  \sum_{\ary{v}\in\Z_q^n}\sum_{j\in\Z_q} \rho_\alpha(j)^2 \sum_{\ary{x}\in\Z^m} \rho_{\beta q}(\ary{x} + j\cdot \ary e)^2\\
            & = q^n\sum_{j\in\Z_q} \rho_\alpha(j)^2 \sum_{\ary{x}\in\Z^m} \rho_{\beta q}(\ary{x})^2 \\
            & > q^n\left(\frac{\beta q}{\sqrt{2}}\right)^{m}\sum_{j\in \Z_q}\rho_\alpha(j)^2
        \end{split}
    \]

    Hence, we get that $\frac{\|\ket{\Phi} - \ket{\Phi'}\|^2}{\|\ket{\Phi}\|^2} < 2^{-\Omega(n)}$, this implies that $\ket{\Phi}$ and $\ket{\Phi'}$ are $2^{-\Omega(n)}$-close to each other, by \Cref{lemma:difftotd}, as desired.
\end{proof}

\subsection{Proof of \Cref{lem:change_variable}}\label{sec:proof_change_variable}
\begin{proof}[Proof of \Cref{lem:change_variable}]
For any vector $\ary{v} \in q\Lattice + \Lattice \ary{a}$ such that $\|\ary{v}\| \le \sqrt{n}r$, as $q/2 > \sigma\sqrt{n}$, the state
\begin{equation}\label{eqn:righthand_1}
    \sum_{e\in\Z_{qR}/R} \rho_{\sigma}(e)\ket{\ipd{\ary{s}}{\ary{a}} + \ipd{\ary{x'}}{\ary{v}}+e \bmod q}
\end{equation}
is $2^{-\Omega(n)}$-close to the state
\[\sum_{e\in[-\sigma\sqrt{n}, \sigma\sqrt{n}] \cap (\Z/R)} \rho_{\sigma}(e)\ket{\ipd{\ary{s}}{\ary{a}} + \ipd{\ary{x'}}{\ary{v}}+e \bmod q}\]
in $\ell_2$ norm by Banaszczyk’s tail bound (see~\Cref{lemma:Bana93}), which is $2^{-\Omega(n)}$-close to the state
\[\sum_{e\in\Z_{qR}/R - \ipd{\ary{x'}}{\ary{v}}} \rho_{\sigma}(e)\ket{\ipd{\ary{s}}{\ary{a}} + \ipd{\ary{x'}}{\ary{v}}+e \bmod q}\]
in $\ell_2$ norm due to Banaszczyk’s tail bound and the fact that $\left|\ipd{\ary{x'}}{\ary{v}}\right| \le \frac{\alpha q}{r}\cdot \sqrt{n}r < q/2 - \sigma\sqrt{n}$, which implies $[-\sigma\sqrt{n}, \sigma\sqrt{n}] \cap (\Z/R) \subseteq \Z_{qR}/R - \ipd{\ary{x'}}{\ary{v}}$. 

We can perform a change of variable $u \leftarrow \ipd{\ary{x'}}{\ary{v}}+e$ to write the last state as
\begin{equation}\label{eqn:righthand_2}
    \sum_{u\in \Z_{qR}/R} \rho_{\sigma}(u - \ipd{\ary{x'}}{\ary{v}})\ket{\ipd{\ary{s}}{\ary{a}} + u \bmod q}.
\end{equation}
Observe that \Cref{eqn:righthand_1} and \Cref{eqn:righthand_2} are $2^{-\Omega(n)}$-close to each other in $\ell_2$ norm for every $\ary{v}\in q\Lattice + \Lattice \ary{a}$ such that $\|\ary{v}\| \le \sqrt{n}r$. Therefore, the given state
\[\sum_{\substack{\ary{v}\in q\Lattice + \Lattice\ary{a}, \\ \|\ary{v}\| \le \sqrt{n} r}} \rho_{r}(\ary{v})\ket{\ary{v}} \otimes \sum_{e\in\Z_{qR}/R} \rho_{\sigma}(e)\ket{\ipd{\ary{s}}{\ary{a}} + \ipd{\ary{x'}}{\ary{v}}+e \bmod q}\]
is $2^{-\Omega(n)}$-close to the state
\[\ket{\Psi}:= \sum_{\substack{\ary{v}\in q\Lattice + \Lattice\ary{a}, \\ \|\ary{v}\| \le \sqrt{n} r}} \rho_{r}(\ary{v})\ket{\ary{v}} \otimes \sum_{u\in\Z_{qR}/R} \rho_{\sigma}(u - \ipd{\ary{x'}}{\ary{v}})\ket{\ipd{\ary{s}}{\ary{a}} + u \bmod q}.\]

We will show $\ket{\Psi}$ is $2^{-\Omega(n)}$-close to the state
\[\ket{\Phi}:=\sum_{\substack{\ary{v}\in q\Lattice + \Lattice\ary{a}, \\ \|\ary{v}\| < R/2}} \rho_{r}(\ary{v})\ket{\ary{v}} \otimes \sum_{u\in\Z_{qR}/R} \rho_{\sigma}(u - \ipd{\ary{x'}}{\ary{v}})\ket{\ipd{\ary{s}}{\ary{a}} + u \bmod q}.\]

The proof is similar to the proof of~\Cref{lem:EDCP_state_close}. We first give an upper bound for $\|\ket{\Psi} - \ket{\Phi}\|^2$:
\begin{align*}
\|\ket{\Psi} - \ket{\Phi}\|^2 &= \sum_{\substack{\ary{v}\in q\Lattice + \Lattice\ary{a}, \\ \sqrt{n} r < \|\ary{v}\| < R/2}}\rho_{r/\sqrt{2}}(\ary{v})\sum_{u\in\Z_{qR}/R}\rho_{\sigma/\sqrt{2}}(u - \ipd{\ary{x'}}{\ary{v}})\\
&\le \sum_{\substack{\ary{v}\in q\Lattice + \Lattice\ary{a}, \\\|\ary{v}\| > \sqrt{n} r}}\rho_{r/\sqrt{2}}(\ary{v})\sum_{u\in\Z/R}\rho_{\sigma/\sqrt{2}}(u)\\
&\le 2^{-\Omega(n)}\sum_{\substack{\ary{v}\in q\Lattice}}\rho_{r/\sqrt{2}}(\ary{v})\sum_{u\in\Z/R}\rho_{\sigma/\sqrt{2}}(u)
\end{align*}
by Banaszczyk’s tail bound.

On the other hand, notice that
\begin{align*}
    \|\ket{\Psi}\|^2 &= \sum_{\substack{\ary{v}\in q\Lattice + \Lattice\ary{a}, \\ \|\ary{v}\| \le \sqrt{n} r}}\rho_{r/\sqrt{2}}(\ary{v})\sum_{u\in\Z_{qR}/R}\rho_{\sigma/\sqrt{2}}(u - \ipd{\ary{x'}}{\ary{v}})\\
    &\ge \sum_{\substack{\ary{v}\in q\Lattice + \Lattice\ary{a}, \\ \|\ary{v}\| \le \sqrt{n} r}}\rho_{r/\sqrt{2}}(\ary{v})\sum_{ u \in [-\sigma\sqrt{n}, \sigma\sqrt{n}] \cap \Z/R }\rho_{\sigma/\sqrt{2}}(u)\\
    &\ge (1 - 2^{-\Omega(n)}) \sum_{\substack{\ary{v}\in q\Lattice}}\rho_{r/\sqrt{2}}(\ary{v})\sum_{ u \in \Z/R }\rho_{\sigma/\sqrt{2}}(u)
\end{align*}
by Banaszczyk’s tail bound and \Cref{lemma:regevclaim3.8} together with the fact that $r/\sqrt{2} > q\eta_{\epsilon}(\Lattice)$ for $\epsilon < 2^{-n}$.

So $\frac{\|\ket{\Phi} - \ket{\Psi}\|^2}{\|\ket{\Phi}\|^2} < 2^{-\Omega(n)}$, which implies that $\ket{\Phi}$ and $\ket{\Psi}$ are $2^{-\Omega(n)}$-close to each other by~\Cref{lemma:difftotd}. 

It remains to prove that $\ket{\Phi}$ is $2^{-\Omega(n)}$-close to \[\ket{\Phi'}:=\sum_{\ary{v}\in q\Lattice + \Lattice\ary{a}} \rho_{r}(\ary{v})\ket{\ary{v} \bmod R} \otimes \sum_{u\in\Z_{qR}/R} \rho_{\sigma}(u - \ipd{\ary{x'}}{\ary{v}})\ket{\ipd{\ary{s}}{\ary{a}} + u \bmod q}.\]

We again use a similar argument as the proof of~\Cref{lem:EDCP_state_close}. 

We first give an upper bound for $\|\ket{\Phi} - \ket{\Phi'}\|^2$:
\iffull
\begin{align*}
    \|\ket{\Phi} - \ket{\Phi'}\|^2
    & \le_{(1)} \left(\sum_{\substack{\ary{v}\in q\Lattice + \Lattice\ary{a}, \\ \|\ary{v}\| \ge R/2}} \rho_r(\ary{v})\left\|\ket{\ary{v} \bmod R}\otimes \sum_{u\in\Z_{qR}/R} \rho_{\sigma}(u - \ipd{\ary{x'}}{\ary{v}})\ket{\ipd{\ary{s}}{\ary{a}} + u \bmod q}\right\|\right)^2\\
    & \le_{(2)} \left(\sum_{\substack{\ary{v}\in q\Lattice + \Lattice\ary{a}, \\ \|\ary{v}\| > \sqrt{n}r\sqrt{\log r}}}\rho_r(\ary{v})\sqrt{\sum_{u\in\Z_{qR}/R} \rho_{\sigma}(u - \ipd{\ary{x'}}{\ary{v}})^2}\right)^2\\
    & \le_{(3)} r^{-6n}\left(\sum_{u \in \Z/R}\rho_\sigma(u)^2\right)\left(\sum_{\ary{v}\in q\Lattice}\rho_r(\ary{v})\right)^2\\
    & \le_{(4)} r^{-4n}\left(\sum_{u \in \Z/R}\rho_\sigma(u)^2\right)\det((q\Lattice)^*)^2 (1 + 2^{-\Omega(n)})
\end{align*}
\fi
where $(1)$ is due to triangle inequality, $(2)$ uses that $R/2 > \sqrt{n} r\sqrt{\log r}$, $(3)$ is due to Banaszczyk’s tail bound, and $(4)$ is due to \Cref{lemma:regevclaim3.8} and $r\ge \eta_{\epsilon}(q\Lattice)$ for $\epsilon < 2^{-n}$.

On the other hand, notice that
\begin{align*}
\|\ket{\Phi}\|^2 &\ge \sum_{\substack{\ary{v}\in q\Lattice + \Lattice\ary{a}, \\ \|\ary{v}\| \le \sqrt{n}r}} \rho_{r}(\ary{v})^2 \sum_{u\in\Z_{qR}/R} \rho_{\sigma}(u - \ipd{\ary{x'}}{\ary{v}})^2\\
&\ge_{(1)} (1 - 2^{-\Omega(n)})\sum_{\substack{\ary{v}\in q\Lattice + \Lattice\ary{a}, \\ \|\ary{v}\| \le \sqrt{n}r}} \rho_{r}(\ary{v})^2 \sum_{u\in\Z/R} \rho_{\sigma}(u)^2\\
&\ge_{(2)} (1 - 2^{-\Omega(n)})(r/\sqrt{2})^n\det((q\Lattice)^*)\sum_{u\in\Z/R} \rho_{\sigma}(u)^2
\end{align*}
where $(1)$ is due to Banaszczyk’s tail bound and the fact that $\sigma\sqrt{n} < q/2 - \left|\ipd{\ary{x'}}{\ary{v}}\right|$ when $\|\ary{v}\| \le \sqrt{n}r$, and $(2)$ is due to Banaszczyk’s tail bound, \Cref{lemma:regevclaim3.8} and $r/\sqrt{2} \ge \eta_{\epsilon}(q\Lattice)$ for $\epsilon < 2^{-n}$.

Hence, we have that $\frac{\|\ket{\Phi} - \ket{\Phi'}\|^2}{\|\ket{\Phi}\|^2} < 2^{-\Omega(n)}$, this implies that $\ket{\Phi}$ and $\ket{\Phi'}$ are $2^{-\Omega(n)}$-close to each other, by \Cref{lemma:difftotd}, as desired.

\end{proof}

\subsection{Proof of \Cref{lem:qftstate_rewrite}}\label{sec:proof_qftstate_rewrite}

\begin{proof}[Proof of \Cref{lem:qftstate_rewrite}]
    The amplitude of $\ket{\ary{y}}\ket{\ipd{\ary{s}}{\ary{a}}+u\bmod q}$ in the given state (\Cref{eqn:ql+lawithphase1}) is
    \iffull
    \[
        \begin{split}
             \sum_{\ary{v}\in q\Lattice+\Lattice \ary{a}}\rho_r(\ary{v})\rho_\sigma(u-\ipd{\ary{x}'}{\ary{v}}) \omega_R^{\ipd{\ary{v}}{\ary{y}}} & = \rho_{t}(u)\sum_{\ary{v}\in q\Lattice+\Lattice\ary{a}} \rho_{\sqrt{\mat\Sigma^{-1}}}\left(\ary{v} - \frac{r^2u}{t^2}\ary{x}'\right) \omega_R^{\ipd{\ary{v}}{\ary{y}}} \\
            & \propto \rho_{t}(u)\sum_{\ary{w}\in (q\Lattice)^*-\ary{y}/R} \rho_{\sqrt{\mat\Sigma}}(\ary{w})\exp{\left(2\pi {\rm i}\ipd{ \Lattice\ary{a}-\frac{r^2u}{t^2}\ary{x'}}{\ary{w}}\right)}\omega_R^{\ipd{\Lattice\ary{a}}{\ary{y}}},
        \end{split}
    \]
    \fi
    where $t=\sqrt{\sigma^2+r^2\|\ary{x}'\|^2} \in [\sigma, \sqrt{2}\sigma)$, matrix $\mat\Sigma = \frac{\mat I_n}{r^2} + \frac{\ary{x}'\ary{x}'^T}{\sigma^2}$ with eigenvalues $1 / r^2, (t / r\sigma)^2$, and we use the Poisson Summation Formula to compute the last equality. We define the amplitude as a function 
    \[
        \phi(\ary{y},u) := \rho_{t}(u)\sum_{\ary{w}\in (q\Lattice)^*-\ary{y}/R} \rho_{\sqrt{\mat\Sigma}}(\ary{w})\exp{\left(2\pi {\rm i}\ipd{ \Lattice\ary{a}-\frac{r^2u}{t^2}\ary{x'}}{\ary{w}}\right)}\omega_R^{\ipd{\Lattice\ary{a}}{\ary{y}}}.
    \]
    Meanwhile, for $\ary{y}/R \in B_{(q\Lattice)^*}$, we define another amplitude function $\phi'(\ary y, u)$ as
    \[
        \phi'(\ary{y},u) := \rho_{t}(u)\rho_{\sqrt{\mat\Sigma}}(\ary z)\exp{ \left(2\pi {\rm i}\ipd{ \Lattice\ary{a}-\frac{r^2u}{t^2}\ary{x'}}{-\ary z}\right)} \omega_R^{\ipd{\Lattice\ary{a}}{\ary{y}}},
    \]
    where we recall that $\ary z = \ary z(\ary{y}) = \ary y/R - \kappa_{(q\Lattice)^*}(\ary y / R)$.
    $\phi'(\ary{y},u)$ is the leading term in $\phi(\ary{y}, u)$, as we will implicitly show below. 

    \iffull
    We prove that the following (unnormalized) states
    \[
        \begin{split}
            \ket{\Phi} & :=\sum_{\ary{y}\in\Z_R^n}\sum_{u\in \Z_{qR}/R}\phi(\ary y,u)\ket{\ary{y}}\ket{\ipd{\ary{s}}{\ary{a}}+u\bmod q}\quad \text{(the given state)} \\
            \ket{\Phi'} & :=\sum_{\substack{\ary{y} \in \Z_R^n \cap R\cdot B_{(q\Lattice)^*}}}\sum_{u\in \Z_{qR}/R}\phi'(\ary y,u)\ket{\ary{y}}\ket{\ipd{\ary{s}}{\ary{a}}+u\bmod q} \quad \text{(the target state in~\Cref{eqn:ql+lawithphase2})}
        \end{split}
    \]
    are $2^{-\Omega(n)}$-close to each other.
    \fi
    To prove it, we need in addition the following (unnormalized) state
    \[
        \ket{\Phi''} :=\sum_{\substack{\ary{y} \in \Z_R^n \cap R\cdot B_{(q\Lattice)^*}}}\sum_{u\in \Z_{qR}/R}\phi(\ary y,u)\ket{\ary{y}}\ket{\ipd{\ary{s}}{\ary{a}}+u\bmod q}.
    \]

    Let's begin by establishing upper bounds for both $\|\ket{\Phi''} - \ket{\Phi'}\|^2$ and $\|\ket{\Phi} - \ket{\Phi''}\|^2$. The first term is relatively simpler to bound: for any $\epsilon<2^{-n}$, according to \Cref{thm:mul_err_tail} and the assumption that $\frac{r\sigma}{t} > \frac{r}{\sqrt{2}} > 2\sqrt 2 \eta_\epsilon(q\Lattice)$, we get that
    \begin{equation}\label{eqn:phi''-phi'}
        \begin{split}
            \|\ket{\Phi''} - \ket{\Phi'}\|^2 & = \sum_{\substack{\ary{y} \in \Z_R^n \cap R\cdot B_{(q\Lattice)^*}}}\sum_{u \in \Z_{qR}/R} \left|\phi(\ary{y}, u) - \phi'(\ary{y}, u)\right|^2 \\
            & \le \sum_{\substack{\ary{y} \in \Z_R^n \cap R\cdot B_{(q\Lattice)^*}}}\sum_{u \in \Z_{qR}/R}\left(\rho_{t}(u) \sum_{\ary{w}\in(q\Lattice)^*\setminus\{\ary{0}\}}\rho_{t/r\sigma}(\ary{w} - \ary z)\right)^2 \\
            & < \epsilon^2 \sum_{u \in \Z_{qR}/R}\rho_{t}(u)^2\sum_{\substack{\ary{y} \in \Z_R^n \cap R\cdot B_{(q\Lattice)^*}}}\rho_{t/r\sigma}(\ary z).
        \end{split}
    \end{equation}

    To establish an upper bound for the latter term $\|\ket{\Phi} - \ket{\Phi''}\|^2$, according to \Cref{thm:mul_err_tail} and the fact that $\frac{r\sigma}{t} > 2\sqrt 2 \eta_\epsilon(q\Lattice)$, we get that for $\ary{y}\in \Z_R^n \setminus R\cdot B_{(q\Lattice)^*}$,
    \[
        \begin{split}
            \sum_{\ary{w}\in (q\Lattice)^*-\ary{y}/R} \rho_{\sqrt{\mat\Sigma}}(\ary{w}) & \le \sum_{\ary{w}\in (q\Lattice)^*-\ary{y}/R} \rho_{t / r\sigma}(\ary{w}) \\
            & \le \exp\left(-\pi\frac{\dist(\ary{y}/R, (q\Lattice)^*)^2}{(t/r\sigma)^2}\right) + \epsilon \cdot \exp\left(-\pi\frac{\dist(\ary{y}/R, (q\Lattice)^*)^2}{2(t/r\sigma)^2}\right) \\
            & \le \exp\left(-\pi\frac{\dist(\ary{y}/R, (q\Lattice)^*)^2}{2(t/r\sigma)^2}\right)\left(\exp\left(-\pi\frac{(\lambda_1((q\Lattice)^*) / 2)^2}{4/r^2}\right) +\epsilon\right) \\
            & \le 2^{-n + 1}\exp\left(-\pi\frac{\dist(\ary{y}/R, (q\Lattice)^*)^2}{2(t/r\sigma)^2}\right)
        \end{split}
    \]    
    where the last inequality holds since $\lambda_1((q\Lattice)^*) \ge \sqrt{\frac{n\ln 2}{\pi}}\cdot\frac{1}{q\eta_{\epsilon}(\Lattice)} \ge \sqrt{\frac{n\ln 2}{\pi}}\cdot \frac{4}{r}$ by \Cref{lemma:smoothingRegev09}. Therefore, 
    \begin{equation}\label{eqn:phi-phi''}
        \begin{split}
            \|\ket{\Phi}-\ket{\Phi''}\|^2 & = \sum_{\substack{\ary{y} \in \Z_R^n \setminus R\cdot B_{(q\Lattice)^*}}}\sum_{u\in \Z_{qR}/R}\left|\phi(\ary{y}, u)\right|^2 \\
            & \le \sum_{u\in \Z_{qR}/R}\rho_t(u)^2 \sum_{\substack{\ary{y} \in \Z_R^n \setminus R\cdot B_{(q\Lattice)^*}}}\left(\sum_{\ary{w}\in (q\Lattice)^*-\ary{y}/R} \rho_{\sqrt{\mat\Sigma}}(\ary{w})\right)^2 \\
            & \le 2^{-2n + 2}\sum_{u\in \Z_{qR}/R}\rho_t(u)^2 \sum_{\substack{\ary{y} \in \Z_R^n \setminus R\cdot B_{(q\Lattice)^*}}} \exp\left(-\pi\frac{\dist(\ary{y}/R, (q\Lattice)^*)^2}{(t/r\sigma)^2}\right).
        \end{split}
    \end{equation}

    Combine the two upper bounds in \Cref{eqn:phi''-phi'} and \Cref{eqn:phi-phi''}, we get that
    \begin{equation}\label{eqn:phi-phi'}
        \begin{split}
            \|\ket{\Phi}-\ket{\Phi'}\|^2 & \le 2\|\ket{\Phi}-\ket{\Phi''}\|^2 + 2\|\ket{\Phi''}-\ket{\Phi'}\|^2 \\
            & \le 2^{-2n + 3}\sum_{u\in \Z_{qR}/R}\rho_t(u)^2 \sum_{\ary{y} \in \Z_R^n} \exp\left(-\pi\frac{\dist(\ary{y}/R, (q\Lattice)^*)^2}{(t/r\sigma)^2}\right).
        \end{split}
    \end{equation}

    On the other hand, we can compute that
    \begin{equation}\label{eqn:phi'}
        \begin{split}
            \|\ket{\Phi'}\|^2 & = \sum_{\substack{\ary{y} \in \Z_R^n \cap R\cdot B_{(q\Lattice)^*}}}\sum_{u \in \Z_{qR}/R}\left|\phi'(\ary{y}, u)\right|^2 \\
            & \ge \sum_{u \in \Z_{qR}/R} \rho_{t}(u)^2\sum_{\substack{\ary{y} \in \Z_R^n \cap R\cdot B_{(q\Lattice)^*}}}\exp\left(-\pi\frac{\dist(\ary{y}/R, (q\Lattice)^*)^2}{(1 / \sqrt{2}r)^2}\right).
        \end{split}
    \end{equation}
    Combining with the bounds in \Cref{eqn:phi-phi'}, we get that
    \begin{equation}\label{eqn:dist_phi_phi'}
        \frac{\|\ket{\Phi} - \ket{\Phi'}\|^2}{\|\ket{\Phi'}\|^2} \le 2^{-2n + 3}\frac{\sum_{\ary{y} \in \Z_R^n} \exp\left(-\pi\frac{\dist(\ary{y}/R, (q\Lattice)^*)^2}{(t/r\sigma)^2}\right)}{\sum_{\substack{\ary{y} \in \Z_R^n \cap R\cdot B_{(q\Lattice)^*}}}\exp\left(-\pi\frac{\dist(\ary{y}/R, (q\Lattice)^*)^2}{(1 / \sqrt{2}r)^2}\right)}.
    \end{equation}
    
    Now let's establish an upper bound on the ratio between the two summations on the right-hand side of the above inequality. We accomplish this through the following steps:
    \begin{enumerate}
        \item Expanding the support of $\ary y$. Given our assumption that $\Lattice$ is an integer lattice, we have $(q\Lattice)^* + \ary k = (q\Lattice)^*$ for any $\ary k\in \Z^n$. Consequently, $\text{dist}((\ary y - R\ary k) / R, (q\Lattice)^*) = \text{dist}(\ary y / R, (q\Lattice)^* + \ary k) = \text{dist}(\ary y / R, (q\Lattice)^*)$. Therefore, we can expand the support of $\ary y$ from $\Z_R^n$ to $\Z_{RN}^n$ for any $N\in \N_+$, which is a combination of $N^n$ hypercubes, each in the form $\Z_R^n + R\ary k$, i.e.
        \[
            \frac{\|\ket{\Phi} - \ket{\Phi'}\|^2}{\|\ket{\Phi'}\|^2} \le 2^{-2n + 3}\lim_{N\to +\infty} \frac{\sum_{\ary{y} \in \Z_{RN}^n} \exp\left(-\pi\frac{\dist(\ary{y}/R, (q\Lattice)^*)^2}{(t/r\sigma)^2}\right)}{\sum_{\substack{\ary{y} \in \Z_{RN}^n \cap R\cdot B_{(q\Lattice)^*}}}\exp\left(-\pi\frac{\dist(\ary{y}/R, (q\Lattice)^*)^2}{(1 / \sqrt{2}r)^2}\right)}.
        \]
        \item Bound the numerator and denominator as $N$ approaches infinity. Assume that $N$ is a sufficiently large integer, and denote $\ell_{\max} = \max_{\ary y\in \R^n} \text{dist}(\ary y, (q\Lattice)^*) < +\infty$. For the numerator, when considering $\ary y\in \Z_{RN}^n$, the closest vector from $\ary y / R$ to the lattice $(q\Lattice)^*$ will have $\ell_{\infty}$ norm at most $RN / 2 + \ell_{\max}$. Thus
        \[
            \sum_{\ary{y} \in \Z_{RN}^n} \exp\left(-\pi\frac{\dist(\ary{y}/R, (q\Lattice)^*)^2}{(t/r\sigma)^2}\right) \le \sum_{\substack{\ary u\in (q\Lattice)^*, \\ \|\ary u\|_{\infty} \le RN / 2 + \ell_{\max}}} \sum_{\ary y\in \Z^n} \rho_{t / r\sigma}(\ary y / R - \ary u).
        \]
        Similarly, for the denominator, when considering $\ary u\in (q\Lattice)^*$ with  $\ell_{\infty}$ norm at most $RN / 2 - \ell_{\max}$, the vectors $\ary y\in \Z^n\cap R\cdot B_{(q\Lattice)^*}$ for which the closest vector to $(q\Lattice)^*$ is $\ary u$ will certainly have an $\ell_{\infty}$ norm at most $RN / 2$. Thus
        \[
            \begin{split}
                &~ \sum_{\substack{\ary{y} \in \Z_{RN}^n \cap R\cdot B_{(q\Lattice)^*}}}\exp\left(-\pi\frac{\dist(\ary{y}/R, (q\Lattice)^*)^2}{(1 / \sqrt{2}r)^2}\right) \\
                \ge &~ \sum_{\substack{\ary u\in (q\Lattice)^*, \\ \|\ary u\|_{\infty} \le RN / 2 - \ell_{\max}}}~ \sum_{\substack{\ary{y} \in \Z^n, \\ \|\ary y / R - \ary u\| < \lambda_1((q\Lattice)^*) / 2}} \rho_{1 / \sqrt{2}r}(\ary y / R - \ary u) \\
                \ge &~ \sum_{\substack{\ary u\in (q\Lattice)^*, \\ \|\ary u\|_{\infty} \le RN / 2 - \ell_{\max}}} \left(\sum_{\ary y\in \Z^n}\rho_{1 / \sqrt{2}r}(\ary y / R - \ary{u}) - 2^{-\Omega(n)}\rho_{1 / \sqrt{2}r}(\Z^n/R)\right),
            \end{split}
        \]
        where the final inequality is based on Banaszczyk's tail bound (\Cref{lemma:Bana93}), with the guarantee that $\lambda_1((q\Lattice)^*) / 2 > \sqrt{\frac{n\ln 2}{\pi}}\cdot\frac{2}{r} > \frac{2}{\sqrt{2\pi}}\cdot \frac{1}{\sqrt{2}r}\sqrt{n}$.

        According to \Cref{lemma:smoothingMR07}, it follows that $\eta_{2^{-n}}(\Z^n / R) \le \frac{1}{R}\cdot \sqrt{\frac{2n}{\pi}} < \frac{1}{\sqrt{2}r} < \frac{t}{r\sigma}$. Therefore, by using \Cref{lemma:regevclaim3.8}, we can conclude that
        \[
            \sum_{\ary y\in \Z^n} \rho_{t / r\sigma}(\ary y / R - \ary u)  \le (1 + 2^{-\Omega(n)}) \cdot R^n\cdot (t / r\sigma)^n,
        \]
        \[
            \sum_{\ary y\in \Z^n}\rho_{1 / \sqrt{2}r}(\ary y / R - \ary{u}) - 2^{-\Omega(n)}\rho_{1 / \sqrt{2}r}(\Z^n/R)  \ge (1 - 2^{-\Omega(n)}) \cdot R^n\cdot (1 / \sqrt{2}r)^n.
        \]
        \item Establish an upper bound for the right hand side of \Cref{eqn:dist_phi_phi'}. By combining the inequalities in the previous step, we obtain that
        \iffull
        \[
            \begin{split}
                \frac{\|\ket{\Phi} - \ket{\Phi'}\|^2}{\|\ket{\Phi'}\|^2} & \le 2^{-2n + 3} (1 + 2^{-\Omega(n)}) \cdot (\sqrt{2}t / \sigma)^n \cdot \lim_{N\to +\infty} \frac{\#\{\ary u\in (q\Lattice)^*: \|\ary u\|_{\infty} \le RN / 2 + \ell_{\max}\}}{\#\{\ary u\in (q\Lattice)^*: \|\ary u\|_{\infty} \le RN / 2 - \ell_{\max}\}} \\
                & \le 2^{-n + 3} (1 + 2^{-\Omega(n)}) \cdot \lim_{N\to +\infty} \frac{\#\{\ary u\in (q\Lattice)^*: \|\ary u\|_{\infty} \le RN / 2 + \ell_{\max}\}}{\#\{\ary u\in (q\Lattice)^*: \|\ary u\|_{\infty} \le RN / 2 - \ell_{\max}\}} \\
                & = 2^{-n + 3} (1 + 2^{-\Omega(n)}) \cdot \lim_{N\to +\infty} \frac{(RN / 2 + \ell_{\max})^n}{(RN / 2 - \ell_{\max})^n} \\
                & = 2^{-n + 3} (1 + 2^{-\Omega(n)}).
            \end{split}
        \]
        \fi
    \end{enumerate}

    Consequently, we can infer that the provided state $\ket{\Phi}$ is $2^{-\Omega(n)}$-close to $\ket{\Phi'}$, according to \Cref{lemma:difftotd}, as desired.
\end{proof}

\end{document}